\documentclass[11pt]{article}
\usepackage[utf8]{inputenc}
\usepackage[margin=1in]{geometry}

\usepackage{amsmath, amssymb, amsthm, dsfont, mathtools, physics}
\usepackage{aliascnt}

\usepackage[british,UKenglish]{babel}
\addto\captionsbritish{}
\addto\captionsUKenglish{}

\usepackage{tocloft}

\usepackage{algorithm}
\usepackage{algorithmic}
\usepackage{bm}
\usepackage{xcolor}
\usepackage{hyperref}
\usepackage[noabbrev, nameinlink, capitalise]{cleveref}
\usepackage{graphicx}
\usepackage{authblk}
\usepackage{pgfplots}
\pgfplotsset{compat=1.18}

\usepackage[numbers,sort&compress]{natbib}

\makeatletter
\renewcommand{\paragraph}{\@startsection{paragraph}{4}{\z@}{1.5ex}{-1em}{\normalfont\normalsize\bfseries}}
\makeatother

\newtheoremstyle{myplain}{5pt}{5pt}{\itshape}{0pt}{\bfseries}{}{5pt plus 1pt minus 1pt}{}
\newtheoremstyle{mydefinition}{5pt}{5pt}{\normalfont}{0pt}{\bfseries}{}{5pt plus 1pt minus 1pt}{}

\theoremstyle{remark}

\theoremstyle{myplain}

\newtheorem*{thm*}{Theorem}
\newaliascnt{theorem}{thm}
\newtheorem{theorem}[theorem]{Theorem}
\aliascntresetthe{theorem}
\newaliascnt{lemma}{thm}
\newtheorem{lemma}[lemma]{Lemma}
\aliascntresetthe{lemma}
\newaliascnt{corollary}{thm}
\newtheorem{corollary}[corollary]{Corollary}
\aliascntresetthe{corollary}
\newaliascnt{proposition}{thm}
\newtheorem{proposition}[proposition]{Proposition}
\aliascntresetthe{proposition}
\newaliascnt{conjecture}{thm}

\aliascntresetthe{conjecture}
\newtheorem*{prop*}{Proposition}

\crefname{theorem}{theorem}{theorems}
\Crefname{theorem}{Theorem}{Theorems}
\crefname{lemma}{lemma}{lemmas}
\Crefname{lemma}{Lemma}{Lemmas}
\crefname{corollary}{corollary}{corollaries}
\Crefname{corollary}{Corollary}{Corollaries}
\crefname{proposition}{proposition}{propositions}
\Crefname{proposition}{Proposition}{Propositions}
\crefname{conjecture}{conjecture}{conjectures}
\Crefname{conjecture}{Conjecture}{Conjectures}

\theoremstyle{mydefinition}
\newaliascnt{definition}{thm}

\aliascntresetthe{definition}
\crefname{definition}{definition}{definitions}
\Crefname{definition}{Definition}{Definitions}

\theoremstyle{remark}
\newtheorem{ex}{Example}[section]
\newtheorem{rmk}[ex]{Remark}
\newtheorem*{rmk*}{Remark}

\DeclareMathOperator{\N}{\mathbb{N}}

\DeclareMathOperator{\cA}{\mathcal{A}}

\DeclareMathOperator{\cD}{\mathcal{D}}

\DeclareMathOperator{\cH}{\mathcal{H}}
\DeclareMathOperator{\cI}{\mathcal{I}}

\DeclareMathOperator{\cL}{\mathcal{L}}
\DeclareMathOperator{\cK}{\mathcal{K}}

\DeclareMathOperator{\cO}{\mathcal{O}}
\DeclareMathOperator{\cP}{\mathcal{P}}

\newtheorem{assumption}[theorem]{Assumption}
\DeclareMathOperator{\dist}{dist}
\DeclareMathOperator{\supp}{supp}
\DeclareMathOperator{\diam}{diam}

\DeclareMathOperator{\wt}{wt}

\providecommand{\Tr}[1][]{\operatorname{Tr}\ifx#1\empty \else \left[ #1 \right] \fi}
\providecommand{\tr}[1][]{\operatorname{tr}\ifx#1\empty \else \left[ #1 \right] \fi}

\providecommand{\inorm}[1]{{\left\vert\kern-0.25ex\left\vert\kern-0.25ex\left\vert #1
        \right\vert\kern-0.25ex\right\vert\kern-0.25ex\right\vert}}

\definecolor{TealDark}{HTML}{008080}
\definecolor{TealMid}{HTML}{66B2B2}
\definecolor{TealLight}{HTML}{CCE5E5}
\definecolor{CoralDark}{HTML}{FF6F61}
\definecolor{CoralMid}{HTML}{FFB8AE}
\definecolor{CoralLight}{HTML}{FFE4E1}
\definecolor{GoldenrodDark}{HTML}{DAA520}
\definecolor{GoldenrodMid}{HTML}{F4D35E}
\definecolor{GoldenrodLight}{HTML}{FFF3D4}
\definecolor{SlateBlueDark}{HTML}{6A5ACD}
\definecolor{SlateBlueMid}{HTML}{A29BDF}
\definecolor{SlateBlueLight}{HTML}{DCD6F7}

\hypersetup{
    colorlinks,
    linkcolor={TealDark},
    citecolor={SlateBlueDark},
    urlcolor={SlateBlueDark}
}

\usepackage{capt-of}
\usepackage{csquotes}
\usepackage{url}

\let\epsilon\varepsilon

\title{\texorpdfstring{Robust Structure Learning of $k$-local Lindbladians}{Robust Structure Learning of k-local Lindbladians}}
\date{}
\newsavebox{\authoremailbox}
\AtBeginDocument{\sbox{\authoremailbox}{\footnotesize
    \textsuperscript{*}moebustim@gmail.com\quad
    \textsuperscript{\textdagger}thiagob@berkeley.edu\quad
    \textsuperscript{\textdaggerdbl}dsfranca@math.ku.dk\quad
    \textsuperscript{\textsection}cambyse.rouze@inria.fr}}
\let\standardfootnoterule\footnoterule
\renewcommand{\footnoterule}{\ifnum\value{page}=1
        \kern-3pt
        \hbox to \columnwidth{\hfil\rule{\dimexpr\wd\authoremailbox+2em\relax}{0.4pt}\hfil}\kern2.6pt
    \else
        \standardfootnoterule
    \fi}
\makeatletter
\newcommand{\printauthoremails}{\begingroup
        \renewcommand{\@makefntext}[1]{\noindent\makebox[\columnwidth][c]{##1}}\footnotetext[0]{\usebox{\authoremailbox}}\endgroup}
\makeatother
\author[1,2]{Tim M\"{o}bus\textsuperscript{*}}
\author[3]{Thiago Bergamaschi \textsuperscript{\textdagger}}
\author[4]{Daniel Stilck Fran\c{c}a \textsuperscript{\textdaggerdbl}}
\author[5]{Cambyse Rouz\'e\textsuperscript{\textsection}}

\affil[1]{\small Department of Applied Mathematics and Theoretical Physics, University of Cambridge, United Kingdom}
\affil[2]{\small Department of Mathematics, University of T\"ubingen, Germany}
\affil[3]{\small Department of EECS, UC Berkeley, USA}
\affil[4]{\small Department of Mathematical Sciences, University of Copenhagen, Denmark}
\affil[5]{\small Inria, T\'el\'ecom Paris - LTCI, Institut Polytechnique de Paris, France}

\begin{document}

\maketitle
\printauthoremails

\vspace*{-6ex}

\begin{abstract}
    We present an efficient protocol for learning an unknown $k$-local Lindblad generator on $n$ qubits using only product-state preparations, short-time evolution, and single-qubit Pauli measurements, without prior knowledge of the interaction structure. For fixed $k$ and a supplied weighted interaction-strength bound $\bar\alpha$, the protocol estimates all Hamiltonian and dissipative Pauli--GKSL coefficients with total coefficient error at each site at most $\varepsilon$ with probability at least $1-\delta$ using $\widetilde{\mathcal O}_k(\bar\alpha^2n^{2k-2}\varepsilon^{-2}\log(n/\delta))$ samples and polylogarithmically many distinct evolution times. This guarantee requires no degree, sparsity, or tail promise. A convex optimization converts these estimates into a valid $k$-local Lindblad generator with diamond-norm error at most $\varepsilon$ using $\widetilde{\mathcal O}_k(\bar\alpha^2n^{2k}\varepsilon^{-2}\log(n/\delta))$ samples and polynomial-time classical postprocessing. With a supplied local effective-sparsity bound $r$ and a sufficiently small sitewise subthreshold tail $\mathcal O_k(\varepsilon)$, coefficient recovery requires only $\widetilde{\mathcal O}_k(\bar\alpha^2r^2\varepsilon^{-2}\log(n/\delta))$ samples, without a supplied threshold or coefficient locations and without a coefficient-gap assumption. In particular, exactly sparse bounded-degree models have logarithmic sample dependence on $n$. We also provide guard-band support recovery and complementary supplied-candidate guarantees, extend the guarantees to model misspecification, and prove complementary sample-complexity lower bounds.
\end{abstract}

\tableofcontents

\vspace*{\fill}

\section{Introduction}
    Characterizing the dynamics of open quantum systems is a central task in quantum information processing. Realizing advanced applications on modern quantum platforms requires robust validation and benchmarking protocols that scale efficiently alongside increasing qubit numbers. Furthermore, because physical hardware is inevitably coupled to its environment, precisely characterizing these non-unitary dynamics is critical for tailoring quantum error correction, mitigation, and fault-tolerant strategies (all while minimizing experimental and computational overhead). In the Markovian regime, the evolution is described by a quantum dynamical semigroup $e^{t\cL}$, whose generator $\cL$ admits the Gorini--Kossakowski--Sudarshan--Lindblad (GKSL) form \cite{GKS1976,Lindblad1976}. Learning such generators is crucial for applications in noise diagnosis~\cite{eisert2020quantum}, verification of analog quantum simulators~\cite{BaireyAradLindner2019}, error mitigation~\cite{endo2021hybrid}, and the design of dissipative state-preparation protocols~\cite{verstraete2009quantum}.

    A general $n$-qubit generator has exponentially many parameters, and full quantum process tomography is therefore infeasible except for very small systems \cite{NielsenChuang2000,ChuangNielsen1997,Poyatos1997}. This is the curse of dimensionality. However, many physically relevant open-system mechanisms, such as dephasing, amplitude damping, or local hopping, are highly localized across the physical lattice~\cite{breuer2002theory}. This motivates the assumption that both the Hamiltonian and the dissipative generators act nontrivially only on subsystems of size at most $k$ qubits, where $k$ is fixed and does not scale with the total number of qubits $n$. In this case the number of relevant coefficients is only polynomial in $n$. While recent breakthroughs have yielded highly efficient routines for learning generators of closed, unitary dynamics~\cite{BakshiLiuMoitraTang2024,StilckFranca.2024,HuangTongFangSu2023,ZubidaYitzhakiLindnerBairey2021,BaireyAradLindner2019,Caro2024PTM,Flynn2022,Gu2024,Holzpfel2015,Moebus.2023,Li2024,Moebus2025Heisenberg,daSilva2011,Wiebe2014}, analogous guarantees for dissipative dynamics remain more restricted. Known results apply, for example, to on-site dissipation \cite{StilckFranca.2025}, have sample complexity depending on a design-matrix conditioning factor \cite{IvashkovEtAl2026}, or fix a white-box ansatz with known Hamiltonian structure and fixed jump-operator form, learning only its couplings and dissipative rates \cite{Heightman2026}. This raises the following question:
    \begin{center}
        \textit{ Can one efficiently and stably reconstruct a physically valid $k$-local Lindbladian from logarithmically short-time experimental data, without prior knowledge of its interaction structure?}
    \end{center}

    As mentioned, previous approaches to this problem frequently relied on prespecified interaction graphs, restrictive ansatz choices, or conditioning assumptions for induced linear systems, which significantly limits their applicability to unknown noise channels in realistic hardware~\cite{BaireyAradLindner2019}. We instead recover the Hamiltonian and Kossakowski coefficients with sitewise error control and fit them over the physical GKSL cone. Here, sitewise error is the sum of absolute errors over coefficients whose support contains a given site. We count Hamiltonian errors twice to account for the two multiplication terms in each commutator. This produces local jump operators defining a valid quantum Markov semigroup and controls the generator error in diamond norm. Neither coefficient recovery nor physical reconstruction requires a prescribed ansatz, prior knowledge of the interaction graph, or a global conditioning assumption on an induced linear system. Furthermore, we achieve this by using a highly restricted experimental suite: the protocol relies solely on initializing random product states in the Pauli basis, allowing the system to undergo brief periods of unmitigated evolution, and performing localized Pauli measurements.

    The proof combines a support-growth bound for short-time Heisenberg evolution with the Fierz identity and local coefficient inversion.  The first ingredient gives a logarithmic-degree Taylor approximation from $k$-locality and bounded weighted interaction strength, without a geometric interaction graph.  The Fierz transform converts the resulting Pauli transfer matrix (PTM) derivatives into local coefficient sums, which are thresholded before inversion.  Geometric Lieb--Robinson bounds are needed only in the later local-observable application. The Fierz identity originates from high-energy field theory for transforming spinor bilinears \cite{fierz1937relativistische}. To our knowledge, the Fierz identity has never been utilized within Hamiltonian or Lindbladian learning theory.

    \subsection{Main result}
        Our main result is an efficient algorithm for learning Markovian quantum dynamics generated by a local Lindbladian $\cL$ in the Schr\"{o}dinger picture of the following form: let $\mathsf G=([n],E)$ denote the unknown interaction hypergraph encoding the interactions between the $n$ qubits. The local generator is
        \begin{equation*}
            \cL=\sum_{e\in E}\cL_e\,,
        \end{equation*}
        where each local term $\cL_e:\mathbb{M}_{2^n}\rightarrow \mathbb{M}_{2^n}$, including its Hamiltonian and dissipative parts, acts nontrivially on qubits within a region $e\subseteq [n]$ of cardinality at most $k=\mathcal{O}(1)$. We represent the same superoperator in the Pauli-superoperator basis as $\cL=\sum_{\mathbf P,\mathbf Q}\chi_{\mathbf P,\mathbf Q}\,\mathbf P\bullet\mathbf Q$, where $(\mathbf P\bullet\mathbf Q)(X)=\mathbf P X\mathbf Q$ for $n$-qubit Pauli string pairs $(\mathbf P, \mathbf Q)\in \cI_k:=\left\{(\mathbf P,\mathbf Q)\in\mathcal P_n\times\mathcal P_n:|\operatorname{supp}(\mathbf P)\cup\operatorname{supp}(\mathbf Q)|\le k\right\}$ and $\cP_n:= \{I,X,Y,Z\}^{\otimes n}$. We write $\mathcal I_k^\circ:=\mathcal I_k\setminus\{(I,I)\}$ for the nonidentity coefficient indices. The coefficient array $\chi$ is the Pauli-basis $\chi$-matrix. Moreover, we assume a constant bound on a weighted version of the maximum number of local terms touching a given site: denoting by $\|\cK\|_{2\to 2}$ the induced Hilbert--Schmidt norm of a superoperator $\cK:\mathbb{M}_{2^n}\rightarrow \mathbb{M}_{2^n}$, we assume
        \begin{equation}\label{eq:weightedintersectiondegree}
            \alpha(\{\cL_e\}):=\max_{u\in[n]}\sum_{\substack{e\ni u}}\|\cL_e^\dagger\|_{2\to2}
        \end{equation}
        or simply $\alpha$ when clear from context. We assume that $\cL$ admits a $k$-local decomposition with $\alpha(\{\cL_e\})\le\bar\alpha=\mathcal O(1)$, where $\bar\alpha>0$ is supplied to the learner. Neither the decomposition nor its supports are supplied. All experiments use times $t\in[0,(4\bar\alpha k)^{-1}]$. Note that the above normalization can always be enforced by rescaling the time unit. For a decomposition whose interaction hypergraph has maximum degree $\Delta$ and $J:=\max_{e\in E}\|\cL_e^\dagger\|_{2\to2}$, one has $\alpha(\{\cL_e\})\le\Delta J$. We refer to $\alpha(\{\cL_e\})$ as the weighted interaction strength of this decomposition. Our protocols only use simple input states and measurements: each experiment prepares a product of single-qubit Pauli eigenstates, evolves it for a sampled time $t\in[0,(4\bar\alpha k)^{-1}]$, and then measures every qubit independently in a single-qubit Pauli basis. The supplied bound $\bar\alpha$ contributes a factor $\bar\alpha^2$ to the sample bounds and $\bar\alpha$ to the total-evolution-time bounds, as explained in Remark~\ref{rmk:strength-total-time}. Throughout, one sample means one preparation--evolution--measurement experiment, equivalently one channel use. The stated classical runtime covers only coefficient recovery from estimated PTM entries, not processing measurements or fitting their time dependence. Total physical evolution time is the sum of evolution durations over all experiments and is reported separately.

        \paragraph{Coefficient and diamond-norm learning}
        The protocol first obtains initial estimates of the generator coefficients and then projects them to a valid $k$-local Lindblad generator.

        The clipped local sparsity $r_\tau$ measures the effective number of significant coefficients (those above the threshold $\tau$) incident on any site at scale $\tau$, as defined in~\eqref{eq:clipped-local-sparsity}. The sitewise subthreshold tail $\beta_\tau$ measures the sum of absolute values of nonsignificant coefficients (those at or below $\tau$) incident on any site, maximized over sites, as defined in~\eqref{eq:sitewise-chi-tail}.

        \begin{theorem}[Efficient learning of local Lindbladians (informal,
            see Thm.~\ref{thm:approximate-degree-learning}, \ref{thm:full-lindblad-learning}, and 
            Cor.~\ref{cor:generic-local-norm-recovery})]
            \label{thm:informal-main}
            Fix $k$ and use the supplied bound $\bar\alpha$ from the model above. Our algorithm returns estimates of every independent $k$-local coefficient with sitewise error at most $\varepsilon$ and probability at least $1-\delta$, using
            \begin{equation*}
                \widetilde{\mathcal O}_k\!\left(
                \bar\alpha^2n^{2k-2}\varepsilon^{-2}
                \log\frac n\delta\right)
                \quad\text{samples}.
            \end{equation*}
            In particular, every Hamiltonian and dissipative coefficient has entrywise error at most $\varepsilon$. If a supplied number $r$ bounds the clipped local sparsity $r_\tau$ at some unknown $\tau>0$, and the corresponding sitewise subthreshold mass $\beta_\tau$ is at most $\varepsilon/(2\cdot4^k)$, the sample bound improves to
            \begin{equation*}
                \widetilde{\mathcal O}_k\!\left(
                \bar\alpha^2r^2\varepsilon^{-2}\log\frac n\delta
                \right)
                \quad\text{samples},
            \end{equation*}
            with total evolution time $\widetilde{\mathcal O}_k(\bar\alpha r^2 \varepsilon^{-2}\log(n/\delta))$. The initial estimate has $\mathcal O_k(r)$ nonzero coefficients incident on each site, and thresholding it gives guard-band support recovery. Neither the threshold nor coefficient locations are supplied, and no minimum absolute value is assumed for nonzero coefficients. The protocol uses polylogarithmically many times in $[0,(4\bar\alpha k)^{-1}]$.

            An SDP constructs a valid $k$-local Lindblad generator $\widehat\cL$ satisfying $\|\widehat\cL-\cL\|_\diamond\le\varepsilon_\diamond$ from
            \begin{equation*}
                \widetilde{\mathcal O}_k\!\left(
                \bar\alpha^2n^{2k}\varepsilon_\diamond^{-2}
                \log\frac n\delta
                \right)
            \end{equation*}
            samples. If $r_\tau\le r_\diamond$ and $\beta_\tau\le\varepsilon_\diamond/(16n4^k)$ for some $\tau>0$, the factor $n^{2k}$ can be replaced by $r_\diamond^2n^2$. All classical postprocessing is polynomial-time for fixed $k$.
        \end{theorem}
        \noindent The weighted interaction strength assumption includes both short-range geometrically local systems and dense long-range systems with mean-field-scaled interactions, because it controls the total interaction strength incident on each site. Appendix~\ref{app:rank-selection-conditioning} shows that full rank alone does not ensure stable inversion: even for a $4$-local target with bounded weighted interaction strength and support-intersection degree, an admissible fixed-patch full-rank recovery matrix can have inverse norm exponential in the candidate-support size.

        \paragraph{Complementary supplied-candidate and guarded regimes}
        Two variants exploit additional information under the same locality and strength assumptions. A supplied candidate set $\Omega$ contains the indices of all $\chi$-coefficients with absolute value above a threshold $\tau>0$, and may also contain indices of smaller or zero coefficients. Let $D_\Omega$ bound error amplification during inversion and $\rho_\Omega$ bound the contribution of omitted diagonal extensions, as defined in Section~\ref{subsec:stable-sparse-diagonal-extensions}. Diagonal extensions add the same Pauli factors on additional sites to both coefficient indices and contribute to the local sum used to recover the original coefficient.

        All supplied-candidate resource bounds below apply to nonempty candidate sets. For an empty set, the protocol returns zero coefficients without measurements.

        If $\tau+D_\Omega\rho_\Omega\le\varepsilon_\chi$, Corollary~\ref{cor:dissipative-learning-stable-sparse} gives entrywise error $\mathcal O(\varepsilon_\chi)$ with success probability at least $1-\delta$ using $\widetilde{\mathcal O}_k(\bar\alpha^2D_\Omega^2\varepsilon_\chi^{-2}\log(|\Omega|/\delta))$ samples and $\mathcal O_k(|\Omega|)$ coefficient-recovery operations.

        Without supplied locations, fix a decision threshold $\lambda>0$ and a margin $0<\gamma<\lambda$. Write $\Omega_a$ for the indices of $\chi$-coefficients with absolute value greater than $a$. If a supplied bound $B\ge D_{\Omega_{\lambda-\gamma}}$ satisfies $B\rho_{\Omega_{\lambda+\gamma}}\le\gamma/2$, Corollary~\ref{cor:threshold-structure-learning} gives
        $\Omega_{\lambda+\gamma}\subseteq\widehat\Omega_\lambda\subseteq\Omega_{\lambda-\gamma}$ with success probability at least $1-\delta$,
        using $\widetilde{\mathcal O}_k(\bar\alpha^2B^2\gamma^{-2}\log(n/\delta))$ samples and $\mathcal O_k(n^k)$ coefficient-recovery operations. Thus $\widehat\Omega_\lambda$ is learned from the data, not supplied: it contains every coefficient index above $\lambda+\gamma$ and excludes those at or below $\lambda-\gamma$, with no guarantee inside the intervening band. The tail condition leaves at least half the error budget for measurement errors. Without the tail promise, Corollary~\ref{cor:unconditional-thresholding} gives the same support guarantee at the generic sample cost.

        Corollaries~\ref{cor:bounded-intersection-exact-support} and~\ref{cor:algebraic-diagonal-extension-tails} give logarithmic sample dependence on $n$ for supplied candidate sets when each interaction region intersects a bounded number of others or the ordered absolute values of extension coefficients decay as $j^{-p}$ with $p>k+1$, respectively. For unknown supports with bounded intersection degree, Corollary~\ref{cor:gapped-unknown-bounded-intersection} gives exact support recovery when every nonzero coefficient has absolute value above the decision threshold by the required margin.

        \paragraph{Sample-complexity lower bounds}
        Section~\ref{sec:lower-bound-product-measurements} gives lower bounds for the product-state, short-time, product-Pauli access model. The diamond lower bound uses an unrestricted-strength packing.

        \begin{theorem}[Lower bounds for Lindbladian learning (informal, see Thms.~\ref{thm:single-chi-lower-bound}, \ref{thm:diamond-lower-bound})]
            \label{thm:diamond-lower-bound-product-measurementsintro}
            There exist $k$-local Hamiltonian subfamilies for which every adaptive protocol using product input states, times $t\le t_{\max}$, and product Pauli measurements requires
            \begin{align*}
                &\Omega_k\left(
                \frac{1}{t_{\max}^2\varepsilon_1^2}\right)&\text{ samples for learning a single entry $\chi_{\mathbf{P},\mathbf{Q}}$ to $\varepsilon_1$-precision, and }& \\
                &\Omega_k\left(
                \frac{n^k}{t_{\max}^2\varepsilon_\diamond^2}\right)& \text{ samples for $\varepsilon_\diamond$-diamond norm recovery }\text{in the unrestricted-strength packing}.&
            \end{align*}

        \end{theorem}

    \subsection{Technical overview}
        The proof of Theorem \ref{thm:informal-main} has three components.

        \paragraph{Estimating PTM entries from logarithmically scaled short-time data.} For Pauli strings $\mathbf P,\mathbf Q\in\mathcal{P}_n$, the relevant Pauli-transfer-matrix (PTM) entry is the derivative
        \begin{equation*}
            L_{\mathbf P,\mathbf Q}=\frac{d}{dt}\biggr|_{t=0}2^{-n}\Tr(\mathbf P e^{t\cL}(\mathbf Q))\,.
        \end{equation*}
        For a target precision $\varepsilon_{L}$, we estimate all the time-dependent overlaps
        \begin{equation*}
            2^{-n}\Tr\!\left[\mathbf P e^{t\cL}(\mathbf Q)\right]
        \end{equation*}
        for constant-weight Pauli strings $\mathbf{P},\mathbf{Q}$ at $m=\mathcal O(\operatorname{polylog}(\varepsilon_{L}^{-1}))$ times in the interval $[0,(4\bar\alpha k)^{-1}]$, with spacing on the order of $\cO(\operatorname{polylog}(\varepsilon_{L}^{-1})^{-1})$. The overlap estimates are obtained using process-shadow measurements, and the derivative at zero is recovered by robust polynomial interpolation. The analysis relies on a Taylor approximation theorem for Heisenberg evolution under a $k$-local Lindbladian with a supplied upper bound $\bar\alpha$ on the weighted interaction strength, which we believe of independent interest:
        \begin{lemma}[Polynomial Approximations to Heisenberg Evolution (informal, see Lem.~\ref{lem:poly-approx})]\label{lem:poly-approx-intro}
            Given $q\ge 1$, let $Q$ be an operator supported on at most $q$ sites. Then, for any error $\varepsilon \in (0,1)$ and time $t\in[0,(4\bar\alpha k)^{-1}]$, where $\bar\alpha$ is the supplied upper bound above, there exists an operator-valued polynomial $Q^{(d)}$ of degree $d =\cO(\lceil q/k\rceil+\log(1/\varepsilon))$ in $t$ such that:
            \begin{equation}
                \|e^{t\mathcal{L}^\dagger}[Q] - Q^{(d)}(t)\|_{2} \le \varepsilon \cdot \|Q\|_{2}\,.
            \end{equation}
        \end{lemma}
        For the local PTM entries used below, $q\le k$, so polynomial degree $d=\mathcal O(1+\log(1/\varepsilon))$ suffices.

        \paragraph{Local inversion from PTM data to Pauli--GKSL coefficients.} In a first step, we establish an inversion algorithm via the Fierz identity (our second key ingredient for the present result), that maps the PTM coefficients $L_{\mathbf{P},\mathbf{Q}}$ to the Pauli--GKSL coefficients $\chi_{\mathbf{P},\mathbf{Q}}$ in a surprisingly clean manner. However, this would involve exponentially many entries. Therefore, we further refine the inversion to constant-size operations by exploiting locality. For each region $R\subseteq[n]$ with $|R|\le k$, the reduced generator
        \begin{equation}\label{eq:locally_reduced_generator}
            \cL^{(R)}:=\Tr_{\overline R}\left[\cL\!\left(\bullet\otimes\frac{I_{\overline R}}{2^{|\overline R|}}\right)\right]
        \end{equation}
        determines local $\chi$-coefficients. Moreover, the PTM coefficients of $\cL^{(R)}$ coincide with the coefficients $L_{\mathbf{P},\mathbf{Q}}$ of $\cL$ with $\mathbf{P},\mathbf{Q}$ acting nontrivially only on $R$, and can then be estimated to precision $\varepsilon_{L}$ with $\mathcal{O}(\varepsilon_{L}^{-2}\log(n))$ samples using the process shadow tomography protocol described above. A local coefficient inversion procedure then recovers all nonidentity entries indexed by $\mathcal I_k^\circ$. $\chi_{I,I}$ is determined afterwards by trace preservation:

        \begin{lemma}[Local inversion (informal, see Lem.~\ref{lem:chi-recovery-stability})]\label{lem:chi-recovery-stabilityintro}
            Given an estimate $\widehat{L}_{\mathbf{P},\mathbf{Q}}$ of all PTM entries $L_{\mathbf{P},\mathbf{Q}}$ for all Pauli strings with $|\supp(\mathbf P)\cup\supp(\mathbf Q)|\le k$ and error $\varepsilon_L$, there is an $\mathcal{O}(n^k)$-time algorithm that outputs estimates $\widehat{\chi}_{\mathbf{P},\mathbf{Q}}$ of all independent nonidentity coefficients indexed by $\mathcal I_k^\circ$ with error
            \begin{equation*}
                \max_{(\mathbf P,\mathbf Q)\in\mathcal I_k^\circ}\left|\widehat\chi_{\mathbf P,\mathbf Q}-\chi_{\mathbf P,\mathbf Q}\right|=\mathcal O_k(n^{k-1})\varepsilon_L.
            \end{equation*}
            The coefficient $\chi_{I,I}$ is fixed by trace preservation and is not an independent parameter needed to reconstruct $G$ and $h$.
        \end{lemma}

        \paragraph{Threshold-before-inversion and local tails.} The local coefficient sums are an invertible linear transformation of the generator coefficients (Lem.~\ref{lem:local-inclusion-exclusion}). This transformation and its inverse have norm at most $2^k$ in the local coefficient norm defined in~\eqref{eq:local-coefficient-norm-degree}. Let $\nu$ bound the absolute estimation error in each local coefficient sum. Thresholding the estimated sums at $2\nu$ before inversion gives
        \begin{equation*}
            \|\widehat\chi-\chi\|_{\mathrm{loc},1}
            \le3\cdot4^k r_\tau\nu+4^k\beta_\tau,
        \end{equation*}
        where $r_\tau$ is the clipped local sparsity and $\beta_\tau$ is the sitewise mass below threshold. This is the key step behind the support-free theorem.

        \paragraph{Complementary guarded $\chi$-support recovery.} For guarded threshold-support recovery in the $\chi$-coordinates, or when a candidate set $\Omega$ is supplied, we also use a guarded recursive inversion. Denoting by $\chi^{(S)}$ the matrix of $\chi$-coefficients of $\cL^{(S)}$ (cf.~Eq.~\eqref{eq:locally_reduced_generator}),
        \begin{equation*}
            \chi^{(S)}_{\mathbf P_S,\mathbf Q_S}=\sum_{\mathbf A_{\overline S}}\chi_{\mathbf P_S\otimes \mathbf A_{\overline S},\mathbf Q_S\otimes \mathbf A_{\overline S}},
        \end{equation*}
        i.e.~a local coefficient is equal to the sum of all its diagonal extensions. In the worst case this produces an $\mathcal O_k(n^{k-1})$ recursive amplification, attained at support size one. Our sparse analysis identifies the relevant obstruction more precisely: accepted diagonal extensions must be counted explicitly, while omitted extensions contribute a bias. For structure learning this requires a guard band. At decision threshold $\lambda$ and margin $\gamma$, accepted extensions are counted using the lower threshold $\lambda-\gamma$, while omitted extensions are charged to the tail below the upper threshold $\lambda+\gamma$:

        \begin{lemma}[Guarded sparse inversion and error propagation (informal, see Lem.~\ref{lem:sparse-diagonal-inversion})]
            Let $D_-:=D_{\lambda-\gamma}$ and $\rho_+:=\rho_{\lambda+\gamma}$.  Assume access to $\varepsilon_L$-accurate PTM estimates for all Pauli strings of weight at most $k$, and suppose
            \begin{equation*}
                D_-(\varepsilon_L+\rho_+)\le\gamma.
            \end{equation*}
            Then thresholding the recursively corrected local coefficients at level $\lambda$ returns a support $\widehat\Omega_\lambda$ such that
            \begin{equation*}
                \{|\chi|>\lambda+\gamma\}\subseteq
                \widehat\Omega_\lambda
                \subseteq
                \{|\chi|>\lambda-\gamma\},
            \end{equation*}
            and every accepted coefficient estimate has error at most $\gamma$.
        \end{lemma}

        \paragraph{Support search versus parameter inversion.} We distinguish the sparse parameter-learning problem from structure learning. Given a candidate set $\Omega$, the recursive inversion can be restricted to $\Omega$, leading to sample complexity depending on $D_\Omega$ and $|\Omega|$, and classical postprocessing time $\mathcal O_k(|\Omega|)$. If $\Omega$ is not supplied, we obtain it by thresholding the recursively corrected local coefficients with a guard band. Suppose the learner is supplied with a bound $B\ge D_{\lambda-\gamma}$ and that $B\rho_{\lambda+\gamma}\le\gamma/2$. Choosing $\varepsilon_L\le\gamma/(2B)$ gives the support sandwich with
        \begin{equation*}
            \widetilde{\mathcal O}_k\!\left(
            \bar\alpha^2B^2\gamma^{-2}\log(n/\delta)
            \right)
        \end{equation*}
        samples and an exhaustive $\mathcal O_k(n^k)$ classical support-search cost. This is the cost of checking all $k$-local candidate coefficient indices.

        \paragraph{Projection onto valid Lindblad generators.} A noisy coefficient estimate need not define a completely positive semigroup.  We therefore minimize its local coefficient $1$-norm distance over positive local Kossakowski blocks.  The resulting semidefinite program is polynomial-size for fixed $k$, and spectral decomposition of the fitted blocks yields local jump operators.  Since generator diamond distance is at most $2n$ times the local coefficient $1$-norm distance, this projection gives the valid-GKSL guarantee in Theorem~\ref{thm:informal-main} with sample dependence $n^{2k}$.

        \paragraph{Model-misspecified Lindbladian learning.} We also consider the model-misspecified setting, where the true generator need not itself be $k$-local or satisfy the structural assumptions imposed by our learning model. In this case the benchmark is a best-fitting generator $\mathcal L_k^*$ in the class of $k$-local Lindbladians with bounded weighted interaction strength. The generic entrywise and supplied-candidate pipelines remain stable: their coefficient estimates recover its coefficients up to the usual statistical error plus an agnostic bias proportional, up to polynomial factors in the system size and polylogarithmic factors in the target precision, to the distance between the true generator and $\mathcal L_k^*$. Under the sparse diagonal-extension assumptions, this bias and the sample complexity improve in the same way as in the realizable sparse theorem. 

        \paragraph{Lower bounds in the experimental access model.} Finally, we complement the algorithms with lower bounds for the same restricted access model. A two-point test gives an $\Omega(t_{\max}^{-2}\varepsilon^{-2})$ single-coefficient lower bound. For unrestricted interaction strength, a packing over $\binom nk$ commuting $k$-local Hamiltonian directions gives an $\Omega_k(n^kt_{\max}^{-2}\varepsilon_\diamond^{-2})$ lower bound for diamond recovery. Thus, quadratic precision dependence is unavoidable, and the $n^k$ system-size dependence is necessary in the unrestricted-strength setting.

    \subsection{Comparison with related work.}

        \paragraph{Full quantum process tomography}
        Full quantum process tomography learns an arbitrary quantum channel and is exponential in $n$ without structural assumptions \cite{NielsenChuang2000,ChuangNielsen1997,Poyatos1997}. Projected least-squares and related convex-projection methods enforce complete positivity for channel tomography, but they do not exploit the locality structure of a Lindblad generator and therefore do not yield polynomial-time many-body guarantees in the setting considered here \cite{SurawyStepney2022PLS}.

        \paragraph{Process Shadow tomography}
        Our measurement primitive is closest to classical-shadow and process-shadow tomography \cite{HuangKuengPreskill2020,Kunjummen2023ShadowProcess,StilckFranca.2024,kraft2025boundederrorquantumsimulationhamiltonian}. Those results give sample-efficient estimators for many Pauli observables or channel overlaps. However, estimating overlaps is not the same as learning a generator. We additionally solve the inverse problem from PTM derivatives to Pauli--GKSL coefficients and then enforce the positivity constraints required for a valid Lindbladian.

        \paragraph{Parameter learning from Hamiltonian dynamics}
        There is extensive literature on Hamiltonian learning from real-time dynamics~\cite{BaireyAradLindner2019,ZubidaYitzhakiLindnerBairey2021, HaahKothariTang2024,DutkiewiczOBrienSchuster2024, HuangTongFangSu2023,HuMaGongYeTongFlammiaYelin2025, BakshiLiuMoitraTang2024,Zhao2025}, including in the time-dependent regime~\cite{StilckFranca.2025}.  Recent advances have pushed two complementary directions that are most relevant to our work. Hu et al.~\cite{HuMaGongYeTongFlammiaYelin2025} develop an ansatz-free, Heisenberg-limited protocol for learning sparse Hamiltonians from dynamics, achieving estimation error scaling at the Heisenberg limit for Hamiltonian coefficients but targeting unitary generators rather than full GKSL Lindbladians.

        In contrast, our generic result applies to general $k$-local dissipative dynamics under bounded weighted interaction strength, without a degree or sparsity promise. Crucially, the dissipative part of a Lindbladian is a positive Kossakowski matrix rather than a vector of Hamiltonian coefficients, so a coefficient-wise estimator is insufficient: one must output a generator that satisfies the full GKSL positivity constraints.

        \paragraph{Learning from Gibbs states}
        For open-system dynamics, one line of work learns Lindbladians from steady states \cite{BaireyGuoPolettiLindnerArad2020}, while another learns local Hamiltonians, or the corresponding Gibbs states, from thermal equilibrium data \cite{AnshuArunachalamKuwaharaSoleimanifar2021, HaahKothariTang2024,BakshiLiuMoitraTang2026, ChenAnshuNguyen2025,Artymowicz2024}. These works exploit the special structure of stationary or thermal states and do not address the direct learning of general local dissipative generators from short-time dynamics.

        \paragraph{Learning open quantum systems}
        Other works learn Lindbladians from real-time data but typically impose strong structural or weak-noise assumptions, or provide numerical procedures without worst-case guarantees~\cite{PastoriOlsacherKokailZoller2022,VandenBergMitchellWeiMalekakhlagh2025,OlsacherKraftKokailKrausZoller2025,birke2026demonstratingbenchmarkingclassicalshadows}. Ivashkov et al.~\cite{IvashkovEtAl2026} learn unknown sparse, potentially nonlocal Lindbladians from product-Pauli data, with a conditioning factor in the coefficient stage.  Our fixed-locality results use a different structural promise, avoid this factor, and include a valid-GKSL diamond guarantee.  Section~\ref{sec:concurrent-work} compares the two settings, and Proposition~\ref{prop:rank-completion-conditioning} examines conditioning.

        \paragraph{Structure learning.}
        Classical structure learning for graphical models has been extensively studied, including algorithms for Markov random fields and Ising models \cite{BreslerMosselSly2013,Bresler2015,VuffrayMisraLokhovChertkov2016,HamiltonKoehlerMoitra2017,KlivansMeka2017}, information-theoretic lower bounds \cite{SanthanamWainwright2012}, and learning from Glauber-type dynamics \cite{NIPS2014_7d260e35,GaitondeMoitraMossel2024}. For $k$-wise Markov random fields learned from i.i.d. samples from the Gibbs distribution, the $n^{\Theta(k)}$-time dependence is believed to be unavoidable because even an order-$(k+1)$ MRF with a single nonzero interaction can encode $k$-sparse parity with noise \cite{BreslerGamarnikShah2018,KlivansMeka2017}. By contrast, recent work shows that this noisy-parity barrier can be bypassed when the learner observes trajectories of the Glauber dynamics, yielding fixed-parameter polynomial-time learning for higher-order MRFs from dynamics \cite{GaitondeMoitraMossel2024}. The Hamiltonian result of \cite{BakshiLiuMoitraTang2024} is analogous in spirit: it obtains a quantum structure-learning speedup by using a Goldreich--Levin-type search over directly accessible Hamiltonian Pauli coefficients. Our clipped local sparsity $r_\tau$ is the corresponding analogue of their $B_2$ sparsity. In our Lindbladian setting, this search step is not presently available, because the desired $\chi$-coefficients are related to PTM data only after inverting local coefficient sums. Our support-free algorithms check all $\mathcal O_k(n^k)$ candidate coordinates, which is polynomial for fixed $k$ but not fixed-parameter tractable in $k$.

        \paragraph{Independent contemporaneous and subsequent work.} Several preprints appeared independently around our initial posting or subsequently \cite{Yelin.2026,Itai.2026,Sinha2026,LewisTangWright2026, ZhouGong2026,SongEtAl2026,ChenYu2026}.  The original results and proofs were developed independently.  Section~\ref{sec:concurrent-work} compares their assumptions and guarantees with ours.

        \paragraph{Organization.}
        In Section \ref{sec:framework-and-definitions}, we fix notation and relate the Hamiltonian and dissipative coefficients $h,G$ to the Pauli-superoperator coefficients $\chi$.  Section \ref{sec:robust-local-inversion} derives the local Fierz transform, the exact local coefficient inversion, the threshold-before-inversion estimate, and the complementary guarded variants. Section \ref{sec:learning-ptm-elements} estimates the required PTM entries from short-time process-shadow data and proves generic and local-tail coefficient recovery.  Section \ref{sec:sdp-projection} projects in local coefficient $1$-norm onto valid $k$-local Lindblad generators and derives the diamond-norm guarantee.  We then treat model misspecification in Section \ref{sec:agnostic-lindbladian-learning}, prove access-model lower bounds in Section \ref{sec:lower-bound-product-measurements}, and apply the learning method to local-observable verification in Section~\ref{sec:local-observable-verification}.  Independent contemporaneous and subsequent results are compared in Section~\ref{sec:concurrent-work}. Appendices~\ref{app:rank-selection-conditioning} and~\ref{app:sdp-projection-details} give a conditioning counterexample for rank-only probe selection and technical details of the physical projection, respectively.

\section{Framework}\label{sec:framework-and-definitions}
    We consider a quantum system of $n$ qubits with underlying Hilbert space $\cH:=\bigotimes_{j\in[n]}\cH_j$, $\cH_j\equiv \mathbb{C}^2$, and denote by $\mathcal{P}_n = \{I,X,Y,Z\}^{\otimes n}$ the set of $n$-qubit Pauli strings
    \begin{equation}\label{eq:pauli-strings}
        \mathbf{P}=\bigotimes_{j=1}^n P_j
    \end{equation}
    where $P_j\in \{I,X,Y,Z\}$ is a Pauli matrix acting on site $j$. The set $\mathcal{P}_n$ forms an orthogonal basis with respect to the Hilbert-Schmidt inner product, i.e.~$\Tr(\mathbf{P}\mathbf{Q}) = 2^n \delta_{\mathbf{P},\mathbf{Q}}$. The support of a Pauli string $\mathbf{P}$ is defined as the set $\operatorname{supp}(\mathbf{P}):= \{j\in [n]:P_j\ne I\}$, and its weight is defined as the number of qubits on which it acts nontrivially, i.e.~$\wt(\mathbf{P}) = |\operatorname{supp}(\mathbf{P})|$. For any region $R\subseteq [n]$ of size $r$, let $\mathcal P_R$ denote the set of Pauli strings on $R$, including the identity. We denote by $\mathbf{P}_{R}\in\mathcal{P}_{R}$ the Pauli substring obtained from $\mathbf{P}$ by discarding qubits in $\overline{R}:=[n]\backslash R$, so that $\mathbf{P}=\mathbf{P}_R\otimes \mathbf{P}_{\overline{R}}$. For $0\le r\le n$, we denote by $\mathcal{P}_{n,r}$ the set of Pauli strings of weight $r$, and by $\mathcal{P}_{n,\le r}$ that of Pauli strings of weight at most $r$. We also let $   \cP_{\le r}^{\circ} := \mathcal P_{n,\le r}\setminus\{I\}$ denote the non-identity Pauli strings of weight at most $r$. For every region $R\subseteq[n]$, let $ \cP_R^\circ := \{\mathbf P\in\mathcal P_n\setminus\{I\}:\supp(\mathbf P)\subseteq R\}$ and set $d_R:=|\cP_R^\circ|=4^{|R|}-1$. We denote by $\|\bullet\|_2$ the Hilbert-Schmidt norm on $\mathbb{M}_{2^n}$. For coefficient arrays, $\|\cdot\|_\infty$ denotes the largest absolute entry. For operators on the system Hilbert space, $\|\cdot\|_\infty$ denotes the operator norm (largest singular value). The induced superoperator norm is $\|\mathcal K\|_{\infty\to\infty}:=\sup_{A\ne0}\|\mathcal K(A)\|_\infty/\|A\|_\infty$. The family $\{E_{\mathbf{P},\mathbf{Q}}\}_{\mathbf{P},\mathbf{Q}\in\mathcal{P}_n}$, defined by $E_{\mathbf{P},\mathbf{Q}}\coloneqq \mathbf{P}\bullet \mathbf{Q}$, forms a basis of the space of linear maps over $\mathbb{M}_{2^n}$.

    Throughout this paper, we assume that the dynamics of the system is governed by a $k$-local Markovian quantum master equation
    \begin{equation*}
        \frac{d}{dt} \rho(t) = \mathcal{L}(\rho(t))\,,
    \end{equation*}
    determined by the $k$-local generator $\mathcal{L}$ in GKSL form \cite{GKS1976,Lindblad1976} defined below. For the set $\mathcal{R}_{n,\le k}$ of subsets of $[n]$ of size at most $k$, we define the local GKSL generator by
    \begin{equation}\label{eq:linblad-form}
        \cL=\sum_{e\in \mathcal{R}_{n,\le k}}\cL_e\equiv-i[H,\bullet ]+\sum_{e\in \mathcal{R}_{n,\le k}}\sum_{a=1}^{r_e}\Bigl(L_{e,a}\bullet L_{e,a}^\dagger-\tfrac{1}{2}\{L_{e,a}^\dagger L_{e,a},\bullet \}\Bigr),
    \end{equation}
    for a local, traceless Hamiltonian $H:=\sum_{\mathbf{P}\in\mathcal{P}_{n,\le k}}h_{\mathbf{P}}\mathbf{P}$, $r_e\in\N$, and traceless jump operators $L_{e,a}=\sum_{\mathbf{P}\in \mathcal{P}_{e}}\ell_{e,a,\mathbf{P}}\mathbf{P}$. Note that $\cL_e$ denotes full contribution of Hamiltonians and jump operators supported on $e$. Writing the Hamiltonian part of the local term as $H_e=\sum_{\mathbf P\in\mathcal P_e}h_{e,\mathbf P}\mathbf P$, we extend its local coefficient arrays by $h_{e,\mathbf P}=\ell_{e,a,\mathbf P}=0$ for $\mathbf P\notin\mathcal P_e$. We call a representation of the form~\eqref{eq:linblad-form} a local-positive $k$-local decomposition. In the physical representation above, the Lindblad operators describe localized dissipative channels (elementary quantum jumps (e.g., decay, dephasing, loss)) that model how the system exchanges energy or information with its environment. Equivalently, expanding the same generator in the Pauli basis yields the Kossakowski matrix, a positive semidefinite matrix of coefficients that encodes the same dissipative structure algebraically and whose sparsity and locality-induced block structure make it especially useful for learning theory and practical reconstruction from Pauli measurement statistics. This structure is investigated in the following result:
    \begin{lemma}\label{lem:locality-kossakowski}
        The generator $\cL$ of \eqref{eq:linblad-form} is uniquely expanded in the Pauli string basis
        \begin{equation}\label{eq:pauli-gksl-form}
            \mathcal{L} = -i[H, \bullet ] + \sum_{\mathbf{P},\mathbf{Q}\in\mathcal{P}_n\backslash \{I\}} G_{\mathbf{P},\mathbf{Q}} \left( \mathbf{P} \bullet  \mathbf{Q} - \frac{1}{2}\{\mathbf{Q} \mathbf{P}, \bullet \} \right)
        \end{equation}
        where
        \begin{equation*}
            G_{\mathbf{P},\mathbf{Q}}=\sum_{e\in \mathcal{R}_{n,\le k}}\sum_{a=1}^{r_e} {\ell_{e,a,\mathbf{P}}}\,\overline{\ell_{e,a,\mathbf{Q}}}.
        \end{equation*}
        In particular, $G_{\mathbf{P},\mathbf{Q}}=0$ whenever $|\mathrm{supp}(\mathbf{P})\cup\mathrm{supp}(\mathbf{Q})|>k$. Thus the number of nonzero matrix coefficients $G_{\mathbf{P},\mathbf{Q}}$ and vector coefficients $h_{\mathbf{P}}$ is bounded by $\mathcal{O}(n^k)$.
    \end{lemma}
    \begin{proof}
        Expanding each jump operator $L_{e,a}=\sum_{\mathbf P\in\mathcal P_n}\ell_{e,a,\mathbf P}\mathbf P$, with the convention that $\ell_{e,a,\mathbf P}=0$ whenever $\mathbf P\notin\mathcal P_e$, gives
        \begin{equation*}
            \cL = -i[H,\bullet] + \sum_{\mathbf P,\mathbf Q\in\mathcal P_n} G_{\mathbf P,\mathbf Q}\left( \mathbf P\bullet\mathbf Q - \frac12\{\mathbf Q\mathbf P,\bullet\} \right),
        \end{equation*}
        where
        \begin{equation*}
            G_{\mathbf P,\mathbf Q} = \sum_{e\in\mathcal R_{n,\le k}} \sum_{a=1}^{r_e} \ell_{e,a,\mathbf P}\overline{\ell_{e,a,\mathbf Q}}\,.
        \end{equation*}
        The terms with $\mathbf P=I$ or $\mathbf Q=I$ can be absorbed into the Hamiltonian part and the scalar part of the anticommutator, yielding the stated form with $\mathbf P,\mathbf Q\in\mathcal P_n\setminus\{I\}$. We now prove uniqueness. Suppose that two pairs $(H,G)$ and $(H',G')$ give the same generator in the form \eqref{eq:pauli-gksl-form}. Set
        \begin{equation*}
            \Delta H:=H-H'\qquad\text{ and }\qquad \Delta G:=G-G'\,.
        \end{equation*}
        Subtracting the two expressions gives, for all $X\in\mathbb M_{2^n}$,
        \begin{equation*}
            -i[\Delta H,X] + \sum_{\mathbf P,\mathbf Q\in\mathcal P_n\setminus\{I\}} \Delta G_{\mathbf P,\mathbf Q} \left(\mathbf P X\mathbf Q - \frac12\{\mathbf Q\mathbf P,X\} \right)=0\,.
        \end{equation*}
        Equivalently,
        \begin{equation*}
            \left(\sum_{\mathbf P,\mathbf Q\in\mathcal P_n\setminus\{I\}}\!\!\!\!\!\!\frac{\Delta G_{\mathbf P,\mathbf Q}\mathbf Q\mathbf P}{2} + i\Delta H \!\!\right)\!\!XI + IX\!\!\left( \sum_{\mathbf P,\mathbf Q\in\mathcal P_n\setminus\{I\}}\!\!\!\!\!\! \frac{\Delta G_{\mathbf P,\mathbf Q}\mathbf Q\mathbf P}{2} - i\Delta H \!\!\right)\! = \!\!\!\!\!\!\sum_{\mathbf P,\mathbf Q\in\mathcal P_n\setminus\{I\}}\!\!\!\!\!\!\!\! \Delta G_{\mathbf P,\mathbf Q}\, \mathbf P X\mathbf Q .
        \end{equation*}
        Since the maps $E_{\mathbf A,\mathbf B}:X\mapsto\mathbf A X\mathbf B$, $\mathbf A,\mathbf B\in\mathcal P_n$, form a basis of $\mathcal B(\mathbb M_{2^n})$, the first two terms above lie in the span of basis elements with at least one index equal to $I$, namely $E_{\mathbf A,I}$ or $E_{I,\mathbf A}$. The last term lies in the span of basis elements $E_{\mathbf P,\mathbf Q}$ with $\mathbf P,\mathbf Q\neq I$. These two families are disjoint, so their coefficients vanish separately. Therefore
        \begin{equation*}
            \sum_{\mathbf P,\mathbf Q\in\mathcal P_n\setminus\{I\}} \Delta G_{\mathbf P,\mathbf Q}E_{\mathbf P,\mathbf Q}=0,
        \end{equation*}
        and linear independence gives $\Delta G_{\mathbf P,\mathbf Q}=0$ for all $\mathbf P,\mathbf Q\in\mathcal P_n\setminus\{I\}$. Thus $G=G'$. Substituting $\Delta G=0$ back into the difference identity yields
        \begin{equation*}
            -i[\Delta H,X]=0 \qquad \text{for all }X\in\mathbb M_{2^n}.
        \end{equation*}
        Hence $\Delta H$ commutes with every matrix in $\mathbb M_{2^n}$, so $\Delta H=cI$ for some scalar $c$. Since both $H$ and $H'$ are expanded only over nonidentity Pauli strings, $\Delta H$ is traceless. Hence $c=0$, and therefore $\Delta H=0$. This proves uniqueness. Finally, if $|\operatorname{supp}(\mathbf P)\cup\operatorname{supp}(\mathbf Q)|>k$, then there is no $e\in\mathcal R_{n,\le k}$ such that both $\mathbf P,\mathbf Q\in\mathcal P_e$. Hence for every $e,a$, at least one of $\ell_{e,a,\mathbf P}$ and $\ell_{e,a,\mathbf Q}$ vanishes, and therefore $G_{\mathbf P,\mathbf Q}=0$. Since $k$ is fixed, the number of pairs $(\mathbf P,\mathbf Q)$ with $|\operatorname{supp}(\mathbf P)\cup\operatorname{supp}(\mathbf Q)|\le k $ is $O(n^k)$, and the same is true for the number of nonzero Hamiltonian coefficients $h_{\mathbf P}$. This completes the proof.
    \end{proof}
    In what follows, we consider two more useful representations of $\cL$: first, the Pauli transfer matrix (PTM) of $\mathcal{L}$ has coefficients
    \begin{equation}
        L_{\mathbf{P},\mathbf{Q}} = \frac{1}{2^n} \Tr(\mathbf{P} \mathcal{L}(\mathbf{Q}))\,,\qquad \mathbf{P},\mathbf{Q}\in\mathcal{P}_n\,.
    \end{equation}
    The PTM representation allows us to express the action of $\mathcal{L}$ on any operator in terms of its expansion in the Pauli basis. This representation is standard in quantum process tomography and quantum information (see e.g.~\cite{NielsenChuang2000,ChuangNielsen1997,Poyatos1997}). To avoid confusion we reserve $\mathcal{L}$ for the superoperator and use the capital-letter matrix $L$ for its Pauli-transfer matrix with entries $L_{\mathbf{P},\mathbf{Q}}$. Second, we write the expansion of $\mathcal{L}$ in the Pauli-superoperator basis $\{E_{\mathbf{P},\mathbf{Q}}\}_{\mathbf{P},\mathbf{Q}\in \mathcal{P}_n}$ as
    \begin{equation}
        \mathcal{L} = \sum_{\mathbf{P},\mathbf{Q}\in\mathcal{P}_n} \chi_{\mathbf{P},\mathbf{Q}} E_{\mathbf{P},\mathbf{Q}}\,.
        \label{eq:chi-expansion}
    \end{equation}
    In the following, we often refer to the coefficients $\chi_{\mathbf{P},\mathbf{Q}}$ of a superoperator $\cL$ as its GKSL or $\chi$-coefficients. From the decomposition in Lemma \ref{lem:locality-kossakowski}, it follows directly that for any $\mathbf{P},\mathbf{Q}\in\mathcal{P}_n\backslash \{I\}$,
    \begin{equation}\label{eq:chiG}
        G_{\mathbf{P},\mathbf{Q}} = \chi_{\mathbf{P},\mathbf{Q}} \qquad \text{ and }\qquad  h_{\mathbf{P}} = \frac{i}{2}\bigl(\chi_{\mathbf{P},I}-\chi_{I,\mathbf{P}}\bigr)\,.
    \end{equation}
    \begin{rmk}\label{rmk:chi-id-id}
        Note that $\chi_{I,I}$ is not required to recover all necessary coefficients of $G_{\mathbf{P}\mathbf{Q}}$ and $h_\mathbf{P}$ describing the considered generator. Nevertheless, it is given by
        \begin{equation*}
            \chi_{I,I}=-\sum_{\mathbf P\in\mathcal P_n\setminus\{I\}}\chi_{\mathbf P,\mathbf P}=-\sum_{\mathbf P\in\mathcal P_n\setminus\{I\}}G_{\mathbf P,\mathbf P}
        \end{equation*}
        Thus, after estimating or projecting $G$, one may recover $\chi_{I,I}$ by this diagonal sum. The learning task below concerns the independent nonidentity coefficients.
    \end{rmk}

    We combine the Hamiltonian and dissipative coefficients into $x=(h,G)$, where $h\in\mathbb R^{|\cP_{\le k}^{\circ}|}$ and $G=G^\dagger\succeq0$. For a Hamiltonian index $v=\mathbf P\ne I$, set $S_v:=\supp(\mathbf P)$. For a dissipative or $\chi$-index $v=(\mathbf P,\mathbf Q)\ne(I,I)$, set $S_v:=\supp(\mathbf P)\cup\supp(\mathbf Q)$. For any array $z$, define
    \begin{equation}\label{eq:local-coefficient-norm-degree}
        \|z\|_{\mathrm{loc},1}
        :=\max_{j\in[n]}\sum_{v:\,j\in S_v}|z_v|,
        \qquad
        \deg(z)
        :=\max_{j\in[n]}\bigl|\{v:j\in S_v,\ z_v\ne0\}\bigr|.
    \end{equation}
    This norm always counts each coordinate once. We write the Hamiltonian weight explicitly as $\|(2h,G)\|_{\mathrm{loc},1}$ when bounding physical coefficient strengths or errors. We call a generator coefficient-local when its Kossakowski matrix $G\succeq0$ and every nonzero coefficient is supported on at most $k$ sites. This is weaker than requiring a sum of positive local GKSL blocks.

    Let $F$ map a coefficient array $(h,G)$ to the $\chi$-coefficients obtained by expressing \eqref{eq:pauli-gksl-form} in the basis of \eqref{eq:chi-expansion}, omitting $\chi_{I,I}$. Let $\Psi\chi$ be the pair $(h,G)$ recovered through \eqref{eq:chiG}. Both maps extend linearly to nonphysical coefficient arrays.
    Each coefficient contributes to at most three $\chi$-coordinates, with total absolute value at most twice its own and supports contained in its original support. Therefore,
    \begin{equation}\label{eq:physical-chi-map-bounds}
        \begin{gathered}
            \Psi F=\mathrm{id},
            \quad
            \|Fz\|_{\mathrm{loc},1}\le2\|z\|_{\mathrm{loc},1},
            \quad
            \deg(Fz)\le3\deg(z),\quad
            \|(2(\Psi\chi)^h,(\Psi\chi)^G)\|_{\mathrm{loc},1}
            \le\|\chi\|_{\mathrm{loc},1}.
        \end{gathered}
    \end{equation}
    Thus $\Psi$ is a left inverse of $F$, and $F\Psi\chi=\chi$ precisely for $\chi$ in the image of $F$.

\section{Robust local inversion}\label{sec:robust-local-inversion}
    The first main novel technical tool of the present paper is a simple inversion formula providing the $\chi$-coefficients in terms of the PTM coefficients of $\cL$: substituting Eq.~(\ref{eq:chi-expansion}) into the definition of the PTM element $L_{\mathbf{P}, \mathbf{Q}}$, we obtain:
    \begin{equation}
        L_{\mathbf{P},\mathbf{Q}} = \frac{1}{2^n} \sum_{\mathbf{P}',\mathbf{Q}'\in\mathcal{P}_n} \chi_{\mathbf{P}',\mathbf{Q}'} \Tr(\mathbf{P} \mathbf{P}'\mathbf{Q}\mathbf{Q}')\,.
        \label{eq:L-from-chi}
    \end{equation}
    To invert this linear system and solve for $\chi_{\mathbf{P},\mathbf{Q}}$, we multiply both sides of Eq.~(\ref{eq:L-from-chi}) by $\Tr(\mathbf{P} \mathbf{P}''\mathbf{Q} \mathbf{Q}'')$ and sum over all $\mathbf{P},\mathbf{Q} \in \mathcal{P}_n$. Next, we apply the Fierz trace identity for Pauli matrices (see Lem.~\ref{lem:fierz}), which states that for any matrices $A, B, C, D$:
    \begin{equation}\label{Fierzequation}
        \sum_{\mathbf{P},\mathbf{Q}\in\mathcal{P}_n} \Tr(\mathbf{P} A \mathbf{Q} B) \Tr(\mathbf{P} C \mathbf{Q} D) = 4^n \Tr(AD) \Tr(BC)\,.
    \end{equation}
    For $A=\mathbf{P}', B=\mathbf{Q}', C=\mathbf{P}'', D=\mathbf{Q}''$, this yields:
    \begin{align*}
        \sum_{\mathbf{P},\mathbf{Q}} L_{\mathbf{P},\mathbf{Q}} \Tr(\mathbf{P} \mathbf{P}''\mathbf{Q} \mathbf{Q}'') &= \frac{1}{2^n} \sum_{\mathbf{P}',\mathbf{Q}'} \chi_{\mathbf{P}',\mathbf{Q}'} \left[ 4^n \Tr(\mathbf{P}'\mathbf{Q}'') \Tr(\mathbf{Q}'\mathbf{P}'') \right] \\
        &= 2^n \sum_{\mathbf{P}',\mathbf{Q}'} \chi_{\mathbf{P}',\mathbf{Q}'} (2^n \delta_{\mathbf{P}',\mathbf{Q}''}) (2^n \delta_{\mathbf{Q}',\mathbf{P}''}) = 2^{3n} \chi_{\mathbf{Q}'',\mathbf{P}''}\,.
    \end{align*}
    Relabelling the indices, we obtain the following exact inversion formula:
    \begin{lemma}[Global inversion formula]\label{lem:global-inversion}
        With the notation of the previous paragraph, for any two Pauli strings $\mathbf{P},\mathbf{Q}\in\mathcal{P}_n$,
        \begin{equation*}
            \chi_{\mathbf{P},\mathbf{Q}} = \frac{1}{2^{3n}} \sum_{\mathbf{P}',\mathbf{Q}'\in\mathcal{P}_n} L_{\mathbf{P}',\mathbf{Q}'} \Tr(\mathbf{P}' \mathbf{Q} \mathbf{Q}'\mathbf{P})\,.
        \end{equation*}
    \end{lemma}

    \subsection{Fierz trace identity}
        It remains to show \eqref{Fierzequation}:
        \begin{lemma}[Fierz trace identity for Pauli matrices]\label{lem:fierz}
            For any matrices $A, B, C, D \in \mathbb{M}_{2^n}$, the following identity holds:
            \begin{equation}
                \sum_{\mathbf{P},\mathbf{Q}\in\mathcal{P}_n} \Tr(\mathbf{P} A \mathbf{Q} B) \Tr(\mathbf{P} C \mathbf{Q} D) = 4^n \Tr(AD) \Tr(BC)\,.
            \end{equation}
        \end{lemma}
        \begin{proof}
            To prove the Fierz trace identity, we express the traces in explicit matrix index notation and rely on the completeness relation for the orthogonal basis of $2^n \times 2^n$ matrices $\mathbf{P}$ satisfying $\Tr(\mathbf{P} \mathbf{Q}) = 2^n \delta_{\mathbf{P},\mathbf{Q}}$, which is given by:
            \begin{equation}\label{eq:completeness}
                \sum_{\mathbf{P}} (\mathbf{P})_{ij} (\mathbf{P})_{kl} = 2^n \delta_{i,l} \delta_{j,k}\,.
            \end{equation}
            This can be seen by expanding an arbitrary $2^n\times 2^n$ matrix $M = \sum_{\mathbf{P}} c_{\mathbf{P}} \mathbf{P}$ with $c_{\mathbf{P}} = \frac{1}{2^n} \text{Tr}(\mathbf{P} M)$. In terms of matrix components $i, j$, and writing out the trace and reordering the sums, we have:
            \begin{align*}
                M_{ij} &= \frac{1}{2^n} \sum_{\mathbf{P}} \left( \sum_{k,l} (\mathbf{P})_{kl} M_{lk} \right) (\mathbf{P})_{ij}= \sum_{k,l} M_{lk} \left( \frac{1}{2^n} \sum_{\mathbf{P}} (\mathbf{P})_{ij} (\mathbf{P})_{kl} \right)
            \end{align*}
            For this equality to hold for an arbitrary matrix $M$, the term in parentheses satisfies Equation \eqref{eq:completeness}. With that identity, we expand the left-hand side of the claimed identity:
            \begin{align*}
                \sum_{\mathbf{P},\mathbf{Q}}& \Tr(\mathbf{P} A \mathbf{Q} B) \Tr(\mathbf{P} C \mathbf{Q} D)\\
                &= \sum_{\substack{i_1,...,i_4\\j_1,...,j_4}}\sum_{\mathbf{P},\mathbf{Q}} (\mathbf{P})_{i_1i_2} A_{i_2i_3} (\mathbf{Q})_{i_3i_4} B_{i_4i_1} (\mathbf{P})_{j_1j_2} C_{j_2j_3} (\mathbf{Q})_{j_3j_4} D_{j_4j_1} \\
                &= \sum_{\substack{i_1,...,i_4\\j_1,...,j_4}}\Bigl(\sum_{\mathbf{P}} (\mathbf{P})_{i_1i_2}(\mathbf{P})_{j_1j_2}\Bigr)\Bigl(\sum_{\mathbf{Q}}(\mathbf{Q})_{i_3i_4}(\mathbf{Q})_{j_3j_4}\Bigr)A_{i_2i_3}B_{i_4i_1}C_{j_2j_3}D_{j_4j_1} \\
                &= 4^n\sum_{\substack{i_1,...,i_4\\j_1,...,j_4}}\delta_{i_1,j_2}\delta_{i_2,j_1}\delta_{i_3,j_4}\delta_{i_4,j_3}A_{i_2i_3}B_{i_4i_1}C_{j_2j_3}D_{j_4j_1} \\
                &= 4^n\sum_{i_1,...,i_4}A_{i_2i_3}B_{i_4i_1}C_{i_1i_4}D_{i_3i_2} \\
                &= 4^n\operatorname{Tr}(AD)\operatorname{Tr}(BC)\,,
            \end{align*}
            which finishes the proof.
        \end{proof}

    \subsection{Local inversion} \label{genericscheme}
        At this stage, it is unclear how to efficiently compute the coefficients $\chi_{\mathbf{P},\mathbf{Q}}$ from Lemma \ref{lem:global-inversion}, since the sum involves an exponential number of coefficients. For this, given a region $R\subseteq[n]$ to be fixed later, we define the maps
        \begin{align*}
            \cL^{(R)}:=\Tr_{\overline{R}}\left[\cL\left(\bullet \otimes \frac{I_{\overline{R}}}{2^{|\overline{R}|}}\right)\right].
        \end{align*}
        A direct consequence of this definition is that the corresponding local PTM entries are given by the same formula as the global PTM entries:
        \begin{equation}\label{eq:local-vs-global-PTM}
            L^{(R)}_{\mathbf{P}_R,\mathbf{Q}_R} := \frac{1}{2^{|R|}}\Tr_{R}\bigl(\mathbf{P}_{R}\cL^{(R)}(\mathbf{Q}_{R})\bigr) = \frac{1}{2^n}\Tr\bigl((\mathbf{P}_{R}\otimes I_{\overline{R}})\,\cL(\mathbf{Q}_R\otimes I_{\overline{R}})\bigr) = L_{\mathbf{P}_R\otimes I_{\overline{R}},\mathbf{Q}_R\otimes I_{\overline{R}}}\,.
        \end{equation}
        Next, we get
        \begin{align*}
            \cL^{(R)}=\sum_{\mathbf{P},\mathbf{Q}\in\mathcal{P}_n}\chi_{\mathbf{P},\mathbf{Q}}\Tr_{\overline{R}}\left[\mathbf{P}(\bullet\otimes I_{\overline{R}}/2^{|\overline{R}|})\mathbf{Q}\right]=\sum_{\substack{\mathbf{P},\mathbf{Q}\in\mathcal{P}_n\\\mathbf{P}_{\overline{R}}=\mathbf{Q}_{\overline{R}}}}\chi_{\mathbf{P},\mathbf{Q}}\,\mathbf{P}_R\bullet \mathbf{Q}_R.
        \end{align*}
        By uniqueness of the $\chi$-decomposition, denoting by $\chi^{(R)}_{\mathbf{P}_R,\mathbf{Q}_R}$ the $\chi$-coefficients of $\cL^{(R)}$, the above equation yields
        \begin{align}\label{eq:chiRtochi}
            \chi^{(R)}_{\mathbf{P}_R,\mathbf{Q}_R}=\sum_{\mathbf A_{\overline R}\in\mathcal P_{\overline R}}\chi_{\mathbf{P}_R\otimes\mathbf A_{\overline R},\mathbf{Q}_R\otimes\mathbf A_{\overline R}}\,.
        \end{align}
        Equation~\eqref{eq:chiRtochi} expresses each local coefficient as a sum over diagonal extensions. We call these sums coefficient marginals.  For $v=(\mathbf P,\mathbf Q)\in\mathcal I_k^\circ$, write $S:=S_v=\supp(\mathbf P)\cup\supp(\mathbf Q)$ and define the local coefficient sums
        \begin{equation*}
            y_v:=\chi^{(S)}_{\mathbf P_S,\mathbf Q_S}
            =\sum_{\mathbf A_{\overline S}}
            \chi_{\mathbf P_S\otimes\mathbf A_{\overline S},
            \mathbf Q_S\otimes\mathbf A_{\overline S}}
            =:(\mathsf Z\chi)_v,
        \end{equation*}
        where $\mathbf A_{\overline S}\in\mathcal P_{\overline S}$ ranges over all Pauli strings outside $S$, including the identity, and is appended to both indices. By $k$-locality, only strings with $\operatorname{wt}(\mathbf A_{\overline S})\le k-|S|$ contribute.
        We use the local norm in \eqref{eq:local-coefficient-norm-degree} for both arrays, counting $\chi_v$ and $y_v$ at each site in $S_v$. The inverse below uses inclusion--exclusion, a form of M\"obius inversion, to cancel contributions from larger supports. The following lemma expresses the descending local inversion in closed form and bounds its stability in the local coefficient norm.
        \begin{lemma}[Local inclusion--exclusion]\label{lem:local-inclusion-exclusion}
            For $v=(\mathbf P,\mathbf Q)\in\mathcal I_k^\circ$ and $S=S_v$,
            \begin{equation}\label{eq:local-inclusion-exclusion}
                (\mathsf Z^{-1}y)_v
                =
                \sum_{\substack{\mathbf A_{\overline S}:\\
                \operatorname{wt}(\mathbf A_{\overline S})\le k-|S|}}
                (-1)^{\operatorname{wt}(\mathbf A_{\overline S})}
                y_{(\mathbf P_S\otimes\mathbf A_{\overline S},
                \mathbf Q_S\otimes\mathbf A_{\overline S})},
            \end{equation}
            and
            \begin{equation}\label{eq:local-inclusion-exclusion-norm}
                \|\mathsf Z\|_{\mathrm{loc},1\to\mathrm{loc},1}\le2^k,
                \qquad
                \|\mathsf Z^{-1}\|_{\mathrm{loc},1\to\mathrm{loc},1}\le2^k.
            \end{equation}
        \end{lemma}
        \begin{proof}
            Substituting the definition of $y$ into the right-hand side of \eqref{eq:local-inclusion-exclusion} and collecting terms gives
            \begin{equation*}
                \sum_{\substack{\mathbf A_{\overline S}\in\mathcal P_{\overline S}\\
                \operatorname{wt}(\mathbf A_{\overline S})\le k-|S|}}
                \chi_{\mathbf P_S\otimes\mathbf A_{\overline S},\mathbf Q_S\otimes\mathbf A_{\overline S}}
                \sum_{T\subseteq\operatorname{supp}(\mathbf A_{\overline S})}(-1)^{|T|}.
            \end{equation*}
            The subsets $T$ specify which added sites occur in the intermediate local coefficient sum. The inner sum vanishes unless $\mathbf A_{\overline S}=I_{\overline S}$, when it equals one. Thus only $\chi_v$ remains, proving the inverse formula. A column of either $\mathsf Z$ or $\mathsf Z^{-1}$ has at most $2^k$ nonzero entries, each of absolute value one, and every output support is contained in the column support. Summing the column contributions incident on any fixed site proves \eqref{eq:local-inclusion-exclusion-norm}.
        \end{proof}
        Note that the above bounds are local-coefficient-$1$-norm bounds and do not contradict the system-size-dependent entrywise bound below. By the same proof as for the inversion formula Lemma \ref{lem:global-inversion}, we obtain
        \begin{align}\label{eq:localinversionform}
            \chi^{(R)}_{\mathbf{P}_R,\mathbf{Q}_R}=\frac{1}{2^{3|R|}}\sum_{\mathbf{P}'_R,\mathbf{Q}'_R}L^{(R)}_{\mathbf{P}'_R,\mathbf{Q}'_R}\Tr(\mathbf{P}'_R\mathbf{Q}_R\mathbf{Q}'_R\mathbf{P}_R).
        \end{align}
        We refer to this equation as a local inversion formula. This can be efficiently computed as long as one has access to the coefficients $L^{\smash{(R)}}_{\mathbf{P}_R,\mathbf{Q}_R}=L_{\mathbf{P}_R\otimes I_{\overline{R}},\mathbf{Q}_{R}\otimes I_{\overline{R}}}$. Now, by Lemma \ref{lem:locality-kossakowski} and Equation~\eqref{eq:chiG}, for any $\mathbf{P},\mathbf{Q}\in\mathcal{P}_n\backslash \{I\}$, $\chi_{\mathbf{P},\mathbf{Q}}=0$ whenever $|\mathrm{supp}(\mathbf{P})\cup\mathrm{supp}(\mathbf{Q})|>k$. Let $\mathcal R_{n,\le k}^\circ:=\mathcal R_{n,\le k}\setminus\{\emptyset\}$ denote the nonempty regions. For any $(\mathbf P,\mathbf Q)\in\mathcal I_k^\circ$,
        \begin{equation}
            \chi_{\mathbf{P},\mathbf{Q}}\ne 0\,\Rightarrow\, \exists R\in\mathcal R_{n,\le k}^\circ:\,\mathbf{P}_{\overline{R}}=\mathbf{Q}_{\overline{R}}=I_{\overline{R}}\,
        \end{equation}
        (simply take $R=\operatorname{supp}(\mathbf{P})\cup\operatorname{supp}(\mathbf{Q})$). Now, choosing a maximal set $R\in\mathcal R_{n,\le k}^\circ$ with respect to the partial order induced by inclusions (i.e. for all $R'\in \mathcal R_{n,\le k}^\circ$, $R\subseteq R'\Rightarrow R=R'$), we get from \eqref{eq:chiRtochi} that for any two Pauli strings $\mathbf{P}_R,\mathbf{Q}_R$ with $\operatorname{supp}(\mathbf{P}_R)\cup\operatorname{supp}(\mathbf{Q}_R)=R$,
        \begin{align*}
            \chi^{(R)}_{\mathbf{P}_R,\mathbf{Q}_R}=\sum_{\mathbf A_{\overline R}}\chi_{\mathbf{P}_R\otimes \mathbf A_{\overline R},\mathbf{Q}_R\otimes\mathbf A_{\overline R}}=\chi_{\mathbf{P}_R\otimes I_{\overline{R}},\mathbf{Q}_R\otimes I_{\overline{R}}}\,,
        \end{align*}
        where the sum collapses to a single element by maximality of $R$. Once the $\chi$-coefficients of $\cL$ have all been computed for maximal regions $R$, we remove them from the list of unprocessed regions, keeping $\mathcal R_{n,\le k}^\circ$ fixed. Then, for any maximal region $R'$ remaining on that list and any Pauli strings $\mathbf{P}_{R'},\mathbf{Q}_{R'}$ with $\operatorname{supp}(\mathbf{P}_{R'})\cup \operatorname{supp}(\mathbf{Q}_{R'})=R'$, we have
        \begin{align*}
            \chi_{\mathbf{P}_{R'}\otimes I_{\overline{R'}},\mathbf{Q}_{R'}\otimes I_{\overline{R'}}}
            &=\chi^{(R')}_{\mathbf{P}_{R'},\mathbf{Q}_{R'}}-\sum_{\substack{\mathbf A_{\overline{R'}}\in\mathcal P_{\overline{R'}}\\
            1\le\operatorname{wt}(\mathbf A_{\overline{R'}})\le k-|R'|}}
            \chi_{\mathbf P_{R'}\otimes\mathbf A_{\overline{R'}},\mathbf Q_{R'}\otimes\mathbf A_{\overline{R'}}}.
        \end{align*}
        Every coefficient subtracted on the right-hand side has already been recovered. Repeating this until all sets in $\mathcal R_{n,\le k}^\circ$ have been explored computes all nonidentity $\chi$-entries in $\mathcal I_k^\circ$. This procedure is turned into an efficient protocol in Algorithm~\ref{alg:recover-chi}.

        \begin{algorithm}[ht]
            \caption{Recovery of the nonidentity $\chi$-coefficients in $\mathcal I_k^\circ$ from local PTM entries}
            \label{alg:recover-chi}
            \begin{algorithmic}[1]
                \REQUIRE Locality parameter $k$ and estimates
                \begin{equation*}
                    \widehat L_{\mathbf A\otimes I_{\overline R},\,\mathbf B\otimes I_{\overline R}}
                \end{equation*}
                for all $R\subseteq[n]$, $|R|\le k$, and all $\mathbf A,\mathbf B\in\mathcal P_R$.
                \ENSURE Estimates $\widehat\chi_{\mathbf P,\mathbf Q}$ for all $(\mathbf P,\mathbf Q)\in\mathcal I_k^\circ$.

                \STATE For every $R\subseteq[n]$, $|R|\le k$, and every $\mathbf P_R,\mathbf Q_R\in\mathcal P_R$, compute the local coefficient
                \begin{equation*}
                    \widehat\chi^{(R)}_{\mathbf P_R,\mathbf Q_R}
                    :=
                    \frac{1}{2^{3|R|}}
                    \sum_{\mathbf A_R,\mathbf B_R\in\mathcal P_R}
                    \widehat L_{\mathbf A_R\otimes I_{\overline R},\,\mathbf B_R\otimes I_{\overline R}}
                    \Tr\!\left(\mathbf A_R\mathbf Q_R\mathbf B_R\mathbf P_R\right).
                \end{equation*}

                \STATE Initialize all recovered global coefficients $\widehat\chi_{\mathbf P,\mathbf Q}$ as undefined.

                \FOR{$s=k,k-1,\ldots,1$\qquad}
                    \FOR{every region $R\subseteq[n]$ with $|R|=s$\qquad}
                        \FOR{every pair $\mathbf P_R,\mathbf Q_R\in\mathcal P_R$ such that $\operatorname{supp}(\mathbf P_R)\cup\operatorname{supp}(\mathbf Q_R)=R$\qquad}
                            \STATE
                            \begin{equation*}
                                \widehat\chi_{\mathbf P_R\otimes I_{\overline R}, \mathbf Q_R\otimes I_{\overline R}} := \widehat\chi^{(R)}_{\mathbf P_R,\mathbf Q_R} - \sum_{\substack{\mathbf A_{\overline R}\in\mathcal P_{\overline R}\\
                                \mathbf A_{\overline R}\neq I_{\overline R}\\
                                |R\cup\operatorname{supp}(\mathbf A_{\overline R})|\le k}} \widehat\chi_{\mathbf P_R\otimes \mathbf A_{\overline R},\mathbf Q_R\otimes \mathbf A_{\overline R}}.
                            \end{equation*}
                        \ENDFOR
                    \ENDFOR
                \ENDFOR

                \RETURN The independent coefficients $\widehat\chi_{\mathbf P,\mathbf Q}$ indexed by $\mathcal I_k^\circ$.
            \end{algorithmic}
        \end{algorithm}
        \begin{lemma}[Stability of the descending extension recursion]\label{lem:descending-extension-recursion}
            Fix $k$ and let $e_1,\ldots,e_k\ge0$ satisfy $e_s\le\varepsilon_L+\sum_{t=1}^{k-s}\binom{n-s}{t}3^t e_{s+t}$ for $1\le s\le k$. Then $e_s\le C_{k-s}n^{k-s}\varepsilon_L$, where the constants $C_0,\ldots,C_{k-1}$ depend only on $k$.
        \end{lemma}
        \begin{proof}
            Set $C_0:=1$ and, for $m\ge1$, $C_m:=1+\sum_{t=1}^{m}(3^t/t!)C_{m-t}$. The case $s=k$ follows directly from the assumed recursion. If the claim holds above level $s$ and $m:=k-s$, then $\binom{n-s}{t}\le n^t/t!$ gives
            \begin{equation*}
                e_s\le\varepsilon_L+\sum_{t=1}^{m}\frac{3^t}{t!}n^t C_{m-t}n^{m-t}\varepsilon_L\le C_m n^m\varepsilon_L.
            \end{equation*}
            This closes the descending induction.
        \end{proof}

        \begin{lemma}[Correctness, efficiency, and stability of the local $\chi$-recovery]
            \label{lem:chi-recovery-stability}
            Given a $k$-local Lindbladian $\mathcal L$, assume moreover access to $\varepsilon_L$-accurate PTM estimates uniformly for all $R\subseteq[n]$, $|R|\le k$, and all $\mathbf A_R,\mathbf B_R\in\mathcal P_R$:
            \begin{equation*}
                \left|
                \widehat L_{\mathbf A_R\otimes I_{\overline R},
                \mathbf B_R\otimes I_{\overline R}}
                -
                L_{\mathbf A_R\otimes I_{\overline R},
                \mathbf B_R\otimes I_{\overline R}}
                \right|
                \le \varepsilon_L .
            \end{equation*}
            Then Algorithm~\ref{alg:recover-chi} recovers all nonidentity global coefficients indexed by $\mathcal I_k^\circ$. More precisely, the recovered coefficients $\widehat{\chi}_{\mathbf{P},\mathbf{Q}}$ satisfy the following entrywise error bound:
            \begin{equation*}
                \varepsilon_\chi:=\max_{(\mathbf P,\mathbf Q)\in\mathcal I_k^\circ}\left|\widehat\chi_{\mathbf P,\mathbf Q}-\chi_{\mathbf P,\mathbf Q}\right|=\mathcal O_k(n^{k-1})\varepsilon_L.
            \end{equation*}
            The coefficient $\chi_{I,I}$ is determined by trace preservation and is excluded from this independent-parameter norm. Moreover, for fixed $k$, the number of elementary operations performed by Algorithm~\ref{alg:recover-chi} is $\mathcal{O}(n^k)$, with constants depending only on $k$.
        \end{lemma}
        \begin{proof}
            For each $R\subseteq[n]$, the local inversion formula gives
            \begin{equation*}
                \chi^{(R)}_{\mathbf P_R,\mathbf Q_R}
                =
                \frac{1}{2^{3|R|}}
                \sum_{\mathbf A_R,\mathbf B_R\in\mathcal P_R}
                L_{\mathbf A_R\otimes I_{\overline R},
                \mathbf B_R\otimes I_{\overline R}}
                \Tr\!\left(\mathbf A_R\mathbf Q_R\mathbf B_R\mathbf P_R\right).
            \end{equation*}
            Thus the first line of Algorithm~\ref{alg:recover-chi} exactly computes $\chi^{(R)}_{\mathbf P_R,\mathbf Q_R}$ when the PTM entries are exact. Next, by Equation \eqref{eq:chiRtochi}, for $R=\operatorname{supp}(\mathbf P)\cup\operatorname{supp}(\mathbf Q)$, we have
            \begin{equation*}
                \chi^{(R)}_{\mathbf P_R,\mathbf Q_R}
                =
                \sum_{\mathbf A_{\overline R}\in\mathcal P_{\overline R}}
                \chi_{\mathbf P_R\otimes\mathbf A_{\overline R},
                \mathbf Q_R\otimes\mathbf A_{\overline R}}.
            \end{equation*}
            Since all terms with $|R\cup\operatorname{supp}(\mathbf A_{\overline R})|>k$ vanish by locality, we have
            \begin{equation*}
                \chi_{\mathbf P_R\otimes I_{\overline R},
                \mathbf Q_R\otimes I_{\overline R}}
                =
                \chi^{(R)}_{\mathbf P_R,\mathbf Q_R}
                -
                \sum_{\substack{
                \mathbf A_{\overline R}\neq I_{\overline R}\\
                |R\cup\operatorname{supp}(\mathbf A_{\overline R})|\le k
                }}
                \chi_{\mathbf P_R\otimes\mathbf A_{\overline R},
                \mathbf Q_R\otimes\mathbf A_{\overline R}}.
            \end{equation*}
            Every coefficient appearing in the sum on the right-hand side has support union strictly larger than $R$. Hence processing the regions in decreasing order of $|R|$ gives a valid descending induction and recovers all coefficients with support union of size at most $k$. Note that the special case $\chi_{I,I}$ is given by
            \begin{equation*}
                \chi_{I,I}=-\sum_{\mathbf P\in\mathcal P_n\setminus\{I\}}\chi_{\mathbf P,\mathbf P}
            \end{equation*}
            This coefficient is determined afterwards by trace preservation. Its accumulated diagonal-sum error is not included in the $\mathcal O_k(n^{k-1})\varepsilon_L$ independent-coefficient bound, but $\chi_{I,I}$ is not needed to reconstruct $G$, $h$, or the GKSL generator.

            We now prove the error bound. From the local inversion  \eqref{eq:localinversionform}, using the uniform PTM bound and the Pauli trace orthogonality, one obtains
            \begin{equation*}
                \left|
                \widehat\chi^{(R)}_{\mathbf P_R,\mathbf Q_R}
                -
                \chi^{(R)}_{\mathbf P_R,\mathbf Q_R}
                \right|
                \le \varepsilon_L .
            \end{equation*}
            Indeed, for fixed Pauli strings $\mathbf P_R,\mathbf Q_R$, among the pairs $(\mathbf A_R,\mathbf B_R)$ only $4^{|R|}$ terms have nonzero trace $\Tr(\mathbf A_R\mathbf Q_R\mathbf B_R\mathbf P_R)$, and each nonzero trace has absolute value $2^{|R|}$. Hence
            \begin{equation*}
                \left|
                \widehat\chi^{(R)}_{\mathbf P_R,\mathbf Q_R}
                -
                \chi^{(R)}_{\mathbf P_R,\mathbf Q_R}
                \right|
                \le
                2^{-3|R|}
                \cdot 4^{|R|}
                \cdot 2^{|R|}
                \cdot \varepsilon_L
                =
                \varepsilon_L .
            \end{equation*}
            Let
            \begin{equation*}
                e_s
                :=
                \max_{\substack{\mathbf P,\mathbf Q\\
                |\operatorname{supp}(\mathbf P)\cup\operatorname{supp}(\mathbf Q)|=s}}
                \left|
                \widehat\chi_{\mathbf P,\mathbf Q}
                -
                \chi_{\mathbf P,\mathbf Q}
                \right|.
            \end{equation*}
            The recursive subtraction formula gives
            \begin{equation}\label{eq:inversion-error}
                e_s \le \varepsilon_L + \sum_{t=1}^{k-s} \binom{n-s}{t}3^t\, e_{s+t}\,,
            \end{equation}
            where the combinatorial factor counts the ways of placing $t$ balls into $n-s$ distinct boxes without repetition. Lemma~\ref{lem:descending-extension-recursion} applied to \eqref{eq:inversion-error} yields $e_s\le C_{k-s}n^{k-s}\varepsilon_L$. Since the independent coefficients used for $G$ and $h$ have $s\ge1$, their maximum error is $\mathcal O_k(n^{k-1})\varepsilon_L$. Finally, the number of regions $R\subseteq[n]$ with $|R|\le k$ scales as $\mathcal{O}(n^k)$. Moreover, for each such $R$, the number of local Pauli pairs and the cost of the local inversion are bounded by constants depending only on $k$. The subtraction step also involves at most
            \begin{equation*}
                \sum_{t=0}^{k-|R|}\binom{n-|R|}{t}3^t
            \end{equation*}
            extensions for a fixed pair, and summing over all $R$ with $|R|\le k$ gives $\mathcal{O}(n^k)$ operations for fixed $k$. This proves the claimed efficiency.
        \end{proof}
        \noindent As a direct corollary, we can reconstruct the matrix $G$ and vector $h$ from the PTM entries:
        \begin{corollary}\label{thm:dissipative-learning}
            With the notation of Lemma \ref{lem:chi-recovery-stability}, there exist efficiently computable estimators $\widehat{G}_{\mathbf{P},\mathbf{Q}}$ and $\widehat{h}_{\mathbf{Q}}$ such that
            \begin{equation}
                |\widehat{G}_{\mathbf{P},\mathbf{Q}} - G_{\mathbf{P},\mathbf{Q}}| =\mathcal O_k(n^{k-1})\varepsilon_L \qquad \text{and} \qquad |\widehat{h}_{\mathbf{P}} - h_{\mathbf{P}}| =\mathcal O_k(n^{k-1})\varepsilon_L.
            \end{equation}
            The $G$-bound holds for $\mathbf P,\mathbf Q\ne I$ with $|\supp(\mathbf P)\cup\supp(\mathbf Q)|\le k$, and the $h$-bound holds for $\mathbf P\in\mathcal P_{n,\le k}\setminus\{I\}$. The classical reconstruction requires $\mathcal O_k(n^k)$ arithmetic operations.
        \end{corollary}

        \begin{proof}
            The proof directly follows from Lemma \ref{lem:chi-recovery-stability} by setting $\widehat{G}_{\mathbf{P},\mathbf{Q}} = \widehat{\chi}_{\mathbf{P},\mathbf{Q}}$ and $\widehat{h}_{\mathbf{P}} = \frac{i}{2}\bigl(\widehat{\chi}_{\mathbf{P},I}-\widehat{\chi}_{I,\mathbf{P}}\bigr)$ as in Equation \eqref{eq:chiG}.
        \end{proof}

    \subsection{Thresholding before local inversion}
        \label{subsec:threshold-before-inversion}
        The following estimate avoids recursive noise amplification by thresholding the estimated local coefficient sums and only then applying the exact inverse in Lemma~\ref{lem:local-inclusion-exclusion}.

        For $\tau>0$, define the threshold support and clipped local sparsity by
        \begin{equation}\label{eq:clipped-local-sparsity}
            \Omega_\tau:=\{v\in\mathcal I_k^\circ:|\chi_v|>\tau\},
            \qquad
            r_\tau:=\max\left\{1,\max_{u\in[n]}
            \sum_{\substack{v\in\mathcal I_k^\circ:\,u\in S_v}}
            \min\left(1,\frac{|\chi_v|^2}{\tau^2}\right)\right\},
        \end{equation}
        and define the sitewise subthreshold tail
        \begin{equation}\label{eq:sitewise-chi-tail}
            \beta_\tau:=\max_{u\in[n]}
            \sum_{\substack{v\in\mathcal I_k^\circ\setminus\Omega_\tau:\,u\in S_v}}
            |\chi_v|.
        \end{equation}
        Unlike the diagonal-extension tail $\rho_\tau$ below, $\beta_\tau$ sums over all coefficients incident on a site.

        \begin{lemma}[Threshold-before-inversion estimate]
            \label{lem:threshold-before-inversion}
            Suppose $\nu>0$ and $\|\widehat y-y\|_\infty\le\nu$, and define
            \begin{equation*}
                \widetilde y_v
                :=\widehat y_v\mathbf 1_{\{|\widehat y_v|>2\nu\}},
                \qquad \widehat\chi:=\mathsf Z^{-1}\widetilde y.
            \end{equation*}
            For every $\tau>0$,
            \begin{equation}\label{eq:threshold-before-inversion-error}
                \|\widehat\chi-\chi\|_{\mathrm{loc},1}
                \le 3\cdot4^k r_\tau\nu+4^k\beta_\tau.
            \end{equation}
            Moreover,
            \begin{equation}\label{eq:threshold-before-inversion-degree}
                \deg(\widehat\chi)
                =\mathcal O_k\!\left(r_\tau+\frac{\beta_\tau}{\nu}\right).
            \end{equation}
            Mapping to $\widehat x=(\widehat h,\widehat G)=\Psi\widehat\chi$ and symmetrizing its Hamiltonian and dissipative coefficients preserve \eqref{eq:threshold-before-inversion-error} as a bound on $\|(2(\widehat h-h),\widehat G-G)\|_{\mathrm{loc},1}$ and increase the degree by at most a factor of two.
        \end{lemma}
        \begin{proof}
            Thresholding gives, for every coordinate,
            \begin{equation*}
                |\widetilde y_v-y_v|\le\min\{3\nu,|y_v|\},
                \qquad
                \mathbf1_{\{\widetilde y_v\ne0\}}
                \le\min\{1,|y_v|/\nu\}.
            \end{equation*}
            Indeed, if $|\widehat y_v|>2\nu$, the estimate is kept and has error at most $\nu<|y_v|$. Otherwise it is set to zero, giving error $|y_v|\le|\widehat y_v|+\nu\le3\nu$.

            Since $y=\mathsf Z\chi$, each $y_v$ is a sum of diagonal extensions. By the column bounds in Lemma~\ref{lem:local-inclusion-exclusion}, each coefficient contributes to at most $2^k$ such sums, whose supports are contained in its own. The triangle inequality and $\min\{a,\sum_j b_j\}\le\sum_j\min\{a,b_j\}$ for $a,b_j\ge0$ therefore give
            \begin{align*}
                \|\widetilde y-y\|_{\mathrm{loc},1}
                &\le2^k\max_{u\in[n]}\sum_{v:\,u\in S_v}
                \min\{3\nu,|\chi_v|\}\\
                &\le2^k(3\nu r_\tau+\beta_\tau),\\
                \deg(\widetilde y)
                &\le2^k\max_{u\in[n]}\sum_{v:\,u\in S_v}
                \min\{1,|\chi_v|/\nu\}\\
                &\le2^k\left(r_\tau+\frac{\beta_\tau}{\nu}\right).
            \end{align*}
            The last inequalities split the sums over $\Omega_\tau$ and its complement: at most $r_\tau$ coefficients above $\tau$ are incident on any site, while the remaining absolute values sum to at most $\beta_\tau$. Applying the norm and column bounds for $\mathsf Z^{-1}$ introduces a further factor $2^k$, proving \eqref{eq:threshold-before-inversion-error} and \eqref{eq:threshold-before-inversion-degree}.

            Finally, \eqref{eq:physical-chi-map-bounds} gives the physical error bound. Each $\chi$-coordinate contributes through $\Psi$ to at most one physical coordinate with the same support, so $\Psi$ does not increase degree. Replacing $\widehat h$ by $\Re\widehat h$ and $\widehat G$ by $(\widehat G+\widehat G^\dagger)/2$ cannot increase $\|(2(\widehat h-h),\widehat G-G)\|_{\mathrm{loc},1}$, since $h$ is real and $G$ is Hermitian. These operations increase degree by at most a factor of two.
        \end{proof}

    \subsection{Improved inversion under stable sparse diagonal extensions}
        \label{subsec:stable-sparse-diagonal-extensions}

        The guarded recursion below instead exploits sparsity of the diagonal extensions in the marginal identity. Without this restriction, up to $\mathcal O_k(n^{k-|S|})$ extensions can occur above a nonempty support $S$. We quantify their thresholded degree using $\Omega_\tau$ defined above. Let $\mathfrak F_k^\circ:=\{(S,\mathbf P_S,\mathbf Q_S):S\subseteq[n],\ 1\le|S|\le k,\ \mathbf P_S,\mathbf Q_S\in\mathcal P_S,\ \operatorname{supp}(\mathbf P_S)\cup\operatorname{supp}(\mathbf Q_S)=S\}$ be the set of triples used by the descending inversion. For $(S,\mathbf P_S,\mathbf Q_S)\in\mathfrak F_k^\circ$, define
        \begin{equation*}
            \operatorname{Ext}_{\Omega_\tau}(S,\mathbf P_S,\mathbf Q_S)
            :=
            \left\{
            \mathbf A_{\overline S}\in\mathcal P_{\overline S}\setminus\{I_{\overline S}\}:
            (\mathbf P_S\otimes\mathbf A_{\overline S},\mathbf Q_S\otimes\mathbf A_{\overline S})\in\Omega_\tau
            \right\}.
        \end{equation*}
        Then, we define the thresholded diagonal-extension degree by
        \begin{align}\label{eq:degreeomegatau}
            \mathfrak d_\tau
            :=\max_{(S,\mathbf P_S,\mathbf Q_S)\in\mathfrak F_k^\circ}
            \left|
            \operatorname{Ext}_{\Omega_\tau}(S,\mathbf P_S,\mathbf Q_S)
            \right|
        \end{align}
        and we denote by $D_\tau:=\sum_{\ell=0}^{k}\mathfrak d_\tau^\ell$ the associated recursive amplification factor. We also define the unresolved diagonal-extension tail
        \begin{align}\label{eq:tailunresolvedomegatau}
            \rho_\tau
            :=\max_{(S,\mathbf P_S,\mathbf Q_S)\in\mathfrak F_k^\circ}
            \sum_{\substack{\mathbf A_{\overline S}\in\mathcal P_{\overline S}\setminus\{I_{\overline S}\}:\\
            (\mathbf P_S\otimes\mathbf A_{\overline S},\mathbf Q_S\otimes\mathbf A_{\overline S})\notin\Omega_\tau}}
            \left|
            \chi_{\mathbf P_S\otimes\mathbf A_{\overline S},
            \mathbf Q_S\otimes\mathbf A_{\overline S}}
            \right|.
        \end{align}
        Thus $\mathfrak d_\tau$ controls the number of large diagonal extensions entering any local marginal, while $\rho_\tau$ controls the bias from discarded subthreshold extensions.

        \begin{rmk}\label{rmk:sparsity-comparison}
            For every $(S,\mathbf P_S,\mathbf Q_S)\in\mathfrak F_k^\circ$ and $u\in S$, each coefficient counted by $\operatorname{Ext}_{\Omega_\tau}(S,\mathbf P_S,\mathbf Q_S)$ contains $u$ and contributes $1$ to $r_\tau$. Hence $\mathfrak d_\tau\le r_\tau$. The restriction to $\mathfrak F_k^\circ$ ensures that $S=\operatorname{supp}(\mathbf P_S)\cup\operatorname{supp}(\mathbf Q_S)\ne\emptyset$, as required by the inversion. In the Hamiltonian setting, the coefficients are not recovered through a diagonal extension inversion at all: each Hamiltonian coefficient can be isolated from a single PTM entry. Indeed, the Lindblad $\chi$-matrix is supported only on coefficients of the form $(\mathbf R,I)$ and $(I,\mathbf R)$. Hence every nontrivial diagonal extension in the marginal identity vanishes, leaving no unresolved tail or recursive amplification.
        \end{rmk}

        For the complementary guarded recursion below, we assume that the unresolved diagonal-extension tail remains small after recursive amplification. We formalize this target-accuracy condition as follows.
        \begin{assumption}[Stable sparse diagonal extensions at accuracy $\varepsilon_\chi$]
            \label{ass:stable-sparse-diagonal-extensions}
            For the target accuracy $\varepsilon_{\chi}>0$ under consideration, there exists a threshold $\tau_\chi>0$ such that
            \begin{equation*}
                \tau_\chi+D_{\tau_\chi}\rho_{\tau_\chi}\le {\varepsilon_\chi}.
            \end{equation*}
        \end{assumption}
        For any finite coefficient array, choosing $\tau$ below the absolute value of every nonzero coefficient gives $\rho_\tau=0$, so the condition can be met at any accuracy. The useful requirement is that the corresponding amplification $D_\tau$ remains small, since it determines the sample cost. The guarded recursion below complements Theorem~\ref{thm:approximate-degree-learning} by recovering a threshold support under different bounds on the diagonal-extension degree and unresolved tail. When a candidate set is supplied, the recursion instead reduces the parameter-learning and classical postprocessing costs.

        We now describe the structure-learning version of the recursion.  The local inversion does not directly reveal a coefficient $\chi_{\mathbf P,\mathbf Q}$.  It reveals the marginal coefficient $\chi^{(S)}_{\mathbf P_S,\mathbf Q_S}$, where $S=\operatorname{supp}(\mathbf P)\cup\operatorname{supp}(\mathbf Q)$.  By \eqref{eq:chiRtochi}, this marginal is the desired coefficient plus all its nontrivial diagonal extensions,
        \begin{equation*}
            \chi^{(S)}_{\mathbf P_S,\mathbf Q_S}
            =
            \chi_{\mathbf P,\mathbf Q}
            +
            \sum_{\mathbf A_{\overline S}\neq I_{\overline S}}
            \chi_{\mathbf P_S\otimes\mathbf A_{\overline S},
            \mathbf Q_S\otimes\mathbf A_{\overline S}} .
        \end{equation*}
        Note that every nontrivial diagonal extension has strictly larger support union.  Hence we process coefficients from support size $k$ down to support size $1$.  Once a large coefficient is accepted, its estimate is propagated downward to all smaller marginals in which it appears as a diagonal extension.  This is a peeling procedure: subtract the already detected higher-level contamination, threshold what remains, and continue.

        The subtle point is that accepted coefficients are data-dependent.  A coefficient below the intended threshold may still be accepted because of statistical error or because it lies in a margin band.  Once accepted, it is propagated downward and contributes estimation error to later coefficients. Thus the proof cannot count only the true coefficients above the same threshold used for omitted tails.  We use a guard band around the decision threshold $\lambda$: accepted propagated terms are counted using the lower support $\Omega_{\lambda-\gamma}$, while omitted true terms are charged to the upper tail outside $\Omega_{\lambda+\gamma}$.

        \begin{algorithm}[ht!]
            \caption{Guarded sparse local $\chi$-recovery}
            \label{alg:sparse-diagonal-chi-recovery}
            \begin{algorithmic}[1]
                \REQUIRE Local PTM estimates
                \begin{equation*}
                    \widehat L_{\mathbf A_S\otimes I_{\overline S},\mathbf B_S\otimes I_{\overline S}}
                \end{equation*}
                for all nonempty $S\subseteq[n]$, $|S|\le k$, and all $\mathbf A_S,\mathbf B_S\in\mathcal P_S$, and decision threshold $\lambda>0$.
                \ENSURE An estimated threshold support and estimates for accepted coefficients.

                \STATE For every nonempty $S\subseteq[n]$, $|S|\le k$, and every $\mathbf P_S,\mathbf Q_S\in\mathcal P_S$, compute
                \begin{equation*}
                    \widehat\chi^{(S)}_{\mathbf P_S,\mathbf Q_S}
                    :=
                    2^{-3|S|}
                    \sum_{\mathbf A_S,\mathbf B_S\in\mathcal P_S}
                    \widehat L_{\mathbf A_S\otimes I_{\overline S},\mathbf B_S\otimes I_{\overline S}}
                    \Tr(\mathbf A_S\mathbf Q_S\mathbf B_S\mathbf P_S).
                \end{equation*}
                \STATE Initialize $\widehat\Omega_\lambda:=\varnothing$, and initialize all correction registers
                \begin{equation*}
                    \mathsf R(S,\mathbf P_S,\mathbf Q_S):=0.
                \end{equation*}
                \FOR{$s=k,k-1,\ldots,1$}
                    \FOR{every $(\mathbf P,\mathbf Q)\in\mathcal I_k^\circ$ with $|\operatorname{supp}(\mathbf P)\cup\operatorname{supp}(\mathbf Q)|=s$}
                        \STATE Set $S:=\operatorname{supp}(\mathbf P)\cup\operatorname{supp}(\mathbf Q)$, and write
                        \begin{equation*}
                            \mathbf P=\mathbf P_S\otimes I_{\overline S},
                            \qquad
                            \mathbf Q=\mathbf Q_S\otimes I_{\overline S}.
                        \end{equation*}
                        \STATE Define
                        \begin{equation*}
                            \widetilde\chi_{\mathbf P,\mathbf Q}
                            :=
                            \widehat\chi^{(S)}_{\mathbf P_S,\mathbf Q_S}
                            -
                            \mathsf R(S,\mathbf P_S,\mathbf Q_S).
                        \end{equation*}
                        \IF{$|\widetilde\chi_{\mathbf P,\mathbf Q}|> \lambda$}
                            \STATE Add $(\mathbf P,\mathbf Q)$ to $\widehat\Omega_\lambda$, and set
                            \begin{equation*}
                                \widehat\chi_{\mathbf P,\mathbf Q}:=\widetilde\chi_{\mathbf P,\mathbf Q}.
                            \end{equation*}
                            \STATE For every $\emptyset \neq T\subsetneq S$ such that $\mathbf P_{S\setminus T}=\mathbf Q_{S\setminus T}$, update
                            \begin{equation*}
                                \mathsf R(T,\mathbf P_T,\mathbf Q_T)
                                \leftarrow
                                \mathsf R(T,\mathbf P_T,\mathbf Q_T)+\widehat\chi_{\mathbf P,\mathbf Q}.
                            \end{equation*}
                        \ENDIF
                    \ENDFOR
                \ENDFOR
                \RETURN $\widehat\Omega_\lambda$ and $\{\widehat\chi_{\mathbf P,\mathbf Q}\}_{(\mathbf P,\mathbf Q)\in\widehat\Omega_\lambda}$.
            \end{algorithmic}
        \end{algorithm}

        \begin{lemma}[Guarded sparse inversion and error propagation]
            \label{lem:sparse-diagonal-inversion}
            Fix a decision threshold $\lambda>0$ and a margin $0<\gamma<\lambda$. Set
            \begin{equation*}
                \Omega_-:=\Omega_{\lambda-\gamma},
                \qquad
                \Omega_+:=\Omega_{\lambda+\gamma},
                \qquad
                \mathfrak d_-:=\mathfrak d_{\Omega_-},
                \qquad
                D_-:=\sum_{\ell=0}^k \mathfrak d_-^\ell,
                \qquad
                \rho_+:=\rho_{\Omega_+}.
            \end{equation*}
            Assume access to $\varepsilon_L$-accurate PTM estimates,
            \begin{equation*}
                \left|
                \widehat L_{\mathbf A_S\otimes I_{\overline S},\mathbf B_S\otimes I_{\overline S}}
                -
                L_{\mathbf A_S\otimes I_{\overline S},\mathbf B_S\otimes I_{\overline S}}
                \right|
                \le \varepsilon_L
            \end{equation*}
            for every nonempty $S\subseteq[n]$, $|S|\le k$, and all $\mathbf A_S,\mathbf B_S\in\mathcal P_S$, and suppose
            \begin{equation*}
                D_-(\varepsilon_L+\rho_+)\le \gamma.
            \end{equation*}
            Then Algorithm~\ref{alg:sparse-diagonal-chi-recovery}, run with decision threshold $\lambda$, outputs $\widehat\Omega_\lambda$ satisfying
            \begin{equation*}
                \Omega_+\subseteq \widehat\Omega_\lambda\subseteq \Omega_-.
            \end{equation*}
            Equivalently,
            \begin{equation*}
                |\chi_{\mathbf P,\mathbf Q}|> \lambda+\gamma
                \quad\Longrightarrow\quad
                (\mathbf P,\mathbf Q)\in\widehat\Omega_\lambda,
            \end{equation*}
            and
            \begin{equation*}
                |\chi_{\mathbf P,\mathbf Q}|\le \lambda-\gamma
                \quad\Longrightarrow\quad
                (\mathbf P,\mathbf Q)\notin\widehat\Omega_\lambda .
            \end{equation*}
            Moreover, every accepted coefficient estimate satisfies
            \begin{equation*}
                \left|
                \widehat\chi_{\mathbf P,\mathbf Q}
                -
                \chi_{\mathbf P,\mathbf Q}
                \right|
                \le \gamma .
            \end{equation*}
            The implementation of Algorithm~\ref{alg:sparse-diagonal-chi-recovery} requires $\mathcal O_k(n^k+2^k|\widehat\Omega_\lambda|)=\mathcal O_k(n^k)$ classical operations.
        \end{lemma}

        \begin{proof}
            Algorithm~\ref{alg:sparse-diagonal-chi-recovery} proceeds by descending support size. It first computes all local marginal coefficients from the estimated local PTM entries. For a coefficient $(\mathbf P,\mathbf Q)$, with union support $S$, the marginal identity expresses the marginal coefficient as the desired coefficient plus its diagonal extensions. Since every nontrivial diagonal extension has strictly larger support union, these extensions have already been processed when the algorithm reaches level $|S|$. More precisely, for every $S$, we recall the local inversion formula \eqref{eq:localinversionform}
            \begin{equation*}
                \chi^{(S)}_{\mathbf P_S,\mathbf Q_S}
                =
                2^{-3|S|}
                \sum_{\mathbf A_S,\mathbf B_S\in\mathcal P_S}
                L_{\mathbf A_S\otimes I_{\overline S},\mathbf B_S\otimes I_{\overline S}}
                \Tr(\mathbf A_S\mathbf Q_S\mathbf B_S\mathbf P_S).
            \end{equation*}
            For fixed $\mathbf P_S,\mathbf Q_S$, exactly $4^{|S|}$ summands may have nonzero trace, and each nonzero trace has absolute value $2^{|S|}$. Hence
            \begin{equation*}
                \left|
                \widehat\chi^{(S)}_{\mathbf P_S,\mathbf Q_S}
                -
                \chi^{(S)}_{\mathbf P_S,\mathbf Q_S}
                \right|
                \le
                2^{-3|S|}4^{|S|}2^{|S|}\varepsilon_L
                =
                \varepsilon_L.
            \end{equation*}
            Let $e_s$ denote the worst error of a recursively corrected estimate $\widetilde\chi_{\mathbf P,\mathbf Q}$ among coefficients whose support union has size $s$, at the moment those coefficients are processed.  We prove by descending induction that all coefficients above $\lambda+\gamma$ at levels larger than $s$ have been accepted, every accepted coefficient at levels larger than $s$ belongs to $\Omega_-$, and their estimates have error at most $\max_{s'>s}e_{s'}$.

            Now fix $(\mathbf P,\mathbf Q)\in\mathcal I_k^\circ$, and set
            \begin{equation*}
                S:=\operatorname{supp}(\mathbf P)\cup\operatorname{supp}(\mathbf Q).
            \end{equation*}
            The marginal identity \eqref{eq:chiRtochi} yields
            \begin{equation*}
                \chi^{(S)}_{\mathbf P_S,\mathbf Q_S}
                =
                \chi_{\mathbf P,\mathbf Q}
                +
                \sum_{\mathbf A_{\overline S}\neq I_{\overline S}}
                \chi_{\mathbf P_S\otimes\mathbf A_{\overline S},
                \mathbf Q_S\otimes\mathbf A_{\overline S}}.
            \end{equation*}
            Subtracting the exact identity from the corrected estimate gives
            \begin{equation*}
                \begin{aligned}
                    \widetilde\chi_{\mathbf P,\mathbf Q}-\chi_{\mathbf P,\mathbf Q}
                    &=
                    \left(\widehat\chi^{(S)}_{\mathbf P_S,\mathbf Q_S}
                    -
                    \chi^{(S)}_{\mathbf P_S,\mathbf Q_S}\right)\\
                    &\quad+
                    \sum_{\substack{\mathbf A_{\overline S}\neq I_{\overline S}:\\
                    (\mathbf P_S\otimes\mathbf A_{\overline S},
                    \mathbf Q_S\otimes\mathbf A_{\overline S})\notin\widehat\Omega_\lambda}}
                    \chi_{\mathbf P_S\otimes\mathbf A_{\overline S},
                    \mathbf Q_S\otimes\mathbf A_{\overline S}}\\
                    &\quad-
                    \sum_{\substack{\mathbf A_{\overline S}\neq I_{\overline S}:\\
                    (\mathbf P_S\otimes\mathbf A_{\overline S},
                    \mathbf Q_S\otimes\mathbf A_{\overline S})\in\widehat\Omega_\lambda}}
                    \left(
                    \widehat\chi_{\mathbf P_S\otimes\mathbf A_{\overline S},
                    \mathbf Q_S\otimes\mathbf A_{\overline S}}
                    -
                    \chi_{\mathbf P_S\otimes\mathbf A_{\overline S},
                    \mathbf Q_S\otimes\mathbf A_{\overline S}}
                    \right).
                \end{aligned}
            \end{equation*}
            By the induction invariant, every omitted diagonal extension is outside $\Omega_+$, because all higher-level coefficients in $\Omega_+$ have already been accepted.  Hence the middle sum has absolute value at most $\rho_+$.  Also by the induction invariant, every accepted diagonal extension lies in $\Omega_-$. There are at most $\mathfrak d_-$ such extensions, and each propagated estimate has error at most $\max_{s'>s}e_{s'}$.  Therefore
            \begin{equation*}
                e_s\le \varepsilon_L+\rho_++\mathfrak d_-\max_{s'>s}e_{s'}.
            \end{equation*}
            Iterating over at most $k$ levels gives for all $s$
            \begin{equation*}
                e_s\le D_-(\varepsilon_L+\rho_+)\le \gamma.
            \end{equation*}
            The thresholding guarantees close the induction.  If $(\mathbf P,\mathbf Q)$ is accepted, then
            \begin{equation*}
                |\chi_{\mathbf P,\mathbf Q}|
                \ge
                |\widetilde\chi_{\mathbf P,\mathbf Q}|-\gamma
                >
                \lambda-\gamma,
            \end{equation*}
            so the accepted coefficient belongs to $\Omega_-$.  Conversely, if $|\chi_{\mathbf P,\mathbf Q}|>\lambda+\gamma$, then
            \begin{equation*}
                |\widetilde\chi_{\mathbf P,\mathbf Q}|
                \ge
                |\chi_{\mathbf P,\mathbf Q}|-\gamma
                >
                \lambda,
            \end{equation*}
            so the algorithm accepts it.  Thus $\Omega_+\subseteq\widehat\Omega_\lambda\subseteq\Omega_-$.  The accepted coefficient error bound is exactly $e_s\le\gamma$. Finally, the exhaustive scan over $\mathcal I_k^\circ$ costs $\mathcal O_k(n^k)$ classical operations, and each accepted coefficient updates at most $2^k$ correction registers. The classical postprocessing cost is therefore $\mathcal O_k(n^k+2^k|\widehat\Omega_\lambda|)$.
        \end{proof}
        \begin{corollary}[Threshold structure learning]
            \label{cor:threshold-structure-learning}
            Fix a decision threshold $\lambda>0$ and a margin $0<\gamma<\lambda$. Let
            \begin{equation*}
                D_-:=D_{\lambda-\gamma},
                \qquad
                \rho_+:=\rho_{\lambda+\gamma}.
            \end{equation*}
            Suppose the learner is supplied with a bound $B\ge D_-$ and assume
            \begin{equation*}
                B\rho_+\le \frac{\gamma}{2}.
            \end{equation*}
            Obtain PTM estimates with accuracy $\varepsilon_L\le\gamma/(2B)$. Then Algorithm~\ref{alg:sparse-diagonal-chi-recovery}, run with decision threshold $\lambda$, outputs an estimated threshold support $\widehat\Omega_\lambda\subseteq\mathcal I_k^\circ$ which is correct up to a $\gamma$-margin around $\lambda$:
            \begin{equation*}
                |\chi_{\mathbf P,\mathbf Q}|> \lambda+\gamma
                \quad\Longrightarrow\quad
                (\mathbf P,\mathbf Q)\in\widehat\Omega_\lambda,
            \end{equation*}
            and
            \begin{equation*}
                |\chi_{\mathbf P,\mathbf Q}|\le \lambda-\gamma
                \quad\Longrightarrow\quad
                (\mathbf P,\mathbf Q)\notin\widehat\Omega_\lambda .
            \end{equation*}
            The classical postprocessing time is $\mathcal O_k(n^k)$. Using the process-shadow derivative-estimation routine of Section \ref{sec:learning-ptm-elements} to obtain the required PTM estimates, the number of samples is
            \begin{equation*}
                \widetilde{\mathcal O}_k\!\left(\frac{\bar\alpha^2B^2}{\gamma^2}\log\frac{n}{\delta}
                \right),
            \end{equation*}
            Here the hidden factors are polylogarithmic in $B/\gamma$ and $n$ and depend on $k$. In the exact threshold case $\rho_+=0$, so the tail promise is automatic. If a constant certified bound $B$ is supplied, this becomes
            \begin{equation*}
                \widetilde{\mathcal O}_k\!\left(
                \bar\alpha^2\gamma^{-2}\log\frac{n}{\delta}
                \right)
            \end{equation*}
            samples.
        \end{corollary}
        \begin{proof}
            Since $D_-\le B$, the assumptions imply
            \begin{equation*}
                D_-(\varepsilon_L+\rho_+)
                \le B\varepsilon_L+B\rho_+
                \le\frac{\gamma}{2}+\frac{\gamma}{2}
                =\gamma.
            \end{equation*}
            The result follows from Lemma~\ref{lem:sparse-diagonal-inversion}. The classical postprocessing cost is that of the exhaustive scan in Algorithm~\ref{alg:sparse-diagonal-chi-recovery}. Choosing $\varepsilon_L=\Theta(\gamma/B)$ in the PTM derivative-estimation routine and union bounding over $\mathcal O_k(n^k)$ local PTM pairs gives the stated quantum sample bound.
        \end{proof}

        \begin{rmk}[When the tail term vanishes]\label{rem:rho-plus-zero}
            The condition $D_-\rho_+\le\gamma/2$ is a robustness condition, not a sampling condition.  It says that the true diagonal extensions omitted by thresholding cannot, after recursive amplification, move a coefficient across the margin. In exact sparse or exact local models this term can vanish.  For instance, if there is a threshold gap such that every coefficient is either zero or has absolute value larger than $\lambda+\gamma$, then every coefficient outside $\Omega_+$ is zero and therefore $\rho_+=0$.  Similarly, if the generator is known to be supported on a local hypergraph $E$ and we run the supplied-candidate version on the corresponding candidate set $\Omega_E$, then all coefficients outside $\Omega_E$ vanish and the analogous tail $\rho_{\Omega_E}$ is zero.  In these exact settings the remaining quantity $D_-$ only controls the branching of accepted diagonal extensions.  For bounded-degree graphs or hypergraphs, $D_-$ is bounded in terms of the local degree and constants depending on $k$.
        \end{rmk}

        \noindent Whenever a candidate set $\Omega$ of $\chi$-coefficient indices is known, it becomes unnecessary to run an exhaustive search and the computational cost only follows from performing parameter inversion. In this case and under the sparsity condition, the latter can be much improved: more precisely, suppose that a candidate set $\Omega\subseteq\mathcal I_k^\circ$ is given. Assume that $\Omega$ contains the indices of all coefficients above threshold,
        \begin{equation*}
            |\chi_{\mathbf P,\mathbf Q}|>\tau\quad\Longrightarrow\quad(\mathbf P,\mathbf Q)\in\Omega .
        \end{equation*}
        Let $\mathfrak d_\Omega$ and $\rho_\Omega$ be the corresponding diagonal-extension degree and unresolved diagonal-extension tail, respectively, defined by replacing $\Omega_\tau$ with $\Omega$ in Equations~\eqref{eq:degreeomegatau} and \eqref{eq:tailunresolvedomegatau}. Let $D_\Omega:=\sum_{\ell=0}^k\mathfrak d_\Omega^\ell$. Assume that, for a target accuracy $\varepsilon_\chi>0$,
        \begin{equation*}
            \tau+D_\Omega\rho_\Omega\le\varepsilon_\chi .
        \end{equation*}
        Run the sparse recursive inversion restricted to $\Omega$, define $\widehat\chi_{\mathbf P,\mathbf Q}:=0$ for $(\mathbf P,\mathbf Q)\notin\Omega$, and use the returned estimates on $\Omega$. The supplied-candidate recursion restricts the diagonal-extension procedure to the graph induced by $\Omega$. Its PTM-query list and $\mathcal O_k(|\Omega|)$ implementation are specified in the proof below.
        \begin{corollary}
            \label{cor:dissipative-learning-stable-sparse}
            Under the supplied-candidate assumption and protocol of the previous paragraph, for all $(\mathbf P,\mathbf Q)\in\mathcal I_k^\circ$,
            \begin{equation*}
                \left|
                \widehat\chi_{\mathbf P,\mathbf Q}
                -
                \chi_{\mathbf P,\mathbf Q}
                \right|
                =
                \mathcal O\left(
                \varepsilon_\chi+D_\Omega\varepsilon_L
                \right).
            \end{equation*}
            Consequently, the estimators defined in Equation \eqref{eq:chiG} satisfy
            \begin{equation*}
                \left|
                \widehat G_{\mathbf P,\mathbf Q}
                -
                G_{\mathbf P,\mathbf Q}
                \right|
                =
                \mathcal O\left(
                \varepsilon_\chi+D_\Omega\varepsilon_L
                \right),
                \qquad \text{ and }\qquad
                \left|
                \widehat h_{\mathbf P}
                -
                h_{\mathbf P}
                \right|
                =
                \mathcal O\left(
                \varepsilon_\chi+D_\Omega\varepsilon_L
                \right).
            \end{equation*}
            The classical postprocessing time required to compute these estimates, once $\Omega$ is given, is $ \mathcal O_k(|\Omega|)$. Consequently, choosing $\varepsilon_L\le \varepsilon_\chi/D_\Omega$ and estimating the required PTM entries using the process-shadow derivative-estimation routine of Section \ref{sec:learning-ptm-elements} yields entrywise error $\mathcal O(\varepsilon_\chi)$ with
            \begin{equation*}
                \widetilde{\mathcal O}_k\!\left(
                \frac{\bar\alpha^2D_\Omega^2}{\varepsilon_\chi^2}
                \log\frac{|\Omega|}{\delta}
                \right)
            \end{equation*}
            samples, up to polylogarithmic factors in $D_\Omega/\varepsilon_\chi$ and $|\Omega|$.
        \end{corollary}
        \begin{proof}
            If $\Omega=\varnothing$, every independent $\chi$-coefficient has absolute value at most $\tau\le\varepsilon_\chi$, so the zero estimate satisfies the stated bounds by~\eqref{eq:chiG}. Assume henceforth that $\Omega\ne\varnothing$.

            Let the vertices of the directed diagonal-extension graph be the elements of $\Omega$. For $v=(\mathbf P,\mathbf Q)$ with $S_v:=\supp(\mathbf P)\cup\supp(\mathbf Q)$, draw an edge
            \begin{equation*}
                v\longrightarrow u:=(\mathbf P_T,\mathbf Q_T)
            \end{equation*}
            whenever $T\subsetneq S_v$, $\mathbf P_{S_v\setminus T}=\mathbf Q_{S_v\setminus T}$, and $u\in\Omega$. Thus $v$ is a diagonal extension contributing to the marginal associated with $u$. For every $u=(\mathbf P,\mathbf Q)$, write $S_u:=\supp(\mathbf P)\cup\supp(\mathbf Q)$, compute $\widehat\chi^{(S_u)}_{\mathbf P_{S_u},\mathbf Q_{S_u}}$ from its $\mathcal O_k(1)$ local PTM entries, and initialize a correction register $R_u=0$. Processing vertices in decreasing order of $|S_u|$, set
            \begin{equation*}
                \widehat\chi_u
                :=
                \widehat\chi^{(S_u)}_{\mathbf P_{S_u},\mathbf Q_{S_u}}-R_u,
            \end{equation*}
            and add $\widehat\chi_u$ to the register of every out-neighbour of $u$. Since every edge strictly decreases the support size, this recursion is well defined.

            Each local marginal has error at most $\varepsilon_L$, while the omitted extensions outside $\Omega$ contribute at most $\rho_\Omega$. If $e_s$ denotes the largest error among coefficients in $\Omega$ with support size $s$, the descending recursion therefore gives
            \begin{equation*}
                e_s
                \le
                \varepsilon_L+\rho_\Omega
                +\mathfrak d_\Omega\max_{t>s}e_t.
            \end{equation*}
            Iterating over at most $k$ support levels yields
            \begin{equation*}
                \max_{u\in\Omega}|\widehat\chi_u-\chi_u|
                \le
                D_\Omega(\varepsilon_L+\rho_\Omega)
                \le
                D_\Omega\varepsilon_L+\varepsilon_\chi.
            \end{equation*}
            Outside $\Omega$, the thresholding assumption gives $|\chi_{\mathbf P,\mathbf Q}|\leq\tau\leq\varepsilon_\chi$. Since these estimates are set to zero, the claimed bound holds on all of $\mathcal I_k^\circ$. The corresponding bounds for $G$ and $h$ follow from their linear definitions in terms of $\chi$.

            Finally, each vertex has at most $2^{|S_u|}\leq2^k$ out-neighbours. Hence the graph contains at most $2^k|\Omega|$ edges, while the local Fierz inversion uses $\mathcal O_k(1)$ PTM entries per vertex. The number of requested PTM entries and the classical reconstruction cost are therefore both $\mathcal O_k(|\Omega|)$. Choosing $\varepsilon_L=\Theta(\varepsilon_\chi/D_\Omega)$ in the PTM derivative-estimation bound and union bounding over these $\mathcal O_k(|\Omega|)$ entries gives the stated sample complexity.
        \end{proof}

        Next, we consider two representative regimes in which Assumption~\ref{ass:stable-sparse-diagonal-extensions} holds.

        \begin{corollary}[Exact support on bounded-intersection hypergraphs]
            \label{cor:bounded-intersection-exact-support}
            Fix $k$, and let $E$ be a collection of subsets $e\subseteq[n]$ with $|e|\le k$.  Define
            \begin{equation*}
                \Omega_E:=
                \left\{
                (\mathbf P,\mathbf Q)\in\mathcal I_k^\circ:
                \operatorname{supp}(\mathbf P)\cup\operatorname{supp}(\mathbf Q)\subseteq e
                \text{ for some }e\in E
                \right\}.
            \end{equation*}
            Assume exact independent-coefficient support on $E$, i.e.\ $\chi_{\mathbf P,\mathbf Q}=0$ for all $(\mathbf P,\mathbf Q)\in\mathcal I_k^\circ\setminus\Omega_E$. If $E$ has intersection degree at most $\Delta$,
            \begin{equation*}
                \max_{e\in E}
                \left|
                \{e'\in E:\ e'\neq e,\ e'\cap e\neq\emptyset\}
                \right|
                \le \Delta,
            \end{equation*}
            then the supplied-candidate algorithm of Corollary~\ref{cor:dissipative-learning-stable-sparse}, run with $\Omega=\Omega_E$, learns all entries of $G$ and $h$ to entrywise error $\mathcal O(\varepsilon_\chi)$ using
            \begin{equation*}
                \widetilde{\mathcal O}_k\!\left(
                \bar\alpha^2D_{k,\Delta}^2\varepsilon_\chi^{-2}
                \log\frac{|\Omega_E|}{\delta}
                \right),
                \qquad
                D_{k,\Delta}:=
                \sum_{\ell=0}^k
                \left((\Delta+1)(4^k-1)\right)^\ell,
            \end{equation*}
            samples and $\mathcal O_k(|\Omega_E|)$ classical postprocessing time.  In particular, for fixed $k$ and bounded $\Delta$,
            \begin{equation*}
                \widetilde{\mathcal O}_{k,\Delta}\!\left(
                \bar\alpha^2\varepsilon_\chi^{-2}\log\frac{|\Omega_E|}{\delta}
                \right)
            \end{equation*}
            samples after the interaction regions $E$ are supplied.
        \end{corollary}

        \begin{proof}
            The recovery algorithm only uses nonempty exact supports. The case $S=\emptyset$ concerns the identity component and is irrelevant for the recovery of $G$ and $h$, as noted above.  Fix $(S,\mathbf P_S,\mathbf Q_S)\in\mathfrak F_k^\circ$.  If $\operatorname{Ext}_{\Omega_E}(S,\mathbf P_S,\mathbf Q_S)$ is nonempty, then $S\subseteq e_0$ for some $e_0\in E$.  Every edge $e\in E$ containing $S$ intersects $e_0$, and therefore there are at most $\Delta+1$ such edges.  For each edge $e$, the number of nonidentity diagonal extensions supported in $e\setminus S$ is at most $4^{|e\setminus S|}-1\le 4^k-1$. Thus the relevant diagonal-extension degree is bounded by $(\Delta+1)(4^k-1)$, and the corresponding recursive amplification is at most $D_{k,\Delta}$.  Exact support gives $\rho_{\Omega_E}=0$, while $\Omega_E$ contains every nonzero coefficient, so the threshold condition of Corollary~\ref{cor:dissipative-learning-stable-sparse} holds, for instance, with $\tau=\varepsilon_\chi/2$.  The claimed error, runtime, and sample bounds follow by applying that corollary.
        \end{proof}

        \begin{corollary}[Exact unknown-support recovery on a gapped bounded-intersection hypergraph]
            \label{cor:gapped-unknown-bounded-intersection}
            Fix $\lambda>0$, $0<\gamma<\lambda$, and $\delta\in(0,1)$. Suppose the learner knows $\Delta$ and that there exists an unknown $k$-local hypergraph of intersection degree at most $\Delta$ describing the support of the considered Lindbladian. Assume also that every nonzero coefficient has absolute value greater than $\lambda+\gamma$. Define
            \begin{equation*}
                B_{k,\Delta}:=\sum_{\ell=0}^k\bigl((\Delta+1)(4^k-1)\bigr)^\ell,
            \end{equation*}
            then, with probability at least $1-\delta$, the guarded algorithm recovers the exact support $\Omega_*:=\{(\mathbf P,\mathbf Q)\in\mathcal I_k^\circ:\chi_{\mathbf P,\mathbf Q}\ne0\}$. It uses $\widetilde{\mathcal O}_k(\bar\alpha^2B_{k,\Delta}^2\gamma^{-2}\log(n/\delta))$ quantum samples, $\widetilde{\mathcal O}_k(\bar\alpha B_{k,\Delta}^2\gamma^{-2}\log(n/\delta))$ total evolution time, and $\mathcal O_k(n^k)$ classical postprocessing time.
        \end{corollary}
        \begin{proof}
            The gap gives $\Omega_{\lambda-\gamma}=\Omega_{\lambda+\gamma}=\Omega_*$ and $\rho_{\lambda+\gamma}=0$. The extension-count argument in the proof of Corollary~\ref{cor:bounded-intersection-exact-support} yields $D_-\le B_{k,\Delta}$. Corollary~\ref{cor:threshold-structure-learning} with $B=B_{k,\Delta}$ proves the claim.
        \end{proof}

        \begin{corollary}[Algebraically decaying diagonal-extension tails]
            \label{cor:algebraic-diagonal-extension-tails}
            Fix $k$.  Suppose that there are constants $C>0$ and $p>k+1$ such that, for every $(S,\mathbf P_S,\mathbf Q_S)\in\mathfrak F_k^\circ$, the nonincreasing rearrangement $(a_j)_{j\ge1}$ of the absolute values $ \left\{ \left| \chi_{\mathbf P_S\otimes\mathbf A_{\overline S}, \mathbf Q_S\otimes\mathbf A_{\overline S}} \right|: \mathbf A_{\overline S}\neq I_{\overline S} \right\}$ satisfies $a_j\le Cj^{-p}$.  Let $\overline C_p:=C\left(2+\frac{1}{p-1}\right)$ and, for a target accuracy $\varepsilon_\chi>0$, choose
            \begin{equation*}
                b_\chi
                :=
                \left\lceil
                \max\left\{
                1,
                \left(\frac{4C}{\varepsilon_\chi}\right)^{1/p},
                \left(
                \frac{2(k+1)\overline C_p}{\varepsilon_\chi}
                \right)^{\frac{1}{p-k-1}}
                \right\}
                \right\rceil
            \end{equation*}
            and a threshold $\tau_{b_\chi}$ satisfying
            \begin{equation*}
                C(b_\chi+1)^{-p}<\tau_{b_\chi}\le 2C(b_\chi+1)^{-p},
            \end{equation*}
            and suppose the threshold support $\Omega_{\tau_{b_\chi}}$ is supplied.  Then the supplied-candidate algorithm of Corollary~\ref{cor:dissipative-learning-stable-sparse}, run with $\Omega=\Omega_{\tau_{b_\chi}}$, learns all entries of $G$ and $h$ to entrywise error $\mathcal O(\varepsilon_\chi)$ using
            \begin{equation*}
                \widetilde{\mathcal O}_k\!\left(
                \bar\alpha^2
                \left(\sum_{\ell=0}^k b_\chi^\ell\right)^2
                \varepsilon_\chi^{-2}
                \log\frac{|\Omega_{\tau_{b_\chi}}|}{\delta}
                \right)
                =
                \widetilde{\mathcal O}_{k,p,C}\!\left(
                \bar\alpha^2\varepsilon_\chi^{-2-\frac{2k}{p-k-1}}
                \log\frac{|\Omega_{\tau_{b_\chi}}|}{\delta}
                \right)
            \end{equation*}
            samples and $\mathcal O_k(|\Omega_{\tau_{b_\chi}}|)$ classical postprocessing time.
        \end{corollary}
        The qualifier that $\Omega_{\tau_{b_\chi}}$ is supplied is essential for this corollary. Its ranked diagonal-extension tail condition does not by itself bound the sitewise tail $\beta_\tau$ required by Theorem~\ref{thm:approximate-degree-learning}. A support-free conclusion therefore requires either an explicit conversion to that promise or the guarded condition $D_{\lambda-\gamma}\rho_{\lambda+\gamma}\le\gamma/2$.

        \begin{proof}
            For any fixed $(S,\mathbf P_S,\mathbf Q_S)\in\mathfrak F_k^\circ$, the choice $\tau_{b_\chi}>C(b_\chi+1)^{-p}$ implies that every coefficient above threshold must appear among the first $b_\chi$ terms of the rearrangement. Hence $\mathfrak d_{\tau_{b_\chi}}\le b_\chi$ and
            \begin{equation*}
                D_{\tau_{b_\chi}}
                \le
                \sum_{\ell=0}^k b_\chi^\ell .
            \end{equation*}
            The unresolved tail outside the threshold support obeys
            \begin{equation*}
                \rho_{\tau_{b_\chi}}
                \le
                \sum_{j>b_\chi} Cj^{-p}+b_\chi\tau_{b_\chi}
                \le
                \left(\frac{C}{p-1}+2C\right)b_\chi^{1-p}
                =
                \overline C_p b_\chi^{1-p}.
            \end{equation*}
            Therefore
            \begin{equation*}
                D_{\tau_{b_\chi}}\rho_{\tau_{b_\chi}}
                \le
                (k+1)b_\chi^k\overline C_p b_\chi^{1-p}
                =
                (k+1)\overline C_p b_\chi^{k+1-p}
                \le
                \frac{\varepsilon_\chi}{2}.
            \end{equation*}
            Moreover,
            \begin{equation*}
                \tau_{b_\chi}
                \le
                2C b_\chi^{-p}
                \le
                \frac{\varepsilon_\chi}{2}.
            \end{equation*}
            Thus $\tau_{b_\chi}+D_{\tau_{b_\chi}}\rho_{\tau_{b_\chi}}\le\varepsilon_\chi$, so Corollary~\ref{cor:dissipative-learning-stable-sparse} applies.  The final display follows from $\sum_{\ell=0}^k b_\chi^\ell=\mathcal O_k(b_\chi^k)$ and the definition of $b_\chi$.
        \end{proof}

\section{Learning the PTM elements}\label{sec:learning-ptm-elements}
    Next, we argue how to efficiently estimate the coefficients $L_{\mathbf{P},\mathbf{Q}}$.

    \subsection{Support growth and polynomial approximation}
        For operators $X \in \mathcal{B}(\mathcal{H})$, the Frobenius norm is defined as: \begin{equation*} \|X\|_{2} := \left(\Tr[X^{\dagger}X]\right)^{1/2} \,.\end{equation*} We denote the weighted interaction strength of $\mathcal{L}$ (recall Eq.~\eqref{eq:weightedintersectiondegree}) as:
        \begin{equation}\label{eq:weighted-interaction-strength}
            \alpha := \max_{u \in [n]} \sum_{e\in \mathcal{R}_{n,\le k},u\in e} \|\mathcal{L}^\dagger_{e}\|_{2 \to 2} \,.
        \end{equation}
        The support-growth argument below only uses locality, the fact that each adjoint summand vanishes on observables supported outside its interaction region, and bounds on the sum of induced norms. Its summands need not individually generate completely positive dynamics.  This gives the following bridge between the two strength conventions used here.
        \begin{lemma}[Coefficient-strength bridge]\label{lem:coefficient-strength}
            Let $x=(h,G)$ be the coefficients of a Lindblad generator, with every nonzero coefficient supported on at most $k$ sites. Suppose $\|x\|_{\mathrm{loc},1}\le g$. For nonidentity Pauli strings $\mathbf P,\mathbf Q$, define the elementary coefficient maps
            \begin{equation*}
                \cK^{H}_{\mathbf P}
                :=-ih_{\mathbf P}[\mathbf P,\bullet],
                \qquad
                \cK^{G}_{\mathbf P,\mathbf Q}
                :=G_{\mathbf P,\mathbf Q}
                \left(\mathbf P\bullet\mathbf Q
                -\frac12\{\mathbf Q\mathbf P,\bullet\}\right).
            \end{equation*}
            These maps are supported on their coefficient supports, their adjoints annihilate the identity, and
            \begin{equation*}
                \|(\cK^{H}_{\mathbf P})^\dagger\|_{2\to2}
                \le2|h_{\mathbf P}|,
                \qquad
                \|(\cK^{G}_{\mathbf P,\mathbf Q})^\dagger\|_{2\to2}
                \le2|G_{\mathbf P,\mathbf Q}|.
            \end{equation*}
            The sum of their adjoint norms over terms whose support contains any fixed site is therefore at most $2\|(h,G)\|_{\mathrm{loc},1}\le2g$. Conversely, if $\cL$ admits a local-positive $k$-local decomposition as in \eqref{eq:weighted-interaction-strength}, then
            \begin{equation}\label{eq:witness-strength-controls-coefficients}
                \|(2h,G)\|_{\mathrm{loc},1}
                \le c_k\alpha,
                \qquad
                c_k:=16^k-1.
            \end{equation}
        \end{lemma}
        \begin{proof}
            Unitary invariance and the triangle inequality give the displayed norm bounds. Each adjoint annihilates the identity and hence observables with disjoint support. Summing over coefficients incident on a site gives the strength bound $2g$.

            Conversely, let $\chi^e$ denote the Pauli-superoperator coefficients of $\cL_e$, excluding $(I,I)$. The local Fierz formula bounds each coefficient by $\|\cL_e^\dagger\|_{2\to2}$. There are at most $16^{|e|}-1\le c_k$ such coefficients, so \eqref{eq:physical-chi-map-bounds} gives
            \begin{equation*}
                \|(2h_e,G_e)\|_{\mathrm{loc},1}
                \le\|\chi^e\|_{\mathrm{loc},1}
                \le c_k\|\cL_e^\dagger\|_{2\to2}.
            \end{equation*}
            Since their supports lie in $e$, summing over terms incident on each site proves \eqref{eq:witness-strength-controls-coefficients}.
        \end{proof}
        Given an arbitrary operator $Q$, we consider the degree-$d$ operator-valued Taylor polynomial:
        \begin{equation}\label{eq:approximating-polynomial}
            Q^{(d)}(t) := \sum_{m=0}^{d} \frac{t^{m}}{m!} (\mathcal{L}^{\dagger})^m(Q).
        \end{equation}

        The following lemma expresses the accuracy of $ Q^{(d)}(t)$ as an approximation to the time-evolution of $Q$ under $\mathcal{L}$, and is inspired by the literature of Lieb-Robinson \cite{lieb1972finite, nachtergaele2006lieb, CL19, AnthonyChen2023} and operator-growth bounds \cite{Chen2021, Lucas_2020, LO2020}.

        \begin{lemma}[Polynomial Approximations to Heisenberg Evolution]\label{lem:poly-approx}
            For $q\in\mathbb{N}\backslash \{0\}$, let $Q$ be an operator supported on at most $q$ sites and let $\kappa:=\lceil q/k \rceil$. Then, for any error $\varepsilon \in (0,1)$ and time $t \in [0, (4\alpha k)^{-1}]$, the operator-valued polynomial $Q^{(d)}(t)$ of degree $d:=\kappa+\left\lceil\log_2\frac{1}{\varepsilon}\right\rceil=\cO\!\left(\kappa+\log\frac{1}{\varepsilon}\right)$ satisfies
            \begin{equation}
                \|e^{t\mathcal{L}^\dagger}[Q] - Q^{(d)}(t)\|_{2} \le \varepsilon \cdot \|Q\|_{2}\,.
            \end{equation}
        \end{lemma}
        \noindent

        \noindent To bound the desired accuracy, it suffices to control the norm of $Q_{m} := (\mathcal{L}^{\dagger})^m(Q)=\mathcal{L}^\dagger (Q_{m-1})$.

        \begin{lemma}[Support Decomposition]\label{lem:suppdecom}
            With the notation of Lemma \ref{lem:poly-approx}, for each $m \ge 0$ there exists a decomposition of the Taylor iterate of the form $Q_{m} = \sum_{A \subseteq [n]} Q_{m,A}$ such that $\operatorname{supp}(Q_{m,A}) \subseteq A$ and $Q_{m,A} = 0$ unless $|A| \le q + (k-1) \cdot m$. Furthermore, the total weight is bounded by:
            \begin{equation}
                \sum_A \|Q_{m,A}\|_2 \le \|Q\|_2 \cdot (\alpha k)^m \frac{(m-1+\kappa)!}{(\kappa-1)! }
            \end{equation}
            with $\kappa:=\lceil q/k \rceil$ and $\alpha = \max_{u \in [n]} \sum_{e \ni u} \|\mathcal{L}^\dagger_{e}\|_{2 \to 2}$.
        \end{lemma}

        \begin{proof}
            The first part of the claim follows via the expansion:
            \begin{equation}
                Q_m = \sum_{\text{paths } e} \mathcal{L}^\dagger_{e_m} \dots \mathcal{L}^\dagger_{e_1}(Q),
            \end{equation}
            where the sum is over paths $e=(e_1,\ldots,e_m)$ such that $e_i$ intersects $\operatorname{supp}(Q)\cup e_1\cup\cdots\cup e_{i-1}$ for each $i$. Now, for a fixed $X$ with support on $A \subseteq [n]$, the norm growth by $\mathcal{L}^\dagger$ is bounded by:
            \begin{equation}
                \|\mathcal{L}^\dagger(X)\|_2 = \biggl\|\sum_{e: e \cap A \neq \emptyset} \mathcal{L}^\dagger_e(X)\biggr\|_2 \le \sum_{e: e \cap A \neq \emptyset} \|\mathcal{L}^\dagger_e\|_{2 \to 2} \|X\|_2 \le \alpha |A| \cdot \|X\|_2
            \end{equation}
            where in the first equation we used that $\cL_e^\dagger(X)=0$ for any $e$ with $e\cap A=\emptyset$. Let $W_m$ be the total weight $\sum_A \|Q_{m,A}\|_2$. By induction with $W_0 = \|Q\|_2$:
            \begin{align*}
                W_{m+1} \le \sum_A \sum_{e: e \cap A \neq \emptyset} \|\mathcal{L}^\dagger_e(Q_{m,A})\|_2
                \le \alpha \sum_A |A| \cdot \|Q_{m,A}\|_2 \le \alpha (q + m(k-1)) W_m\,.
            \end{align*}
            The product yields:
            \begin{equation*}
                \prod_{i=0}^{m-1}\alpha(q+i(k-1))
                \le
                (\alpha k)^m\prod_{i=0}^{m-1}(\kappa+i)
                =
                (\alpha k)^m\frac{(m-1+\kappa)!}{(\kappa-1)!}\, .
            \end{equation*}

        \end{proof}

        \noindent We are now ready to prove Lemma \ref{lem:poly-approx}.

        \begin{proof}[Proof of Lemma \ref{lem:poly-approx}]
            With $x := \alpha \cdot t \cdot k$, by Lemma \ref{lem:suppdecom} the error bound is:
            \begin{equation}\label{eq:taylor-path-bound}
                \|e^{t\mathcal{L}^\dagger}(Q) - Q^{(d)}(t)\|_2 \le \|Q\|_2 \sum_{m=d+1}^{\infty} \binom{\kappa+m-1}{m} x^m
            \end{equation}
            For all $m\ge d+1$, let $T_m(x) = \binom{\kappa+m-1}{m}x^m$. The ratio of successive terms satisfies
            \begin{equation*}
                \frac{T_{m+1}(x)}{T_m(x)}
                =
                x\frac{\kappa+m}{m+1}
                \le
                x\left(1+\frac{\kappa-1}{d+2}\right)
                =:
                r_dx.
            \end{equation*}
            Since $d\ge \kappa$ and $x\le1/4$, we have $r_dx<1/2$. Hence
            \begin{equation}\label{eq:local-taylor-path-bound}
                \sum_{m=d+1}^{\infty} T_m(x) \le \frac{T_{d+1}(x)}{1 - r_d \cdot x} \le 2 \cdot T_{d+1}(x)\,.
            \end{equation}
            Moreover, using $x\le1/4$ and $\binom{\kappa+d}{d+1}\le2^{\kappa+d}$, we obtain
            \begin{equation*}
                T_{d+1}(x)
                \le
                \binom{\kappa+d}{d+1}4^{-(d+1)}
                \le
                2^{\kappa-d-2}.
            \end{equation*}
            Therefore
            \begin{equation*}
                \sum_{m=d+1}^{\infty}T_m(x)
                \le
                2^{\kappa-d-1}
                \le
                \varepsilon,
            \end{equation*}
            where the last inequality follows from the choice of $d$.
        \end{proof}

    \subsection{Estimating PTM derivatives}
        Trace preservation gives $L_{I,\mathbf Q}=0$, so we set these estimates to zero without measurements. For $\mathbf P\ne I$ and $\mathbf Q\in\mathcal P_n$, define the scalar degree-$d$ polynomial
        \begin{equation*}
            p_{\mathbf{P},\mathbf{Q}}^{(d)}(t):=\frac{1}{2^n}\Tr( \mathbf{P}^{(d)}(t)\mathbf{Q}).
        \end{equation*}
        By construction, $(p^{(d)}_{\mathbf{P},\mathbf{Q}})'(0)=L_{\mathbf{P},\mathbf{Q}}$. Moreover, by Lemma \ref{lem:poly-approx}, for any $t\in[0,(4\alpha k)^{-1}]$, choosing $\kappa_{\mathbf P}:=\lceil \operatorname{wt}(\mathbf{P})/k\rceil$ and $d=\cO\big(\kappa_{\mathbf P}+\log\frac{1}{\varepsilon_{\operatorname{Shadow}}}\big)$,
        \begin{align*}
            \Big|p^{(d)}_{\mathbf{P},\mathbf{Q}}(t)-\frac{1}{2^n}\Tr({\mathbf{P}} e^{t\cL}({\mathbf{Q}}))\Big|\le \frac{1}{2^n}\,\|\mathbf{Q}\|_2\,\Big\|\mathbf{P}^{(d)}(t)-e^{t\cL^\dagger}[\mathbf{P}]\Big\|_2 \le \frac{\varepsilon_{\operatorname{Shadow}}\|\mathbf{P}\|_2\,\|\mathbf{Q}\|_2}{2^n} \le \varepsilon_{\operatorname{Shadow}}\,.
        \end{align*}

        \noindent To get a good approximation of the polynomial $p^{(d)}_{\mathbf{P},\mathbf{Q}}$, we run a process shadow tomography \cite{StilckFranca.2024,StilckFranca.2025}, which enables the parallel estimation of Pauli overlaps of the form
        \begin{equation*}
            2^{-n}\Tr(\mathbf{P}e^{t\cL}[\mathbf{Q}]).
        \end{equation*}
        The protocol consists of:
        \begin{enumerate}
            \item Preparing a random product Pauli eigenstate, i.e.\ a tensor product of eigenstates of $X$, $Y$, or $Z$ on each qubit.
            \item Evolving under $e^{t\cL}$.
            \item Measuring in a random product Pauli basis.
        \end{enumerate}
        \begin{assumption}[Statistical observation model]\label{ass:statistical-observation-model}
            Experimental shots are mutually independent. In a shot at time $t$, the input state and measurement basis are sampled independently from the product-Pauli ensembles specified above, and the outcome is distributed according to the Born rule for $e^{t\cL}$, with the same generator $\cL$ in every shot.
        \end{assumption}
        For collections of Pauli strings $\{\mathbf{P}_i\}_{i=1}^{K_1}$ and $\{\mathbf{Q}_j\}_{j=1}^{K_2}$ with maximum combined weight
        \begin{equation*}
            {\max_{i,j}}\,\operatorname{wt}(\mathbf{P}_i) + \operatorname{wt}(\mathbf{Q}_j) \leq w\,,
        \end{equation*}
        all $K_1 K_2$ overlaps can be estimated to precision $\varepsilon_{\operatorname{Shadow}}$ with probability at least $1-\delta$ using
        \begin{align}\label{eq:processshadows}
            S = \mathcal{O}\left(3^{w} \log(K_1 K_2 \delta^{-1}) \varepsilon_{\operatorname{Shadow}}^{-2}\right)
        \end{align}
        samples. Here, $K_1$ and $K_2$ are the numbers of queried input/output Pauli strings, and $w=\max_{i,j}(\wt(\mathbf P_i)+\wt(\mathbf Q_j))$. In our local application $w\le 2k$ and $K_1K_2$ is the number of required PTM pairs. The postprocessing involves computing a median-of-means estimator from appropriately weighted measurement outcomes. We refer to~\cite{StilckFranca.2024} for a proof and more details. This tool allows us to efficiently estimate all required expectation values $2^{-n}\Tr(\mathbf{P} e^{t\cL}[\mathbf{Q}])$ in parallel for various times $t$, with sample complexity scaling exponentially only in the maximum weight $w$, which remains constant for local observables.

        Next, we explain how to turn these time-dependent overlap estimates into estimates of the PTM coefficients
        \begin{equation*}
            L_{\mathbf P,\mathbf Q}
            =
            \left.\frac{d}{dt}\right|_{t=0}
            2^{-n}\Tr\!\left(\mathbf P e^{t\cL}(\mathbf Q)\right).
        \end{equation*}
        We follow a scheme already used in \cite{StilckFranca.2024,StilckFranca.2025}: this is done via a robust polynomial-interpolation strategy of~\cite{StilckFranca.2024}, but with one important difference in the present setting: the approximation of the time trace is provided directly by Lemma \ref{lem:poly-approx}. No Lieb--Robinson bound or geometric locality assumption is used. The only dynamical input is that the generator is $k$-local and has a bounded weighted interaction strength $\alpha$. Fix $(\mathbf P,\mathbf Q)$, and define
        \begin{equation*}
            f_{\mathbf P,\mathbf Q}(t)
            :=
            2^{-n}\Tr\!\left(\mathbf P e^{t\cL}(\mathbf Q)\right).
        \end{equation*}
        Then $f'_{\mathbf P,\mathbf Q}(0) = 2^{-n}\Tr\!\left(\mathbf P\cL(\mathbf Q)\right) = L_{\mathbf P,\mathbf Q}$. Set $T:=(4\bar\alpha k)^{-1}$. Since the decomposition has weighted interaction strength at most $\bar\alpha$, Lemma~\ref{lem:poly-approx} applies on $[0,T]$. For $\kappa_{\mathbf P}:=\left\lceil\frac{\wt(\mathbf P)}{k}\right\rceil$ and for $d = \cO\!\left(\kappa_{\mathbf P}+\log\frac{1}{\tau}\right)$, the polynomial $p_{\mathbf P,\mathbf Q}^{(d)}(t) := 2^{-n}\Tr\!\left(\mathbf P^{(d)}(t)\mathbf Q\right)$ satisfies
        \begin{equation}\label{eq:poly-uniform-approx}
            \sup_{t\in[0,T]}
            \left|
            f_{\mathbf P,\mathbf Q}(t)
            -
            p_{\mathbf P,\mathbf Q}^{(d)}(t)
            \right|
            \le \tau.
        \end{equation}
        Moreover, $\left(p_{\mathbf P,\mathbf Q}^{(d)}\right)'(0) = L_{\mathbf P,\mathbf Q}$, and both $f_{\mathbf P,\mathbf Q}$ and $p_{\mathbf P,\mathbf Q}^{(d)}$ are real-valued because the dynamics preserves Hermiticity.
        \begin{lemma}[Robust polynomial regression on a time interval]\label{lem:robust-time-polynomial-regression}
            Fix $\eta\in(0,1)$ and $\beta\in[0,1/2)$. There exists a constant $C_{\beta,\eta}>0$, depending only on $\beta$ and $\eta$, such that the following holds. Let $p^{(d)}:[0,T]\to\mathbb R$ have degree at most $d$, and let $\{(t_i,y_i)\}_{i=1}^m$ be $\beta$-good with respect to a Chebyshev partition of $[0,T]$ into $s\geq C_{\beta,\eta}d$ intervals: every interval contains a sample, and in each interval at most a $\beta$ fraction of the samples violate $|y_i-p^{(d)}(t_i)|\leq\varepsilon_{\operatorname{Shadow}}+\tau$. Then the robust regression algorithm of \cite{KaneKarmalkarPrice2017} returns a degree-$d$ polynomial $\widehat p$ satisfying
            \begin{equation*}
                \sup_{t\in[0,T]}|\widehat p(t)-p^{(d)}(t)|\leq C_{\operatorname{reg}}(\varepsilon_{\operatorname{Shadow}}+\tau),\qquad C_{\operatorname{reg}}:=2+\eta.
            \end{equation*}
        \end{lemma}
        We choose $m=\mathcal{O}(d\log d)$ interpolation times $t_1,\ldots,t_m\in[0,T]$ covering such a Chebyshev partition. At each time $t_i$, the process-shadow estimator gives estimates $\widehat f_{\mathbf P,\mathbf Q}(t_i)$ such that, uniformly over all pairs $(\mathbf P,\mathbf Q)$ to be learned and all interpolation times,
        \begin{equation}\label{eq:time-sample-error}
            \left|
            \widehat f_{\mathbf P,\mathbf Q}(t_i)
            -
            f_{\mathbf P,\mathbf Q}(t_i)
            \right|
            \le \varepsilon_{\operatorname{Shadow}}
        \end{equation}
        with probability at least $1-\delta$. By the process-shadow bound and a union bound over all $m$ interpolation times and all required $M_k:=|\mathcal I_k|=\mathcal O_k(n^k)$, it suffices to use
        \begin{equation*}
            S_i
            =
            \mathcal{O}_k\!\left(
            \log\!\left(\frac{mn^k}{\delta}\right)
            \varepsilon_{\operatorname{Shadow}}^{-2}
            \right)
        \end{equation*}
        samples at each time $t_i$. Then, with probability at least $1-\delta$, \eqref{eq:time-sample-error} holds simultaneously for every $i\in[m]$ and every $(\mathbf P,\mathbf Q)\in\mathcal I_k$. On this event, \eqref{eq:poly-uniform-approx} and the triangle inequality show that the data are $0$-good with error $\varepsilon_{\operatorname{Shadow}}+\tau$. Lemma~\ref{lem:robust-time-polynomial-regression} therefore returns a degree-$d$ polynomial $\widehat p_{\mathbf P,\mathbf Q}$ obeying
        \begin{equation}\label{eq:robust-poly-fit}
            \sup_{t\in[0,T]}
            \left|
            \widehat p_{\mathbf P,\mathbf Q}(t)
            -
            p_{\mathbf P,\mathbf Q}^{(d)}(t)
            \right|
            \le
            C_{\operatorname{reg}}(\varepsilon_{\operatorname{Shadow}}+\tau).
        \end{equation}
        The point of using robust interpolation is that the conclusion is a uniform-in-time polynomial approximation on the whole interval, rather than a bound only at the sampled times. It remains to convert the uniform polynomial approximation into an estimate of the derivative at zero. We use Markov brothers' inequality \cite{StilckFranca.2024,StilckFranca.2025}: for every real or complex polynomial $r$ of degree at most $d$,
        \begin{equation}\label{eq:markov-brothers}
            \sup_{t\in[0,T]} |r'(t)|
            \le
            \frac{2d^2}{T}
            \sup_{t\in[0,T]} |r(t)|.
        \end{equation}
        Applying \eqref{eq:markov-brothers} to $r(t) = \widehat p_{\mathbf P,\mathbf Q}(t) - p_{\mathbf P,\mathbf Q}^{(d)}(t)$ and using \eqref{eq:robust-poly-fit}, setting $\widehat L_{\mathbf P,\mathbf Q} := \widehat p'_{\mathbf P,\mathbf Q}(0)$ and choosing $ \varepsilon_{\operatorname{Shadow}},\tau \le \frac{T\varepsilon_L}{4C_{\operatorname{reg}}d^2} $, we get
        \begin{equation}\label{eq:shadow-ptm-approximation}
            \left|
            \widehat L_{\mathbf P,\mathbf Q}
            -
            L_{\mathbf P,\mathbf Q}
            \right|
            \le
            \varepsilon_L.
        \end{equation}
    \subsection{Coefficient-learning guarantees}
        \begin{theorem}[Support-free learning with a sitewise coefficient tail]
            \label{thm:approximate-degree-learning}
            Fix $k$, $\varepsilon>0$, and $\delta\in(0,1)$.  Suppose the unknown generator admits a local-positive $k$-local decomposition of weighted interaction strength at most the supplied upper bound $\bar\alpha$.  Let the learner also receive a number $r\ge1$ such that, for some threshold $\tau>0$,
            \begin{equation*}
                r_\tau\le r,
                \qquad
                \beta_\tau\le\frac{\varepsilon}{2\cdot4^k}.
            \end{equation*}
            Given $\bar\alpha$ and $r$, but neither $\tau$ nor any coefficient locations, the protocol returns initial estimates $\widehat\chi$ and symmetrized estimates $\widehat h,\widehat G$ such that, with probability at least $1-\delta$,
            \begin{align*}
                \|\widehat\chi-\chi\|_{\mathrm{loc},1}&\le\varepsilon,\\
                \|(2(\widehat h-h),\widehat G-G)\|_{\mathrm{loc},1}&\le\varepsilon.
            \end{align*}
            The initial $\chi$ estimate and symmetrized coefficient estimates have $\mathcal O_k(r)$ nonzero entries incident on each site. In particular every independent coefficient has error at most $\varepsilon$.  The protocol uses
            \begin{equation}\label{eq:approximate-degree-samples}
                \widetilde{\mathcal O}_k\!\left(
                \bar\alpha^2r^2\varepsilon^{-2}
                \log\frac n\delta
                \right)
            \end{equation}
            samples,
            \begin{equation}\label{eq:approximate-degree-total-time}
                \widetilde{\mathcal O}_k\!\left(
                \bar\alpha r^2\varepsilon^{-2}
                \log\frac n\delta
                \right)
            \end{equation}
            total evolution time, and polylogarithmically many distinct times in $[0,(4\bar\alpha k)^{-1}]$.  Scanning and inverting the coefficient array takes $\mathcal O_k(n^k)$ arithmetic operations, in addition to processing the experimental records and fitting the time polynomials.  No minimum absolute value is assumed for nonzero coefficients.

            The same result holds for a Lindblad generator whose nonzero coefficients are each supported on at most $k$ sites if, instead, the learner receives an upper bound $g>0$ satisfying $g\ge\|\chi\|_{\mathrm{loc},1}$. In this case replace $\bar\alpha^2$ and $\bar\alpha$ in \eqref{eq:approximate-degree-samples}-- \eqref{eq:approximate-degree-total-time} by $g^2$ and $g$, respectively, and use times in $[0,(8gk)^{-1}]$.
        \end{theorem}
        \begin{proof}
            By \eqref{eq:physical-chi-map-bounds} and Lemma~\ref{lem:coefficient-strength}, the local-positive model has $\|\chi\|_{\mathrm{loc},1}\le2c_k\bar\alpha$. The zero estimate suffices when $\varepsilon\ge2c_k\bar\alpha$ in this model, or when $\varepsilon\ge g$ in the coefficient-local model. In the remaining cases, choose
            \begin{equation*}
                \nu:=\frac{\varepsilon}{6\cdot4^k r}.
            \end{equation*}
            The process-shadow and derivative-estimation routine estimates all local PTM entries, and hence all local coefficient sums $y_v$ through the fixed-size Fierz transform, uniformly to absolute error $\nu$ with the sample count in \eqref{eq:approximate-degree-samples}. Thresholding at $2\nu$ and applying $\mathsf Z^{-1}$ gives, by Lemma~\ref{lem:threshold-before-inversion},
            \begin{equation*}
                \|\widehat\chi-\chi\|_{\mathrm{loc},1}
                \le3\cdot4^k r\nu+4^k\beta_\tau\le\varepsilon.
            \end{equation*}
            Mapping through $\Psi$ and symmetrizing proves the physical coefficient bound.  Since $\beta_\tau/\nu=\mathcal O_k(r)$, \eqref{eq:threshold-before-inversion-degree} proves the degree bound.  The maximal sample duration proves \eqref{eq:approximate-degree-total-time}. The number of times and classical inversion cost follow from the existing derivative routine and the fact that every local coefficient sum kept after thresholding has at most $2^k$ smaller-support descendants.

            In the coefficient-local case, \eqref{eq:physical-chi-map-bounds} gives $\|x\|_{\mathrm{loc},1}\le\|(2h,G)\|_{\mathrm{loc},1}\le g$.  Lemma~\ref{lem:coefficient-strength} provides algebraic local summands of weighted strength at most $2g$.  The proof of Lemma~\ref{lem:poly-approx} uses only their support, identity-annihilation, and norm bounds, so the same PTM routine applies on $[0,(8gk)^{-1}]$. The full evolution is still the physical channel $e^{t\cL}$.
        \end{proof}

        \begin{corollary}[Generic support-free coefficient recovery]
            \label{cor:generic-local-norm-recovery}
            Suppose the learner is given $\bar\alpha>0$ such that the unknown $k$-local Lindblad generator admits a local-positive $k$-local decomposition of weighted interaction strength at most $\bar\alpha$.  For every $\varepsilon>0$ and $\delta\in(0,1)$, the protocol returns initial estimates $\widehat\chi$ and symmetrized physical coefficients $\widehat x=(\widehat h,\widehat G)$ such that, with probability at least $1-\delta$,
            \begin{equation*}
                \|\widehat\chi-\chi\|_{\mathrm{loc},1}\le\varepsilon,
                \qquad
                \|(2(\widehat h-h),\widehat G-G)\|_{\mathrm{loc},1}\le\varepsilon.
            \end{equation*}
            It uses
            \begin{equation*}
                \widetilde{\mathcal O}_k\!\left(
                \bar\alpha^2n^{2k-2}\varepsilon^{-2}
                \log\frac n\delta
                \right)
            \end{equation*}
            samples, $\widetilde{\mathcal O}_k(\bar\alpha n^{2k-2} \varepsilon^{-2}\log(n/\delta))$ total evolution time, and $\mathcal O_k(n^k)$ coefficient-scan and inversion time.  No degree, sparsity, support, coefficient-gap, or tail promise is required.
        \end{corollary}
        \begin{proof}
            For a fixed site and exact support of size $s$, there are at most $15^s$ independent $\chi$-coordinates. Thus their total number incident on a site is at most $r_{n,k}:=\sum_{s=1}^k\binom{n-1}{s-1}15^s =\mathcal O_k(n^{k-1})$. For a sufficiently small $\tau>0$, all nonzero $\chi$-coefficients belong to $\Omega_\tau$, so $\beta_\tau=0$ and $r_\tau\le r_{n,k}$.  Apply Theorem~\ref{thm:approximate-degree-learning} with $r=r_{n,k}$.
        \end{proof}

        \begin{corollary}[Thresholding the recovered coefficients]
            \label{cor:approximate-degree-thresholding}
            Under the hypotheses of Theorem~\ref{thm:approximate-degree-learning}, with probability at least $1-\delta$, for both $z=\chi$ and $z=x$ and every $\lambda>\varepsilon$,
            \begin{equation*}
                \{v:|z_v|>\lambda+\varepsilon\}
                \subseteq\{v:|\widehat z_v|>\lambda\}
                \subseteq\{v:|z_v|>\lambda-\varepsilon\},
            \end{equation*}
            where $v$ ranges over the corresponding independent coordinates. If every nonzero entry of $z$ has absolute value greater than $\lambda+\varepsilon$, thresholding recovers its exact support. These statements apply before physical projection, which may introduce nonzero entries.
        \end{corollary}
        \begin{proof}
            The theorem gives entrywise error at most $\varepsilon$ for both arrays, so the claims follow from the triangle inequality.
        \end{proof}

        Corollary~\ref{cor:generic-local-norm-recovery} and the supplied-candidate and guarded inversion bounds give the following overall coefficient-learning guarantee.

        \begin{theorem}[Entrywise and guarded coefficient recovery]\label{thm:overall-coefficient-learning}
            Fix $k$, $\varepsilon_\chi,\delta\in(0,1)$, and suppose the learner is given an upper bound $\bar\alpha>0$. Given an unknown $k$-local generator $\cL$ admitting a local-positive $k$-local decomposition with weighted interaction strength at most $\bar\alpha$, the protocol specified above outputs coefficients $\widehat G_{\mathbf P,\mathbf Q}$ and $\widehat h_{\mathbf P}$ such that, with probability at least $1-\delta$,
            \begin{equation*}
                \sup_{\substack{\mathbf P,\mathbf Q\ne I\\ |\operatorname{supp}(\mathbf P)\cup\operatorname{supp}(\mathbf Q)|\le k}}\big|\widehat G_{\mathbf P,\mathbf Q}-G_{\mathbf P,\mathbf Q}\big|
                \le
                \varepsilon_\chi,
                \qquad
                \sup_{\mathbf P\in\mathcal P_{n,\le k}\setminus\{I\}}\big|\widehat h_{\mathbf P}-h_{\mathbf P}\big|
                \le
                \varepsilon_\chi.
            \end{equation*}
            The protocol uses
            \begin{align*}
                &\widetilde{\mathcal O}_k\!\left(
                {\bar\alpha^2\varepsilon_\chi^{-2}n^{2k-2}\log(n/\delta)}
                \right)& \text{ samples, and}\\
                & \mathcal O_k(\operatorname{polylog}(\bar\alpha n/\varepsilon_\chi))& \text{distinct sampled evolution times in }[0,(4\bar\alpha k)^{-1}].
            \end{align*}
            Here the $\widetilde{\mathcal O}$-notation hides factors polynomial in $\log(\bar\alpha n/\varepsilon_\chi)$ and constants depending only on $k$.

            For the supplied-candidate and guarded variants, we separate the costs associated to support search from parameter learning. If, for some threshold $\tau>0$, a candidate set $\Omega\subseteq\mathcal I_k^\circ$ is supplied and satisfies the assumptions of Corollary~\ref{cor:dissipative-learning-stable-sparse}, then estimating the coefficients indexed by this set only requires
            \begin{align*}
                &\widetilde{\mathcal O}_{k}\!\left( \frac{\bar\alpha^2D_{\Omega}^2}{\varepsilon_\chi^2}\log(|\Omega|/\delta)
                \right)&\text{samples and}\\
                &\mathcal O_k(|\Omega|)&\text{ classical postprocessing time.}
            \end{align*}
            Moreover, fix a decision threshold $\lambda>0$ and a margin $0<\gamma<\lambda$. If the learner receives a bound $B\ge D_{\lambda-\gamma}$ and
            \begin{equation*}
                B\rho_{\lambda+\gamma}\le \frac{\gamma}{2},
            \end{equation*}
            then a set $\widehat\Omega_\lambda$ that includes every coefficient with $|\chi_{\mathbf P,\mathbf Q}|>\lambda+\gamma$ and rejects every coefficient with $|\chi_{\mathbf P,\mathbf Q}|\le\lambda-\gamma$ can be found with probability at least $1-\delta$ using
            \begin{align*}
                &\widetilde{\mathcal O}_k\left(
                \frac{\bar\alpha^2B^2}{\gamma^2}\log\frac{n}{\delta}
                \right)&\text{ samples and}\\
                &\mathcal O_k(n^k)&\text{ classical postprocessing time}.
            \end{align*}
        \end{theorem}
        \noindent Thresholding the generic coefficient estimates gives the following unconditional support guarantee.

        \begin{corollary}[Unconditional threshold structure learning]
            \label{cor:unconditional-thresholding}
            Fix $k$, and let $\cL$ be an unknown $k$-local Lindblad generator admitting a local-positive $k$-local decomposition whose weighted interaction strength is at most the supplied upper bound $\bar\alpha$. Moreover, fix $\lambda>0$, $0<\gamma<\lambda$, and $\delta\in(0,1)$. Run the generic support-free estimator to uniform independent-coefficient accuracy $\gamma$, and set $\widehat\Omega_\lambda:=\{(\mathbf P,\mathbf Q)\in\mathcal I_k^\circ:|\widehat\chi_{\mathbf P,\mathbf Q}|>\lambda\}$. Then, with probability at least $1-\delta$,
            \begin{equation*}
                \Omega_{\lambda+\gamma}\subseteq\widehat\Omega_\lambda\subseteq\Omega_{\lambda-\gamma}.
            \end{equation*}
            The algorithm uses $\widetilde{\mathcal O}_k(\bar\alpha^2n^{2k-2}\gamma^{-2}\log(n/\delta))$ samples, total evolution time $\widetilde{\mathcal O}_k(\bar\alpha n^{2k-2}\gamma^{-2}\log(n/\delta))$, and $\mathcal O_k(n^k)$ classical postprocessing.
        \end{corollary}
        \begin{proof}
            Corollary~\ref{cor:generic-local-norm-recovery} at accuracy $\gamma$ gives uniform $\chi$-coefficient error at most $\gamma$ with the stated resources. Both inclusions follow from the triangle inequality.
        \end{proof}
        \noindent Note that this guarantee requires no degree or sparsity bound, no supplied coefficient locations, and no tail condition.

        \noindent The inversion becomes more accurate when only coefficients with sufficiently large support are required.
        \begin{corollary}\label{cor:support-size-stratified}
            Under the assumptions of Theorem~\ref{thm:overall-coefficient-learning}, fix $s_0\in\{1,\ldots,k\}$. All Hamiltonian coordinates with $\operatorname{wt}(\mathbf P)\ge s_0$ and dissipative coordinates with $|\operatorname{supp}(\mathbf P)\cup\operatorname{supp}(\mathbf Q)|\ge s_0$ can be learned to entrywise error $\varepsilon_\chi$ with probability at least $1-\delta$ using
            \begin{equation*}
                \widetilde{\mathcal O}_k\!\left(\bar\alpha^2n^{2(k-s_0)}\varepsilon_\chi^{-2}\log\frac{n}{\delta}\right)
            \end{equation*}
            quantum samples and $\mathcal O_k(n^k)$ classical postprocessing time. In particular, the layer of coefficients with union-support size exactly $k$ requires only logarithmic dependence on $n$ in the quantum sample count.
        \end{corollary}
        \begin{proof}
            First, for a coefficient of union-support size $s$, Lemma~\ref{lem:descending-extension-recursion} bounds the inversion error by $C_{k-s}n^{k-s}\varepsilon_L$. Every recursive child is a strict diagonal extension and therefore has larger union support, so recovering only the layers $s\ge s_0$ is closed under the descending recursion. Uniformly on these layers the error is $\mathcal O_k(n^{k-s_0}\varepsilon_L)$. Choosing $\varepsilon_L=\Theta_k(\varepsilon_\chi/n^{k-s_0})$ and applying the process-shadow derivative estimator proves the sample bound.
        \end{proof}

        \noindent For Hamiltonian generators, the absence of nontrivial diagonal extensions yields a sharper bound.
        \begin{corollary}[Hamiltonian specialization]
            \label{cor:hamiltonian-specialization}
            Suppose $\cL=-i[H,\cdot]$, where $H=\sum_{\mathbf P\in\mathcal P_{n,\le k}\setminus\{I\}}h_{\mathbf P}\mathbf P$, and the learner is given $\bar\alpha$ satisfying $2\|h\|_{\mathrm{loc},1}\le\bar\alpha$. Then all $h_{\mathbf P}$ are learned to entrywise error $\varepsilon_h$ with probability at least $1-\delta$ using $\widetilde{\mathcal O}_k(\bar\alpha^2\varepsilon_h^{-2}\allowbreak\log(n/\delta))$ samples, total evolution time $\widetilde{\mathcal O}_k(\bar\alpha\varepsilon_h^{-2}\allowbreak\log(n/\delta))$, and $\mathcal O_k(n^k)$ classical output time.
        \end{corollary}
        \begin{proof}
            For Hamiltonian dynamics the $\chi$-matrix is supported only on $(\mathbf R,I)$ and $(I,\mathbf R)$, so every nontrivial diagonal extension of either coordinate has both indices nonidentity and vanishes. The local Fierz inversion therefore estimates $\chi_{\mathbf R,I}$ and $\chi_{I,\mathbf R}$ directly to $\mathcal O_k(\varepsilon_L)$, and $h_{\mathbf R}=i(\chi_{\mathbf R,I}-\chi_{I,\mathbf R})/2$. Taking $\varepsilon_L=\Theta_k(\varepsilon_h)$ and union bounding the process-shadow estimates over $\mathcal O_k(n^k)$ coordinates proves the sample bound. The maximal query time gives the total-time bound.
        \end{proof}

\section{\texorpdfstring{Extracting the best approximate $k$-local Lindblad form}{Extracting the best approximate k-local Lindblad form}}\label{sec:sdp-projection}
    The previous sections estimate the Hamiltonian and dissipative coefficients $h,G$. Statistical and interpolation errors can violate positivity, so we fit these coefficients in local coefficient $1$-norm, with Hamiltonian errors counted twice, over the cone generated by positive local Kossakowski blocks.  The resulting semidefinite program is detailed in Section~\ref{app:sdp-projection-details}.

    For each $e\in\mathcal R_{n,\leq k}$, let $X_{\operatorname{true}}^e$ be the Kossakowski matrix of the dissipative part of $\cL_e$ in the basis $\cP_e^\circ$, with entries $(X_{\operatorname{true}}^e)_{\mathbf P,\mathbf Q}:=\sum_a\ell_{e,a,\mathbf P}\overline{\ell_{e,a,\mathbf Q}}$. Then $X_{\operatorname{true}}^e\succeq0$. Writing $X_{\operatorname{true}}=(X_{\operatorname{true}}^e)_e$ and $\cA(X)$ for the embedded sum of the local blocks, \eqref{eq:linblad-form} gives $G=\cA(X_{\operatorname{true}})$. Section \ref{app:sdp-projection-details} verifies the sitewise trace constraints of the semidefinite program.

    Set $r_k:=4^k-1$.  Given the initial estimates $\widehat h,\widehat G$, write $\widehat x=(\widehat h,\widehat G)$ and consider
    \begin{equation}\label{eq:local-norm-physical-projection}
        \begin{aligned}
            \min_{\widetilde h,X}\quad&
            \|(2(\widetilde h-\widehat h),\cA(X)-\widehat G)\|_{\mathrm{loc},1}\\
            \text{subject to}\quad&
            X^e\succeq0,\qquad
            \sum_{e\ni j}\Tr(X^e)\le r_k\bar\alpha
            \quad(j\in[n]),\qquad
            \widetilde h\in\mathbb R^{|\cP_{\le k}^\circ|}.
        \end{aligned}
    \end{equation}
    The local norm is represented using auxiliary variables bounding the absolute values of the residuals, with one linear sum constraint per site. The absolute-value constraints are second-order-cone constraints.  Thus \eqref{eq:local-norm-physical-projection} is a semidefinite program of polynomial size for fixed $k$.

    \begin{proposition}[Local coefficient $1$-norm physical projection]
        \label{prop:local-norm-physical-projection}
        Suppose the true generator admits a local-positive $k$-local decomposition of weighted interaction strength at most the supplied bound $\bar\alpha$, and $\|(2(\widehat h-h),\widehat G-G)\|_{\mathrm{loc},1}\le\varepsilon$. A feasible solution of \eqref{eq:local-norm-physical-projection} whose objective value is at most $\xi$ above the optimum (an additive-$\xi$ solution) defines a valid $k$-local Lindblad generator with coefficient array $x^{\mathrm{phys}}=(h^{\mathrm{phys}},G^{\mathrm{phys}})$ satisfying
        \begin{equation}\label{eq:local-norm-projection-error}
            \|(2(h^{\mathrm{phys}}-h),G^{\mathrm{phys}}-G)\|_{\mathrm{loc},1}
            \le2\varepsilon+\xi.
        \end{equation}
        Thus statistical and numerical errors each at most one third of the target coefficient error suffice. For fixed $k$, a standard dense interior-point implementation uses
        \begin{equation*}
            \widetilde{\mathcal O}_k\!\left(
            n^{9k/2}\log\frac1\xi
            \right)
        \end{equation*}
        arithmetic operations under the normalized polynomial-input convention specified in Section~\ref{app:sdp-projection-details}.  The projection may make the coefficient array denser, so support statements refer to the initial estimator before this step.
    \end{proposition}
    \begin{proof}
         The proof is given in Section~\ref{app:sdp-projection-details}.
    \end{proof}

    The following deterministic estimate converts coefficient error into generator error in diamond norm.

    \begin{lemma}[Coefficient-to-diamond bound]\label{lem:coeff-to-diamond}
        Let $\cL_{h,G}$ and $\cL_{h',G'}$ be two $k$-local generators written in Pauli--GKSL coefficient form. Write $x=(h,G)$ and $x'=(h',G')$. Then
        \begin{equation}\label{eq:local-norm-to-diamond}
            \|\cL_{h,G}-\cL_{h',G'}\|_\diamond
            \le
            2\sum_{\mathbf P}|h_{\mathbf P}-h'_{\mathbf P}|
            +2\sum_{\mathbf P,\mathbf Q}
            |G_{\mathbf P,\mathbf Q}-G'_{\mathbf P,\mathbf Q}|
            \le
            2n\|x-x'\|_{\mathrm{loc},1}.
        \end{equation}
    \end{lemma}
    Combining local coefficient $1$-norm learning with the matching physical projection gives the final valid-generator guarantee.

    \begin{theorem}[Learning a valid $k$-local Lindblad generator]\label{thm:full-lindblad-learning}
        Let $\cL=\cL_{h,G}$ be a general $k$-local Lindblad generator on $n$ qubits admitting a local-positive $k$-local decomposition of weighted interaction strength at most $\bar\alpha$, where $\bar\alpha$ is supplied to the learner. Fix $\varepsilon_\diamond,\delta\in(0,1)$. The projection procedure of Section~\ref{app:sdp-projection-details} returns a valid $k$-local Lindblad generator $\widehat\cL$ such that, with probability at least $1-\delta$,
        \begin{equation*}
            \|\widehat\cL-\cL\|_\diamond
            \le
            \varepsilon_\diamond.
        \end{equation*}
        The algorithm requires
        \begin{align*}
            &\widetilde{\mathcal O}_{k}\!\left(
            \frac{\bar\alpha^2n^{2k}
            \log(n/\delta)
            }{
            \varepsilon_\diamond^2
            }
            \right)& \text{ samples and}\\
            & \mathcal O_k(\operatorname{polylog}(\bar\alpha n/\varepsilon_\diamond))& \text{distinct sampled evolution times in }[0,(4\bar\alpha k)^{-1}].
        \end{align*}
        Here the $\widetilde{\mathcal O}$-notation hides factors polynomial in $\log(\bar\alpha n/\varepsilon_\diamond)$ and constants depending only on $k$. For fixed $k$, the projection in local coefficient $1$-norm and the remaining classical postprocessing run in polynomial time.

        More generally, if the learner is given $r_\diamond\ge1$ such that, for some $\tau>0$,
        \begin{equation*}
            r_\tau\le r_\diamond,
            \qquad
            \beta_\tau\le\frac{\varepsilon_\diamond}{16n4^k},
        \end{equation*}
        then the sample bound can be replaced by
        \begin{equation}\label{eq:approximate-degree-diamond-samples}
            \widetilde{\mathcal O}_k\!\left(
            \bar\alpha^2r_\diamond^2n^2
            \varepsilon_\diamond^{-2}
            \log\frac n\delta
            \right).
        \end{equation}
        The tail condition is thus imposed at the finer local accuracy scale $\Theta_k(\varepsilon_\diamond/n)$. For the broader coefficient-local model of Theorem~\ref{thm:approximate-degree-learning}, the global-PSD projection gives a valid coefficient-local generator using
        \begin{equation*}
            \widetilde{\mathcal O}_k\!\left(
            g^2r_\diamond^2n^2
            \varepsilon_\diamond^{-2}
            \log\frac n\delta
            \right)
        \end{equation*}
        samples and times in $[0,(8gk)^{-1}]$. This output need not admit a decomposition into positive local blocks.
    \end{theorem}

    \begin{rmk}[Strength rescaling and total evolution time]\label{rmk:strength-total-time}
        For the protocols with supplied weighted-strength bound $\bar\alpha>0$, each experiment has duration at most $(4k\bar\alpha)^{-1}$. Thus $N$ samples have total evolution time at most $N/(4k\bar\alpha)$, and the stated sample bounds give total-time bounds with $\bar\alpha^2$ replaced by $\bar\alpha$, up to $k$-dependent constants. For the coefficient-local branches with supplied bound $g>0$, the corresponding total-time bound is $N/(8kg)$.
    \end{rmk}

\section{Model-misspecified Lindbladian learning}\label{sec:agnostic-lindbladian-learning}
    We now allow the data-generating Lindbladian $\cL$ to be outside the model class used by the learner. This is analogous to model-misspecified learning in graphical models (also known as agnostic learning) where the estimator is compared to a best-fitting model rather than assuming exact realizability \cite{Bhattacharyya2021}. For the supplied upper bound $\bar\alpha$, let $\mathfrak L_{k,\bar\alpha}$ denote the class of $k$-local Lindbladians that admit a decomposition of weighted interaction strength at most $\bar\alpha$ and satisfy the usual GKSL feasibility constraints. For $\tau>0$ and a supplied candidate set $\Omega\subseteq\mathcal I_k^\circ$, define
    \begin{equation}\label{eq:closed-misspecified-support-class}
        \mathfrak L_{k,\bar\alpha,\Omega,\tau}
        :=
        \left\{
        \cL_{h,G}\in\mathfrak L_{k,\bar\alpha}:
        |\chi_{\mathbf P,\mathbf Q}(\cL_{h,G})|\leq\tau
        \text{ for all }(\mathbf P,\mathbf Q)\in\mathcal I_k^\circ\setminus\Omega
        \right\}.
    \end{equation}
    Here $\chi(\cL_{h,G})$ denotes the Pauli-superoperator coefficient matrix of $\cL_{h,G}$. For any generator $\cL$, define the best-in-class approximation error
    \begin{equation*}
        \mathrm{opt}(\cL;\mathfrak L_{k,\bar\alpha}):=\inf_{\cL'\in\mathfrak L_{k,\bar\alpha}}\|\cL-\cL'\|_\diamond\,.
    \end{equation*}
    For fixed $n,k,\bar\alpha$, the tuples of local GKSL superoperators satisfying the weighted interaction strength bound form a closed and bounded subset of a finite-dimensional product space, hence a compact set. Its image under the continuous sum map is $\mathfrak L_{k,\bar\alpha}$, so the generic infimum is attained. The non-strict coefficient constraints in \eqref{eq:closed-misspecified-support-class} define a closed, hence compact, subclass $\mathfrak L_{k,\bar\alpha,\Omega,\tau}$, so the supplied-candidate infimum is attained as well. The following proposition shows that the pipeline is stable under misspecification relative to a best-fitting generator in $\mathfrak L_{k,\bar\alpha}$. The sample complexity and sampled-time count are unchanged from the realizable setting apart from the accuracy needed to absorb the approximation bias. The supplied-candidate guarantee has the same form, with the global $\mathcal O_k(n^{k-1})$ independent-coefficient stability factor replaced by the sparse stability parameter $D_\Omega$.

    \begin{proposition}[Model-misspecified Lindbladian learning]\label{prop:agnostic-learning}
        Fix $k$ and the supplied upper bound $\bar\alpha$. Let $\cL$ be any Lindbladian generator, set $\mathrm{opt}:=\mathrm{opt}(\cL;\mathfrak L_{k,\bar\alpha})$, and let $\cL_k^*=\cL_{h^*,G^*}\in\mathfrak L_{k,\bar\alpha}$ be a best-fitting generator attaining $\mathrm{opt}$. For any $\varepsilon_\chi,\delta\in(0,1)$, with probability at least $1-\delta$ the algorithm outputs coefficients $\widehat G$ and $\widehat h$ such that
        \begin{equation*}
            \begin{aligned}
                \sup_{\substack{\mathbf P,\mathbf Q\ne I\\|\operatorname{supp}(\mathbf P)\cup\operatorname{supp}(\mathbf Q)|\le k}}\big|\widehat G_{\mathbf P,\mathbf Q}-G^*_{\mathbf P,\mathbf Q}\big|&\le\varepsilon_\chi+\cO\!\left(n^{k-1}\operatorname{polylog}(n/\varepsilon_\chi)\mathrm{opt}\right),\\
                \sup_{\mathbf P\in\cP_{n,\le k}\setminus\{I\}}\big|\widehat h_{\mathbf P}-h^*_{\mathbf P}\big|&\le\varepsilon_\chi+\cO\!\left(n^{k-1}\operatorname{polylog}(n/\varepsilon_\chi)\mathrm{opt}\right).
            \end{aligned}
        \end{equation*}
        The protocol uses $\widetilde{\mathcal O}_k\!\left(\bar\alpha^2\varepsilon_\chi^{-2}n^{2k-2}\log(n/\delta)\right)$ samples and $\mathcal O_k(\operatorname{polylog}(\bar\alpha n/\varepsilon_\chi))$ times in $[0,(4\bar\alpha k)^{-1}]$, where the $\widetilde{\mathcal O}$-notation hides factors polynomial in $\log(\bar\alpha n/\varepsilon_\chi)$ and constants depending only on $k$.

        Moreover, let
        \begin{equation*}
            \mathrm{opt}_{\Omega,\tau}
            :=
            \inf_{\cL'\in\mathfrak L_{k,\bar\alpha,\Omega,\tau}}
            \|\cL-\cL'\|_\diamond,
        \end{equation*}
        and let $\cL_{\Omega,\tau}^*=\cL_{h_{\Omega,\tau}^*,G_{\Omega,\tau}^*}\in\mathfrak L_{k,\bar\alpha,\Omega,\tau}$ attain $\mathrm{opt}_{\Omega,\tau}$ for a supplied candidate set $\Omega\subseteq\mathcal I_k^\circ$. Evaluate the unresolved tail $\rho_\Omega(\cL')$ using the coefficients of $\cL'$ and assume $\tau+D_\Omega\rho_\Omega(\cL_{\Omega,\tau}^*)\le\varepsilon_\chi/2$. Then, with probability at least $1-\delta$, the sparse supplied-candidate algorithm outputs coefficients satisfying
        \begin{equation*}
            \begin{aligned}
                \sup_{\substack{\mathbf P,\mathbf Q\ne I\\|\operatorname{supp}(\mathbf P)\cup\operatorname{supp}(\mathbf Q)|\le k}}\big|\widehat G_{\mathbf P,\mathbf Q}-(G_{\Omega,\tau}^*)_{\mathbf P,\mathbf Q}\big|
                &\le\varepsilon_\chi+\cO\!\left(D_\Omega\operatorname{polylog}(D_\Omega/\varepsilon_\chi)\mathrm{opt}_{\Omega,\tau}\right),\\
                \sup_{\mathbf P\in\cP_{n,\le k}\setminus\{I\}}\big|\widehat h_{\mathbf P}-(h_{\Omega,\tau}^*)_{\mathbf P}\big|
                &\le\varepsilon_\chi+\cO\!\left(D_\Omega\operatorname{polylog}(D_\Omega/\varepsilon_\chi)\mathrm{opt}_{\Omega,\tau}\right).
            \end{aligned}
        \end{equation*}
        The protocol uses $\widetilde{\mathcal O}_{k}(\bar\alpha^2D_{\Omega}^2\varepsilon_\chi^{-2}\log(|\Omega|/\delta))$ samples and $\mathcal O_k(|\Omega|)$ classical postprocessing time.
    \end{proposition}
    \begin{proof}
        Duhamel's formula and diamond-norm contractivity of Lindblad semigroups give, for $t\in[0,T]$,
        \begin{equation*}
            \begin{aligned}
                \|e^{t\cL}-e^{t\cL_k^*}\|_\diamond
                &\le
                t\int_0^1
                \left\|
                e^{(1-s)t\cL}
                (\cL-\cL_k^*)
                e^{st\cL_k^*}
                \right\|_\diamond ds
                \\
                &\le
                t\,\mathrm{opt}.
            \end{aligned}
        \end{equation*}
        Therefore, for every Pauli pair $(\mathbf P,\mathbf Q)$ used in the interpolation, the time traces of $\cL$ and $\cL_k^*$ differ by at most $t\,\mathrm{opt}$. Define
        \begin{equation*}
            f_{\mathbf P,\mathbf Q}^*(t):=2^{-n}\Tr\!\left(\mathbf P e^{t\cL_k^*}(\mathbf Q)\right),
            \qquad
            f_{\mathbf P,\mathbf Q}(t):=2^{-n}\Tr\!\left(\mathbf P e^{t\cL}(\mathbf Q)\right).
        \end{equation*}
        Let $p_{\mathbf P,\mathbf Q}^{*,(d)}$ be the degree-$d$ polynomial approximation to $f_{\mathbf P,\mathbf Q}^*$ supplied by Equation \eqref{eq:poly-uniform-approx}, with uniform approximation error $\varepsilon_{\rm poly}$. Then
        \begin{equation*}
            \sup_{t\in[0,T]}
            \left|
            f_{\mathbf P,\mathbf Q}(t)
            -
            p_{\mathbf P,\mathbf Q}^{*,(d)}(t)
            \right|
            \le
            \varepsilon_{\rm poly}
            +
            T\mathrm{opt}.
        \end{equation*}
        Shadow tomography adds stochastic error $\varepsilon_{\operatorname{Shadow}}$. Thus, with probability at least $1-\delta$, uniformly over the sampled times and all relevant Pauli pairs, the data given to robust interpolation deviate from the polynomial $p_{\mathbf P,\mathbf Q}^{*,(d)}$ by at most
        \begin{equation*}
            \varepsilon_{\operatorname{Shadow}}
            +
            \varepsilon_{\rm poly}
            +
            T\mathrm{opt}.
        \end{equation*}

        Robust regression (Lem.~\ref{lem:robust-time-polynomial-regression}) followed by Markov brothers' inequality, Equation~\eqref{eq:markov-brothers}, gives a derivative error bound. Setting $\varepsilon_{\operatorname{Shadow}},\varepsilon_{\rm poly}\le T\varepsilon_L/(4C_{\operatorname{reg}}d^2)$, we obtain
        \begin{equation*}
            \Bigl| \widehat L_{\mathbf P,\mathbf Q} - L^*_{\mathbf P,\mathbf Q} \Bigr| \le \varepsilon_L+2C_{\operatorname{reg}}d^2\mathrm{opt}\,,
        \end{equation*}
        where $L^*_{\mathbf P,\mathbf Q}=(f_{\mathbf P,\mathbf Q}^*)'(0)$ is the PTM coefficient of $\cL_k^*$. With $\kappa_{\mathbf P}:=\lceil\wt(\mathbf P)/k\rceil$, the polynomial degree is
        \begin{equation*}
            d
            =
            \cO\!\left(
            \kappa_{\mathbf P}
            +
            \log\frac{1}{\varepsilon_{\operatorname{Shadow}}}
            \right).
        \end{equation*}

        The inversion bounds from Corollary \ref{thm:dissipative-learning} and Theorem \ref{thm:overall-coefficient-learning} then propagate this PTM error to the coefficients of $\cL_k^*$. Taking $\varepsilon_L=\Theta_k(\varepsilon_\chi/n^{k-1})$ gives
        \begin{equation*}
            \begin{aligned}
                \sup_{\substack{\mathbf P,\mathbf Q\ne I\\|\operatorname{supp}(\mathbf P)\cup\operatorname{supp}(\mathbf Q)|\le k}}
                \big|\widehat G_{\mathbf{P},\mathbf{Q}}-G^*_{\mathbf{P},\mathbf{Q}}\big|
                &\le
                \varepsilon_\chi
                +
                \cO\!\left(n^{k-1}\operatorname{polylog}(n/\varepsilon_\chi)\mathrm{opt}\right),
                \\
                \sup_{\mathbf P\in\cP_{n,\le k}\setminus\{I\}}
                \big|\widehat h_{\mathbf{P}}-h^*_{\mathbf{P}}\big|
                &\le
                \varepsilon_\chi
                +
                \cO\!\left(n^{k-1}\operatorname{polylog}(n/\varepsilon_\chi)\mathrm{opt}\right).
            \end{aligned}
        \end{equation*}
        Consequently, the proof uses $\widetilde{\mathcal O}_k(\bar\alpha^2n^{2k-2}\varepsilon_\chi^{-2}\log(n/\delta))$ quantum samples and $\mathcal O_k(\operatorname{polylog}(\bar\alpha n/\varepsilon_\chi))$ distinct evolution times, as in the realizable theorem.

        In the supplied-candidate case, if $\Omega=\varnothing$, every independent $\chi$-coefficient of $\cL_{\Omega,\tau}^*$ has absolute value at most $\tau\le\varepsilon_\chi/2$, so the zero estimate satisfies the stated bounds by~\eqref{eq:chiG}. Otherwise, repeat the argument with $\cL_{\Omega,\tau}^*\in\mathfrak L_{k,\bar\alpha,\Omega,\tau}$. Choose $\varepsilon_L=\varepsilon_\chi/(2D_\Omega)$. The recursion in Corollary~\ref{cor:dissipative-learning-stable-sparse} gives error at most $D_\Omega(\varepsilon_L+\rho_\Omega+2C_{\operatorname{reg}}d^2\mathrm{opt}_{\Omega,\tau})$ on $\Omega$, while the error outside $\Omega$ is at most $\tau$. Thus $\tau+D_\Omega\rho_\Omega\le\varepsilon_\chi/2$ leaves the remaining half for statistical error. Equation~\eqref{eq:chiG} transfers the uniform $\chi$-error bound unchanged to $G$ and $h$, giving
        \begin{equation*}
            \begin{aligned}
                \sup_{\substack{\mathbf P,\mathbf Q\ne I\\|\operatorname{supp}(\mathbf P)\cup\operatorname{supp}(\mathbf Q)|\le k}}
                \big|\widehat G_{\mathbf{P},\mathbf{Q}}-(G_{\Omega,\tau}^*)_{\mathbf{P},\mathbf{Q}}\big|
                &\le
                \varepsilon_\chi
                +
                \cO\!\left(D_\Omega\operatorname{polylog}(D_\Omega/\varepsilon_\chi)\mathrm{opt}_{\Omega,\tau}\right),
                \\
                \sup_{\mathbf P\in\cP_{n,\le k}\setminus\{I\}}
                \big|\widehat h_{\mathbf{P}}-(h_{\Omega,\tau}^*)_{\mathbf{P}}\big|
                &\le
                \varepsilon_\chi
                +
                \cO\!\left(D_\Omega\operatorname{polylog}(D_\Omega/\varepsilon_\chi)\mathrm{opt}_{\Omega,\tau}\right).
            \end{aligned}
        \end{equation*}
        The quantum sample and classical postprocessing bounds are those of the supplied-candidate algorithm.
    \end{proof}

    The preceding guarantee measures misspecification globally, in diamond norm. For geometrically local model generators this is stronger than the protocol needs. Since the data consist only of local Pauli overlaps, it suffices to control the residual on the finite light cones generated by those Pauli observables, plus the Lieb--Robinson tail outside the light cone.

    Fix a finite metric space $(V,\dist)$, $V=[n]$, of effective dimension $D$: there are constants $C_{\rm vol},C_{\partial}$, independent of $n$, such that for every $x\in V$ and $R\ge0$,
    \begin{equation*}
        |B_R(x)|\le C_{\rm vol}(1+R)^D,
        \qquad
        |B_R(x)\setminus B_{R-1}(x)|
        \le
        C_{\partial}(1+R)^{D-1},
    \end{equation*}
    with the second bound interpreted trivially at $R=0$. We write $B_R(A):=\{x\in V:\dist(x,A)\le R\}$ for $A\subset V$. We equip each geometrically $(k,r_0)$-local generator $\cL'$ with a fixed local decomposition $\cL'=\sum_Z\cL'_Z$ such that $|Z|\le k$ and $\diam(Z)\le r_0$ whenever $\cL'_Z\neq0$. All restrictions and locality bounds refer to this decomposition. In particular, $\cL'_B:=\sum_{Z\subseteq B}\cL'_Z$ for $B\subseteq V$. Write $\mathfrak L^{\mathrm{geo}}_{k,\bar\alpha,r_0}\subseteq\mathfrak L_{k,\bar\alpha}$ for the geometrically $(k,r_0)$-local generators equipped with such a decomposition whose weighted interaction strength is at most the supplied upper bound $\bar\alpha$.

    The locality input for the proof is the following uniform Heisenberg truncation estimate: for every $\cL'\in\mathfrak L^{\mathrm{geo}}_{k,\bar\alpha,r_0}$, every observable $O$ supported on $A$, every $u\in[0,T]$, $T=(4\bar\alpha k)^{-1}$, and every $R\ge0$,
    \begin{equation}\label{eq:local-agnostic-lr}
        \left\|
        e^{u(\cL')^\dagger}(O)
        -
        e^{u(\cL'_{B_R(A)})^\dagger}(O)
        \right\|_\infty
        \le
        \varepsilon_{\rm LR}(R,u,A)\,\|O\|_\infty .
    \end{equation}
    This estimate follows from the Lieb--Robinson bounds of \cite{Chen2021, AnthonyChen2023}. Applying the time-dependent dissipative bound of \cite[Prop.~E.1]{StilckFranca.2025} to the time-independent case and bounding boundary terms by the shell-growth estimate gives constants depending only on $k,\bar\alpha,r_0,C_{\rm vol},C_{\partial},D$ such that
    \begin{equation}\label{eq:local-agnostic-lr-factorial}
        \varepsilon_{\rm LR}(d'r_0,u,A)
        \le
        C_{\rm LR}|A|(1+d')^{\nu}
        \sum_{\ell\ge d'}
        \frac{(v_{\rm LR}u)^{\ell+1}}{(\ell+1)!}.
    \end{equation}
    Hence, for the $k$-local Pauli observables used by the protocol, we may take the uniform tail
    \begin{equation}\label{eq:local-agnostic-lr-uniform-tail}
        \Delta_{\rm LR}^{d'}
        :=
        \sup_{\substack{\mathbf P\in\cP_{n,\le k}\\0\le u\le T}}
        \varepsilon_{\rm LR}(d'r_0,u,\supp(\mathbf P))
        \le
        C'_{\rm LR}(1+d')^{\nu'}
        \sum_{\ell\ge d'}
        \frac{(v'_{\rm LR}T)^{\ell+1}}{(\ell+1)!},
    \end{equation}
    which may be further relaxed to the familiar form $C''_{\rm LR}(1+d')^{\nu''}e^{v''_{\rm LR}T-\mu d'}$. We keep the sharper factorial tail in Equation \eqref{eq:local-agnostic-lr-uniform-tail}.

    \begin{corollary}[Local misspecification under Lieb--Robinson truncation]\label{cor:local-agnostic-learning}
        Fix $k,r_0$, $d'\in\mathbb N_0$, and the supplied upper bound $\bar\alpha$, and use the model class $\mathfrak L^{\mathrm{geo}}_{k,\bar\alpha,r_0}$. For a residual $\Delta$, define
        \begin{equation*}
            \eta_{d'}(\Delta)
            :=
            \sup_{\mathbf P\in\cP_{n,\le k}}
            \ \sup_{\substack{\|X\|_\infty\le1\\ \supp(X)\subseteq B_{d'r_0}(\supp(\mathbf P))}}
            \|\Delta^\dagger(X)\|_\infty .
        \end{equation*}
        For $\cL'\in\mathfrak L^{\mathrm{geo}}_{k,\bar\alpha,r_0}$, set $\Phi_{d'}(\cL'):=\eta_{d'}(\cL-\cL')+\Delta_{\rm LR}^{d'}\|(\cL-\cL')^\dagger\|_{\infty\rightarrow\infty}$. Let
        \begin{equation*}
            \mathrm{opt}_{\mathrm{loc}}^{d'}
            :=
            \inf_{\cL'\in\mathfrak L^{\mathrm{geo}}_{k,\bar\alpha,r_0}}
            \Phi_{d'}(\cL').
        \end{equation*}
        Choose $\cL_{\mathrm{loc}}^{*,d'}=\cL_{h_{\mathrm{loc}}^{*,d'},G_{\mathrm{loc}}^{*,d'}}\in\operatorname*{argmin}_{\cL'\in\mathfrak L^{\mathrm{geo}}_{k,\bar\alpha,r_0}}\Phi_{d'}(\cL')$. For $\tau>0$ and supplied $\Omega\subseteq\mathcal I_k^\circ$, define $\cL_{\mathrm{loc},\Omega,\tau}^{*,d'}$ analogously over $\mathfrak L^{\mathrm{geo}}_{k,\bar\alpha,r_0}\cap\mathfrak L_{k,\bar\alpha,\Omega,\tau}$, and set $\mathrm{opt}_{\mathrm{loc},\Omega,\tau}^{d'}:=\Phi_{d'}(\cL_{\mathrm{loc},\Omega,\tau}^{*,d'})$. The generic and supplied-candidate bounds of Proposition~\ref{prop:agnostic-learning} hold relative to these minimizers, with the corresponding local approximation errors, provided $\tau+D_\Omega\rho_\Omega(\cL_{\mathrm{loc},\Omega,\tau}^{*,d'})\le\varepsilon_\chi/2$ in the supplied-candidate case.
    \end{corollary}

    \begin{proof}
        The model classes with their fixed local decompositions are compact, and $\Phi_{d'}$ is continuous in finite dimension. Hence the stated minimizers exist. We follow the proof of Proposition~\ref{prop:agnostic-learning}, replacing the global diamond-norm perturbation estimate by a local Heisenberg estimate. Fix $\cL'\in\mathfrak L^{\mathrm{geo}}_{k,\bar\alpha,r_0}$, set $\Delta:=\cL-\cL'$, and let $(\mathbf P,\mathbf Q)\in\mathcal I_k$. Define
        \begin{equation*}
            f_{\mathbf P,\mathbf Q}(t)
            =
            2^{-n}\Tr\!\left(e^{t\cL^\dagger}(\mathbf P)\mathbf Q\right),
            \qquad
            f'_{\mathbf P,\mathbf Q}(t)
            =
            2^{-n}\Tr\!\left(e^{t(\cL')^\dagger}(\mathbf P)\mathbf Q\right),
        \end{equation*}
        so Duhamel's formula in the Heisenberg picture gives, for $t\in[0,T]$,
        \begin{equation*}
            f_{\mathbf P,\mathbf Q}(t)-f'_{\mathbf P,\mathbf Q}(t)
            =
            \frac{t}{2^{n}}\int_0^1
            \Tr\!\left[
            e^{(1-s)t\cL^\dagger}
            \Delta^\dagger
            e^{st(\cL')^\dagger}(\mathbf P)
            \,\mathbf Q
            \right]\,ds .
        \end{equation*}
        Moving $e^{(1-s)t\cL^\dagger}$ to the Schr\"odinger side and using trace-norm contractivity of the CPTP map $e^{(1-s)t\cL}$ on Hermitian inputs,
        \begin{equation*}
            \|e^{(1-s)t\cL}(\mathbf Q)\|_1\le\|\mathbf Q\|_1=2^n .
        \end{equation*}
        Hence
        \begin{equation}\label{eq:local-agnostic-duhamel}
            |f_{\mathbf P,\mathbf Q}(t)-f'_{\mathbf P,\mathbf Q}(t)|
            \le
            t\int_0^1
            \left\|
            \Delta^\dagger e^{st(\cL')^\dagger}(\mathbf P)
            \right\|_\infty ds .
        \end{equation}
        Let $B_{\mathbf P}^{d'}:=B_{d'r_0}(\supp(\mathbf P))$ and
        \begin{equation*}
            \mathbf P_{d'}(st)
            :=
            e^{st(\cL'_{B_{\mathbf P}^{d'}})^\dagger}(\mathbf P).
        \end{equation*}
        The truncated Heisenberg evolution is supported inside $B_{\mathbf P}^{d'}$ and is an operator-norm contraction, hence $\|\mathbf P_{d'}(st)\|_\infty\le1$. By definition of $\eta_{d'}$,
        \begin{equation*}
            \left\|\Delta^\dagger\mathbf P_{d'}(st)\right\|_\infty
            \le
            \eta_{d'}(\Delta).
        \end{equation*}
        On the other hand, Equations~\eqref{eq:local-agnostic-lr} and~\eqref{eq:local-agnostic-lr-uniform-tail} give
        \begin{equation*}
            \left\|
            e^{st(\cL')^\dagger}(\mathbf P)-\mathbf P_{d'}(st)
            \right\|_\infty
            \le
            \Delta_{\rm LR}^{d'} .
        \end{equation*}
        Inserting these two estimates into Equation \eqref{eq:local-agnostic-duhamel} yields
        \begin{equation*}
            |f_{\mathbf P,\mathbf Q}(t)-f'_{\mathbf P,\mathbf Q}(t)|
            \le
            t\left[
            \eta_{d'}(\Delta)
            +
            \Delta_{\rm LR}^{d'}\,
            \|\Delta^\dagger\|_{\infty\rightarrow\infty}
            \right].
        \end{equation*}
        Apply the preceding estimate to $\cL'=\cL_{\mathrm{loc}}^{*,d'}$. It replaces $t\,\mathrm{opt}$ in Proposition~\ref{prop:agnostic-learning} by $t\,\mathrm{opt}_{\mathrm{loc}}^{d'}$, while interpolation and inversion recover the coefficients of the same generator. For the supplied-candidate claim, apply the same argument to $\cL'=\cL_{\mathrm{loc},\Omega,\tau}^{*,d'}$ and use its stated tail assumption.
    \end{proof}

    \begin{rmk}
        The local residual term tests $\cL-\cL'$ only on observables supported in $B_{d'r_0}(\supp(\mathbf P))$, rather than on the full system. The time-trace comparison uses only the truncation estimate~\eqref{eq:local-agnostic-lr} and extends to other Lieb--Robinson tails, including time-dependent model propagators under \cite[Prop.~E.1]{StilckFranca.2025}. Extending the coefficient-learning guarantees additionally requires suitable temporal polynomial approximation and derivative control.
    \end{rmk}

\section{Lower bound for learning Lindbladians in diamond norm}\label{sec:lower-bound-product-measurements}
    Recently, lower bounds for Hamiltonian learning from time-evolution access were established in~\cite{chen2025lower}. While these results demonstrate important obstructions to Hamiltonian learning, their dependence on $n$ and $\varepsilon$ does not match the scaling of our upper bounds. This mismatch is partly due to the significantly stronger access model considered in~\cite{chen2025lower}, which allows more general quantum control and measurement procedures than the experimental setting studied here, typically allowing for Heisenberg-limited scalings, but which at the same time does not assume the boundedness of the (weighted) intersection degree of $\cL$.

    In this section, we prove a lower bound adapted to our access model, in which the learner is restricted to tensor-product input states, short-time Lindbladian evolution, and single-qubit Pauli measurements. Under these restrictions, we obtain a lower bound for Lindbladian learning whose dependence on the number of local terms matches the natural parameter-counting scaling, and whose dependence on the target accuracy exhibits our standard statistical rate $1/\varepsilon_\diamond^2$.

    More precisely, a shot prepares a tensor-product state, evolves it for a time $t\le t_{\max}$, measures every qubit independently in a chosen single-qubit Pauli basis, and records the product outcome. These choices may be adaptive.

    \begin{theorem}[Single-coefficient lower bound] \label{thm:single-chi-lower-bound}
        Fix a nonidentity Pauli string $\mathbf R\in\mathcal P_n$ satisfying $|\operatorname{supp}(\mathbf R)|\le k$. Consider the one-parameter Hamiltonian subfamily
        \begin{equation*}
            \mathcal L_\theta(\rho):=-i[\theta\mathbf R,\rho],\qquad\theta\in\mathbb R\,.
        \end{equation*}
        Any adaptive protocol using the product-Pauli shots just defined and outputting an estimator $\widehat\chi_{\mathbf R,I}$ satisfying
        \begin{equation*}
            \Pr_\theta\!\left[|\widehat\chi_{\mathbf R,I}-\chi_{\mathbf R,I}|\le \varepsilon_1 \right]\ge \frac23
        \end{equation*}
        uniformly over $\theta\in\{-2\varepsilon_1,+2\varepsilon_1\}$ requires $\frac{c}{t_{\max}^2\varepsilon_1^2}$ samples, for some universal constant $c>0$.
    \end{theorem}
    \begin{proof}
        In the Pauli-superoperator basis, the generator $\cL_\theta$ satisfies $\chi_{\mathbf R,I}=-i\theta$, $\chi_{I,\mathbf R}=i\theta$, and all other Hamiltonian $\chi$-coefficients vanish, i.e.~$\chi_{\mathbf P,I}=\chi_{I,\mathbf P}=0$ for all $\mathbf P\in\cP_n\backslash \{R\}$. Hence learning $\chi_{\mathbf R,I}$ to accuracy $\varepsilon_1$ is equivalent to learning $\theta$ to accuracy $\varepsilon_1$. It is enough to distinguish the two hypotheses $\theta_+=2\varepsilon_1$ versus $\theta_-=-2\varepsilon_1$. Under these hypotheses, the corresponding values of $\chi_{\mathbf R,I}$ are separated by $4\varepsilon_1$. Therefore any estimator satisfying $|\widehat\chi_{\mathbf R,I}-\chi_{\mathbf R,I}|\le \varepsilon_1$ identifies the correct sign of $\theta$.

        We now bound the distinguishability of the two hypotheses in one experiment. Fix an arbitrary allowed experiment: a tensor-product input state $\rho$, an evolution time $t\le t_{\max}$, and a tensor product of single-qubit Pauli measurements. Let $P_+$ and $P_-$ be the resulting outcome distributions under $\theta_+$ and $\theta_-$, respectively. The two evolved states are
        \begin{equation*}
            \rho_+ = e^{-it\theta_+\mathbf R}\rho e^{it\theta_+\mathbf R}, \qquad \rho_- = e^{-it\theta_-\mathbf R}\rho e^{it\theta_-\mathbf R}.
        \end{equation*}
        By unitary invariance of fidelity,
        \begin{equation*}
            F(\rho_+,\rho_-) = F\!\left(\rho, e^{-it(\theta_+-\theta_-)\mathbf R}\rho e^{it(\theta_+-\theta_-)\mathbf R}\right).
        \end{equation*}
        Set $\Delta:=\theta_+-\theta_-=4\varepsilon_1$. First suppose $t_{\max}\varepsilon_1\le1/8$, so $0\le t\Delta\le1/2$ and $\cos(t\Delta)>0$ for every allowed query. Using Uhlmann's theorem, for any state $\rho$,
        \begin{equation*}
            \begin{aligned}
                F\!\left(\rho,e^{-it\Delta\mathbf R}\rho e^{it\Delta\mathbf R}\right)&\ge\left|\Tr\!\left(\rho e^{-it\Delta\mathbf R}\right)\right|^2\\
                &=\left|\cos(t\Delta)-i\sin(t\Delta)\Tr(\rho\mathbf R)\right|^2\\
                &\ge\cos^2(t\Delta).
            \end{aligned}
        \end{equation*}
        Since a measurement cannot decrease fidelity, the classical fidelity of the outcome distributions satisfies
        \begin{equation*}
            F_{\mathrm{cl}}(P_+,P_-):=\left(\sum_x\sqrt{P_+(x)P_-(x)}\right)^2\ge\cos^2(t\Delta).
        \end{equation*}
        Equivalently, their Bhattacharyya coefficient $\operatorname{BC}(P_+,P_-):=\sum_x\sqrt{P_+(x)P_-(x)}$ satisfies
        \begin{equation*}
            \operatorname{BC}(P_+,P_-)\ge\cos(t\Delta).
        \end{equation*}
        For every past transcript, the conditional one-shot Bhattacharyya coefficient is at least $c_0:=\cos(t_{\max}\Delta)$. Iterating the defining sum over the adaptive transcript gives
        \begin{equation*}
            \operatorname{BC}(\mathbb P_+^{(N)},\mathbb P_-^{(N)})\ge c_0^N,
        \end{equation*}
        which implies by $\operatorname{TV}(P,Q)\le \sqrt{2(1-\operatorname{BC}(P,Q))}$
        \begin{equation*}
            \operatorname{TV}(\mathbb P_+^{(N)},\mathbb P_-^{(N)})\le\sqrt{N t_{\max}^2\Delta^2}=4\sqrt{N}\,t_{\max}\varepsilon_1.
        \end{equation*}
        Since any test distinguishing the two hypotheses with success probability at least $2/3$ must have $ \operatorname{TV}(\mathbb P_+^{(N)},\mathbb P_-^{(N)})\ge\frac13$, we directly get
        \begin{equation*}
            N\ge\frac{c}{t_{\max}^2\varepsilon_1^2}
        \end{equation*}
        for a universal constant $c>0$. If $t_{\max}\varepsilon_1>1/8$, an estimator valid for both disjoint hypotheses must use at least one shot, while $(t_{\max}^2\varepsilon_1^2)^{-1}<64$. Decreasing the universal constant proves the claim in this complementary regime.
    \end{proof}

    \noindent Next, we turn our attention to showing a lower bound on the sample complexity for learning $k$-local Lindbladians in diamond norm.

    \begin{theorem}[Diamond-norm lower bound]\label{thm:diamond-lower-bound}
        Fix $1\le k\le n$, $\varepsilon_\diamond>0$, and $t_{\max}>0$, and let $D:=\binom nk$. For every $k$-subset $S\subseteq[n]$, define the Pauli-$Z$ string
        \begin{equation*}
            Z_S:=\bigotimes_{j\in S}Z_j .
        \end{equation*}
        Set $\zeta:=4\varepsilon_\diamond/\sqrt D$. For each sign vector $v=(v_S)_{S\in\binom{[n]}k}\in\{\pm1\}^{D}$, with $\binom{[n]}k:=\{S\subset [n]:|S|=k\}$, define the $k$-local Hamiltonian $H_v:=\zeta\sum_{S\in\binom{[n]}k}v_S Z_S$ and the corresponding Hamiltonian generator $\mathcal L_v(\rho):=-i[H_v,\rho]$. No system-size-independent bound on $\bar\alpha$ is imposed on this packing. Any adaptive protocol using product-Pauli shots and outputting an estimator $\widehat{\mathcal L}$ satisfying
        \begin{equation*}
            \Pr_{\mathcal L_v}\!\left[\|\widehat{\mathcal L}-\mathcal L_v\|_\diamond\le \varepsilon_\diamond\right]\ge \frac23
        \end{equation*}
        uniformly over $v$, requires
        \begin{equation*}
            N\ge\Omega_k\!\left(\frac{n^k}{t_{\max}^2\varepsilon_\diamond^2}\right).
        \end{equation*}
    \end{theorem}
    \begin{proof}
        Enumerate the subsets in $\binom{[n]}k$ as $S_1,\ldots,S_D$ and recall that $\zeta=4\varepsilon_\diamond/\sqrt D$. It suffices to prove the claim in the regime $t_{\max}\varepsilon_\diamond\le \sqrt D/32$. If this condition fails, then $D/(t_{\max}^2\varepsilon_\diamond^2)\le 32^2$, so the asserted lower bound is only a universal constant after adjusting the implicit constant. In this complementary regime, the two generators indexed by $v=(1,\ldots,1)$ and $-v$ are separated by at least $2\zeta\sqrt D=8\varepsilon_\diamond$, so their closed $\varepsilon_\diamond$-balls are disjoint. A zero-shot estimator has the same output law under both hypotheses and therefore cannot succeed with probability at least $2/3$ on both. Thus $N\ge1$, which proves the claimed bound after decreasing its universal constant. In the remaining regime, the choice of $\zeta$ implies $2t_{\max}\zeta\le1/4$, so all cosine factors used below are positive. Let $V$ be uniform on the full hypercube $\{\pm1\}^D$, and consider the Hamiltonians
        \begin{equation*}
            H_v=\zeta\sum_{j=1}^D v_j Z_{S_j},
            \qquad
            \cL_v(\rho)=-i[H_v,\rho].
        \end{equation*}
        For each realization $v,w\in\{\pm1\}^D$, set $\Delta H:=H_v-H_w$. Since $\Delta H$ is traceless,
        \begin{equation}\label{eq:separation-bound-1}
            \|\cL_v-\cL_w\|_\diamond
            =
            \|[\Delta H,\bullet]\|_\diamond
            =
            \lambda_{\max}(\Delta H)-\lambda_{\min}(\Delta H)
            \ge
            \|\Delta H\|_{\infty}.
        \end{equation}
        Pauli orthogonality gives
        \begin{equation}\label{eq:separation-bound-2}
            \|\Delta H\|_{\infty}
            \ge
            \left(2^{-n}\Tr(\Delta H^2)\right)^{1/2}
            = 2\zeta\sqrt{d_H(v,w)},
        \end{equation}
        where $d_H(v,w):=|\{S\in\binom{[n]}k:v_S\ne w_S\}|$ is the Hamming distance. Define the nearest-neighbor decoder
        \begin{equation*}
            \widehat V
            :=
            \operatorname{argmin}_{w\in\{\pm1\}^D}
            \|\widehat{\cL}-\cL_w\|_\diamond .
        \end{equation*}
        On the event $\|\widehat{\cL}-\cL_V\|_\diamond\le\varepsilon_\diamond$, the triangle inequality and Equations~\eqref{eq:separation-bound-1}--\eqref{eq:separation-bound-2} imply
        \begin{equation*}
            d_H(\widehat V,V)
            \le
            \left(\frac{\varepsilon_\diamond}{\zeta}\right)^2
            =
            \frac{D}{16}.
        \end{equation*}
        Since this event has probability at least $2/3$ and always $d_H(\widehat V,V)\le D$, we have
        \begin{equation*}
            \sum_{j=1}^D\Pr[\widehat V_j\neq V_j]=\mathbb E[d_H(\widehat V,V)]\le \frac{3D}{8}.
        \end{equation*}
        For each coordinate $j$, let $Q_j^\pm$ be the law of the full adaptive transcript $Y$ conditioned on $V_j=\pm1$. Assouad's coordinate-wise testing bound \cite[Lem.~2]{Yu1997} gives
        \begin{equation*}
            \Pr[\widehat V_j\neq V_j]
            \ge
            \frac12\left(1-\operatorname{TV}(Q_j^+,Q_j^-)\right).
        \end{equation*}
        Summing over $j$ and using the expected-Hamming bound yields
        \begin{equation*}
            \frac1D\sum_{j=1}^D\operatorname{TV}(Q_j^+,Q_j^-)\ge\frac14,
            \qquad
            \text{hence}\qquad
            \max_j\operatorname{TV}(Q_j^+,Q_j^-)\ge\frac14.
        \end{equation*}

        It remains to upper-bound this coordinate total variation. Fix $j\in[D]$, and write $V_{-j}:=(V_1,\ldots,V_{j-1},V_{j+1},\ldots,V_D)$ for all coordinates except the $j$-th one. For $v_{-j}\in\{\pm1\}^{D-1}$, write $Q^{(v_{-j},\pm)}$ for the transcript law of the fixed generator with $V_{-j}=v_{-j}$ and $V_j=\pm1$. Since $V_1,\ldots,V_D$ are independent Rademacher random variables, conditioning on $V_j=\pm1$ leaves $V_{-j}$ uniform on $\{\pm1\}^{D-1}$, and hence
        \begin{equation*}
            Q_j^\pm
            =
            2^{-(D-1)}
            \sum_{v_{-j}\in\{\pm1\}^{D-1}}
            Q^{(v_{-j},\pm)},
        \end{equation*}
        with identical mixture weights. Joint convexity of total variation gives
        \begin{equation*}
            \operatorname{TV}(Q_j^+,Q_j^-)
            \le
            \max_{v_{-j}}
            \operatorname{TV}\!\left(Q^{(v_{-j},+)},Q^{(v_{-j},-)}\right).
        \end{equation*}
        Fix $v_{-j}$. In one shot at time $t\le t_{\max}$, all $Z_S$'s commute, so the Hamiltonian terms common to the two hypotheses cancel by unitary invariance. Thus the two states differ only by the sign of $\zeta Z_{S_j}$. The same Uhlmann/fidelity calculation as in Theorem~\ref{thm:single-chi-lower-bound} gives a classical Bhattacharyya coefficient at least $\cos(2t\zeta)$. This bound holds conditionally on every past transcript, hence iterating the conditional bound along the adaptive transcript gives
        \begin{equation*}
            \operatorname{BC}
            \!\left(
            Q^{(v_{-j},+)},
            Q^{(v_{-j},-)}
            \right)
            \ge
            \cos(2t_{\max}\zeta)^N\,.
        \end{equation*}
        Using $\operatorname{TV}\le\sqrt{2(1-\operatorname{BC})}$, $1-\cos^N x\le N(1-\cos x)$, and $1-\cos x\le x^2/2$, we obtain
        \begin{equation*}
            \operatorname{TV}(Q_j^+,Q_j^-)
            \le
            2\sqrt N\,t_{\max}\zeta\,.
        \end{equation*}
        Combining this with the previous lower bound on $\max_j\operatorname{TV}(Q_j^+,Q_j^-)$ gives
        \begin{equation*}
            \frac14\le 2\sqrt N\,t_{\max}\zeta,
            \qquad\text{so}\qquad
            N\ge \frac{c}{t_{\max}^2\zeta^2}
            =
            \Omega\!\left(
            \frac{D}{t_{\max}^2\varepsilon_\diamond^2}
            \right)
            =
            \Omega_k\!\left(
            \frac{n^k}{t_{\max}^2\varepsilon_\diamond^2}
            \right)\,,
        \end{equation*}
        which finishes the proof.
    \end{proof}
    \begin{rmk}
        For the packing above, the weighted interaction strength scales as $\alpha=\Theta_k\!\left(\varepsilon_\diamond n^{k/2-1}\right)$ and is therefore not constant for $k>2$. Normalizing the generator by $\alpha$ rescales both the maximal time and the target accuracy, so time rescaling alone does not strengthen the lower bound at fixed normalized accuracy.
    \end{rmk}

\section{Application: local-observable verification}\label{sec:local-observable-verification}
    Most papers in the Hamiltonian learning literature focus on the regime where the number of samples grows polylogarithmically with the system size. However, the previous section shows that global diamond-norm recovery can require polynomially many samples in our access model. As we show next, one can nevertheless learn local evolution models that approximate all local marginals at short times. The lower bound above concerns learning an entire generator in diamond norm. This is stronger than what is needed to verify a local observable in the spirit of~\cite{kraft2025boundederrorquantumsimulationhamiltonian}. The relevant point is that, under a Lieb--Robinson type truncation estimate, the expectation of a local observable at time $t$ depends only on the generator parameters in a neighborhood of the observable, up to a controllable tail. Hence one can learn a local effective generator on that neighborhood and compare the simulator to the learned local model. The required coefficient precision is set by the size of this neighborhood, not by the total system size $n$.

    Let $A=\supp(O)$ and let $B_R(A)\subseteq[n]$ denote the radius-$R$ neighborhood of $A$ in the metric or interaction graph relevant for the Lieb--Robinson bound. For $B\subseteq[n]$, let $\Pi_B$ be the orthogonal projector in Pauli-label space defined by $\Pi_B|\mathbf P\rangle=|\mathbf P\rangle$ when $\mathbf P\ne I$ and $\supp(\mathbf P)\subseteq B$, and $\Pi_B|\mathbf P\rangle=0$ otherwise. Thus $\Pi_BG\Pi_B$ keeps exactly the rows and columns of $G$ indexed by nonidentity Pauli strings supported in $B$. Define the coefficient restriction
    \begin{equation*}
        H_B^{\rm can}:=\sum_{\substack{\mathbf P\ne I:\\\supp(\mathbf P)\subseteq B}}h_{\mathbf P}\mathbf P,
        \qquad
        G_B^{\rm can}:=\Pi_BG\Pi_B,
        \qquad
        \cL_B^{\rm can}:=\cL_{H_B^{\rm can},G_B^{\rm can}}.
    \end{equation*}
    We regard $\cL_B^{\rm can}$ either as a generator on $B$ or as acting trivially on $B^c$.

    \begin{lemma}[Local coefficient restrictions]\label{lem:canonical-local-restriction}
        Suppose $\cL$ admits a $k$-local GKSL decomposition $\cL=\sum_e\cL_e$ with local Hamiltonians $H_e$ and local Kossakowski matrices $X^e\succeq0$, so that $H=\sum_eH_e$ and $G=\cA((X^e)_e)$. Then $\cL_B^{\rm can}$ is a valid $k$-local Lindblad generator for every $B\subseteq[n]$, and $G_B^{\rm can}$ admits a positive decomposition on subsets of $B$ of cardinality at most $k$. The restrictions are consistent: if $B\subseteq C$, then $(\cL_C^{\rm can})_B^{\rm can}=\cL_B^{\rm can}$. If the original decomposition is geometrically $(k,r_0)$-local and has weighted interaction strength at most $\bar\alpha$, the induced decompositions have range at most $r_0$ and weighted interaction strength at most $C_k\bar\alpha$. Here $C_k<\infty$ is the maximal norm amplification of the coefficient restriction $\cL_e\mapsto\cL_e^{\,B}$ over $|e|\le k$ and $B\subseteq e$. It depends only on $k$ and the Pauli normalization.
    \end{lemma}
    \begin{proof}
        For each local block, let $X^{e,B}:=X^e[\mathcal P_{e\cap B}^{\circ},\mathcal P_{e\cap B}^{\circ}]$ be the principal submatrix indexed by Pauli strings supported in $e\cap B$. Principal submatrices preserve positivity and the embeddings satisfy
        \begin{equation*}
            X^{e,B}\succeq0,\qquad
            \Pi_B\iota_e(X^e)\Pi_B=\iota_{e\cap B}(X^{e,B}).
        \end{equation*}
        Consequently,
        \begin{equation*}
            G_B^{\rm can}=\sum_e\iota_{e\cap B}(X^{e,B}),
        \end{equation*}
        which is a local-positive decomposition on regions of size at most $k$. Restricting the Pauli coefficients of the local Hamiltonian contributions gives $H_B^{\rm can}$ and proves validity. Consistency follows from $\Pi_B\Pi_C=\Pi_B$ for $B\subseteq C$ and the analogous identity for Hamiltonian coefficient restriction. The map $e\mapsto e\cap B$ cannot increase geometric range. Let $\cL_e^{\,B}$ denote the contribution reconstructed from $X^{e,B}$ and the correspondingly restricted Hamiltonian coefficients. On at most $k$ qubits this restriction and reconstruction is a linear map on a fixed finite-dimensional space, so norm equivalence gives
        \begin{equation*}
            \|(\cL_e^{\,B})^\dagger\|_{2\to2}\le C_k\|\cL_e^\dagger\|_{2\to2}.
        \end{equation*}
        Only original terms containing a site $u$ can yield restricted terms containing $u$. Summing the last bound over such terms proves the $C_k\bar\alpha$ strength bound.
    \end{proof}
    We state the required truncation estimate as a general assumption to allow different spatial-decay bounds. The finite-range case is justified below.

    \begin{assumption}[Local truncation for observables]
        \label{ass:local-observable-truncation}
        There is a function $\Delta_{\rm LR}(R,t,A)$ such that the coefficient restrictions of Lemma~\ref{lem:canonical-local-restriction} satisfy, for every region $B\supseteq A$, every observable $O$ supported on $A$, and $R:=\dist(A,B^c)$,
        \begin{equation*}
            \left\|
            e^{t\cL^\dagger}(O)
            -e^{t(\cL_B^{\rm can})^\dagger}(O)
            \right\|_{\infty}
            \le
            \Delta_{\rm LR}(R,t,A)\,\|O\|_{\infty}.
        \end{equation*}
    \end{assumption}

    \paragraph{Finite-range generators.}
    If the decomposition in Lemma~\ref{lem:canonical-local-restriction} has range at most $r_0>0$ and weighted strength at most $\bar\alpha$, Assumption~\ref{ass:local-observable-truncation} holds with
    \begin{equation}\label{eq:finite-range-coefficient-truncation}
        \Delta_{\rm LR}(R,t,A)
        :=\min\left\{2,\,
        C_k'|A|e^{-R/r_0}
        \left(e^{C_k'\bar\alpha t}-1\right)\right\},
    \end{equation}
    where $C_k'$ depends only on $k$. This follows by keeping the interaction weights in the path argument of~\cite[proof of Thm.~1]{BarthelKliesch2012}, also used in~\cite[Prop.~E.1, Eq.~(171)]{StilckFranca.2025}. To verify it for coefficient restrictions, set $w_e:=4^k\|\cL_e^\dagger\|_{2\to2}$, $w:=4^k\bar\alpha$, and $\cD_e:=\cL_e^\dagger-(\cL_e^{\,B})^\dagger$. Norm comparison on at most $k$ qubits, including extension by the identity map, gives $\|\cL_e^\dagger\|_{\infty\to\infty}\le w_e$ and $\|\cD_e\|_{\infty\to\infty}\le(1+C_k)w_e$. Moreover, $\cD_e$ is supported on $e$, satisfies $\cD_e(I)=0$, and vanishes for $e\subseteq B$. Duhamel's formula and the cited recursion, applied to these local identity-annihilating maps, sum chains of overlapping interaction regions from $A$ to $B^c$. Every contributing chain has $\ell r_0\ge R$. Summing the weights of all $\ell$-region chains from $A$ gives at most $(1+C_k)|A|w(kw)^{\ell-1}$, since $\sum_{e\ni u}w_e\le w$ and each region contains at most $k$ sites. Integrating over time yields
    \begin{equation*}
        \begin{aligned}
            &\left\|e^{t\cL^\dagger}(O)
            -e^{t(\cL_B^{\rm can})^\dagger}(O)\right\|_\infty\le
            \frac{(1+C_k)|A|}{k}\|O\|_\infty
            \sum_{\ell\ge\max\{1,\lceil R/r_0\rceil\}}
            \frac{(kwt)^\ell}{\ell!}.
        \end{aligned}
    \end{equation*}
    The sum is at most $e^{-R/r_0}(e^{\mathrm e kwt}-1)$, which, together with contractivity, proves~\eqref{eq:finite-range-coefficient-truncation}. Polynomial volume growth is needed only to bound neighborhood sizes in the sample estimates. Long-range variants require an appropriate Lieb--Robinson tail and the separate coefficient-tail control discussed after the proof below.

    \begin{theorem}[Learning local verification parameters]
        \label{thm:local-observable-verification}
        Fix a nonidentity observable $O$ supported on at most $q=O(1)$ sites, a product input $\rho_0$, an evolution time $t>0$, and $\varepsilon,\delta\in(0,1)$. Assume that the metric and a finite interaction range $r_0=O(1)$ are known, that $\cL$ admits a geometrically $(k,r_0)$-local GKSL decomposition with supplied strength upper bound $\bar\alpha$ as in Lemma~\ref{lem:canonical-local-restriction}, and that Assumption~\ref{ass:local-observable-truncation} holds for its coefficient restrictions. Set $A:=\supp(O)$, choose $B_{\operatorname{in}}:=B_{R_{\operatorname{in}}}(A)$ such that $\Delta_{\rm LR}(\dist(A,B_{\operatorname{in}}^c),t,A)\le\varepsilon/3$, and let $b:=|B_{\operatorname{in}}|$. Choose $R_{\operatorname{out}}\ge R_{\operatorname{in}}+r_0$, set $B_{\operatorname{out}}:=B_{R_{\operatorname{out}}}(A)$ and $m:=|B_{\operatorname{out}}|$. Define the local geometric coefficient-degree bound
        \begin{equation}\label{eq:local-geometric-degree}
            \mathsf d_{\operatorname{out}}
            :=
            \max_{u\in B_{\operatorname{out}}}
            \sum_{s=1}^k
            \binom{|B_{r_0}(\{u\})\cap B_{\operatorname{out}}|-1}{s-1}
            (3^s+15^s).
        \end{equation}
        There is a protocol that learns a valid $k$-local Lindblad generator $\widehat\cL_{B_{\operatorname{in}}}$, supported on $B_{\operatorname{in}}$ and estimating $\cL_{B_{\operatorname{in}}}^{\rm can}$, such that, with probability at least $1-\delta$,
        \begin{equation*}
            \left|
            \Tr\!\left[O e^{t\cL}(\rho_0)\right]
            -
            \Tr\!\left[
            O e^{t\widehat\cL_{B_{\operatorname{in}}}}(\rho_{0,B_{\operatorname{in}}})
            \right]
            \right|
            \le
            \varepsilon\,\|O\|_{\infty}.
        \end{equation*}
        The number of samples is
        \begin{equation*}
            \widetilde{\mathcal O}_{k}\!\left(
            \frac{\bar\alpha^2t^2b^2\mathsf d_{\operatorname{out}}^2}
            {\varepsilon^2}
            \log\frac{1}{\delta}
            \right).
        \end{equation*}
        Since $\mathsf d_{\operatorname{out}}=\mathcal O_k(m^{k-1})$ and $b\le m$, this is at most $\widetilde{\mathcal O}_k(\bar\alpha^2t^2m^{2k} \varepsilon^{-2}\log(1/\delta))$. The hidden constants depend on $k$ and the Lieb--Robinson constants but not on $n$ except through $m$. For polynomial volume growth and exponential Lieb--Robinson tails, $m=\operatorname{poly}(t,\log(1/\varepsilon))$ up to polylogarithmic factors, so $t=\operatorname{polylog}(n)$ and fixed $\varepsilon$ give polylogarithmic sample complexity.
    \end{theorem}
    \begin{proof}
        
        Restricting the coefficient array deletes entries, so Lemmas~\ref{lem:coefficient-strength} and~\ref{lem:coeff-to-diamond} give $\|\cL_{B_{\operatorname{in}}}^{\rm can}\|_\diamond\le2bc_k\bar\alpha$. If $2bc_k\bar\alpha t\le\varepsilon/3$, return the zero generator without measurements. Its generator error is at most $\varepsilon/(3t)$, so the truncation and Duhamel bounds below prove the claim. Otherwise $\bar\alpha tb/\varepsilon>1/(6c_k)$, so integer rounding is absorbed in the stated sample bound.

        By Assumption~\ref{ass:local-observable-truncation}, replacing the full evolution by that generated by $\cL_{B_{\operatorname{in}}}^{\rm can}$ contributes at most $\varepsilon\|O\|_{\infty}/3$.  Let $\mathcal V_{\operatorname{in}}$ be the coordinates of the coefficient array $x=(h,G)$ supported in $B_{\operatorname{in}}$, and let $\mathcal U_{\operatorname{in}}$ be the $\chi$-coordinates used by $\Psi$ to recover them. Define $\mathcal Q_{\operatorname{in}}$ as the set of all diagonal extensions of indices in $\mathcal U_{\operatorname{in}}$ with union-support size at most $k$ and support contained in $B_{\operatorname{out}}$, including the original indices.
        Let $x_{\operatorname{rel}}$ be the algebraic restriction of the full coefficient array to coordinates whose support meets $B_{\operatorname{in}}$.  Every nonzero such coordinate has diameter at most $r_0$ and is therefore supported in $B_{r_0}(B_{\operatorname{in}})\subseteq B_{\operatorname{out}}$. Moreover, every queried extension contains a nonempty anchor in $B_{\operatorname{in}}$.  Hence $(\mathsf Z(Fx_{\operatorname{rel}}))_w =(\mathsf Z(Fx))_w=y_w$ for every $w\in\mathcal Q_{\operatorname{in}}$, all required strict extensions remain in this family, and omitted extensions vanish.  Consequently,
        \begin{equation*}
            R_{\operatorname{in}}\Psi\mathsf Z^{-1}y
            =x_{B_{\operatorname{in}}}^{\rm can},
        \end{equation*}
        where $x_{B_{\operatorname{in}}}^{\rm can}$ is the physical coefficient array of $\cL_{B_{\operatorname{in}}}^{\rm can}$ and $R_{\operatorname{in}}$ restricts to $\mathcal V_{\operatorname{in}}$.

        For each required local PTM pair, the protocol estimates the full-system trace
        \begin{equation*}
            f_{\mathbf P,\mathbf Q}(s)
            :=2^{-n}\Tr\!\left[\mathbf P e^{s\cL}(\mathbf Q)\right].
        \end{equation*}
        Since $f_{\mathbf P,\mathbf Q}'(0)=L_{\mathbf P,\mathbf Q}$, the process-shadow and interpolation routine applies without replacing the measured dynamics by a truncated one.  Estimate the resulting local coefficient sums uniformly to absolute error $\nu$, threshold them at $2\nu$, apply the restricted map $\Psi\mathsf Z^{-1}$, and symmetrize. The array $x_{\operatorname{rel}}$ has degree at most $\mathsf d_{\operatorname{out}}$, so $Fx_{\operatorname{rel}}$ has degree at most $3\mathsf d_{\operatorname{out}}$ by \eqref{eq:physical-chi-map-bounds}. Taking a threshold below its smallest nonzero coefficient (or any threshold if it vanishes) makes its sitewise tail zero. The proof of Lemma~\ref{lem:threshold-before-inversion}, restricted to the extension-closed family above, therefore gives
        \begin{equation}\label{eq:local-observable-coefficient-error}
            \|(2(\widehat h_{\operatorname{in}}-h_{B_{\operatorname{in}}}^{\rm can}),
            \widehat G_{\operatorname{in}}-G_{B_{\operatorname{in}}}^{\rm can})\|_{\mathrm{loc},1}
            \le9\,4^k\mathsf d_{\operatorname{out}}\nu.
        \end{equation}
        Here $(\widehat h_{\operatorname{in}},\widehat G_{\operatorname{in}})$ estimates the coefficients $(h_{B_{\operatorname{in}}}^{\rm can},G_{B_{\operatorname{in}}}^{\rm can})$ of $\cL_{B_{\operatorname{in}}}^{\rm can}$. Apply the local physical projection on $B_{\operatorname{in}}$ with the strength upper bound $C_k\bar\alpha$ from Lemma~\ref{lem:canonical-local-restriction}.  If its additive accuracy is $\xi$, Proposition~\ref{prop:local-norm-physical-projection} and Lemma~\ref{lem:coeff-to-diamond} yield
        \begin{equation*}
            \|\widehat\cL_{B_{\operatorname{in}}}
            -\cL_{B_{\operatorname{in}}}^{\rm can}\|_\diamond
            \le2b\left(
            18\,4^k\mathsf d_{\operatorname{out}}\nu+\xi
            \right).
        \end{equation*}
        Taking
        \begin{equation*}
            \nu:=\frac{\varepsilon}
            {216\,4^k t b\mathsf d_{\operatorname{out}}},
            \qquad
            \xi:=\frac{\varepsilon}{12tb},
        \end{equation*}
        makes this error at most $\varepsilon/(3t)$.  The derivative estimates use
        \begin{equation*}
            \widetilde{\mathcal O}_k\!\left(
            \frac{\bar\alpha^2t^2b^2
            \mathsf d_{\operatorname{out}}^2}{\varepsilon^2}
            \log\frac1\delta
            \right)
        \end{equation*}
        channel uses, as claimed.

        Since the Heisenberg semigroups generated by Lindbladians are unital completely positive maps, they are contractions in operator norm. Duhamel's formula therefore gives
        \begin{equation*}
            \begin{aligned}
                &\left\|
                e^{t\widehat\cL_{B_{\operatorname{in}}}^\dagger}(O)
                -
                e^{t(\cL_{B_{\operatorname{in}}}^{\rm can})^\dagger}(O)
                \right\|_{\infty} \\
                &\qquad\le
                \int_0^t
                \left\|
                e^{s\widehat\cL_{B_{\operatorname{in}}}^\dagger}
                (\widehat\cL_{B_{\operatorname{in}}}^\dagger-(\cL_{B_{\operatorname{in}}}^{\rm can})^\dagger)
                e^{(t-s)(\cL_{B_{\operatorname{in}}}^{\rm can})^\dagger}(O)
                \right\|_{\infty}\,ds \\
                &\qquad\le
                t\,
                \|\widehat\cL_{B_{\operatorname{in}}}-\cL_{B_{\operatorname{in}}}^{\rm can}\|_\diamond
                \|O\|_{\infty}
                \le
                \frac{\varepsilon}{3}\|O\|_{\infty}.
            \end{aligned}
        \end{equation*}
        Combining the truncation and learned-model errors by the triangle inequality gives a total error of at most $2\varepsilon\|O\|_{\infty}/3$, and hence proves the claim.
    \end{proof}
    The finite-range hypothesis is used to close every nonzero diagonal-extension query inside $B_{\operatorname{out}}$. For decaying long-range generators, an analogous statement would require an additional coefficient-tail estimate for extensions crossing $B_{\operatorname{out}}$. A Lieb--Robinson bound on the time trace alone does not provide that estimate. This immediately gives the following corollary:
    \begin{corollary}[Simultaneous local marginals]
        \label{cor:simultaneous-local-marginals}
        Assume the metric, finite-range, strength, and coefficient-truncation hypotheses of Theorem~\ref{thm:local-observable-verification}. Fix $q=O(1)$, a product state $\rho_0$, $t>0$, and $\varepsilon,\delta\in(0,1)$. Define
        \begin{equation*}
            \mathcal A_q:=\{A\subseteq[n]:1\le |A|\le q\},
            \qquad
            N_q:=|\mathcal A_q|=\mathcal O_q(n^q),
            \qquad
            \varepsilon_{\rm loc}:=2^{-q}\varepsilon.
        \end{equation*}
        For every $A\in\mathcal A_q$, choose regions $B_A^{\operatorname{in}}\subseteq B_A^{\operatorname{out}}$ such that
        \begin{equation*}
            \Delta_{\rm LR}\!\left(
            \dist(A,(B_A^{\operatorname{in}})^c),t,A
            \right)
            \le\frac{\varepsilon_{\rm loc}}{3},
            \qquad
            B_{r_0}(B_A^{\operatorname{in}})
            \subseteq B_A^{\operatorname{out}},
        \end{equation*}
        and let $m_*:=\max_A|B_A^{\operatorname{out}}|$. Also set
        \begin{equation*}
            b_*:=\max_A|B_A^{\operatorname{in}}|,
            \qquad
            \mathsf d_*:=
            \max_{\substack{A\in\mathcal A_q\\
            u\in B_A^{\operatorname{out}}}}
            \sum_{s=1}^k
            \binom{|B_{r_0}(\{u\})\cap B_A^{\operatorname{out}}|-1}{s-1}
            (3^s+15^s).
        \end{equation*}
        A single shared data set yields valid $k$-local generators $\widehat\cL_{B_A^{\operatorname{in}}}$ supported on $B_A^{\operatorname{in}}$ such that, with probability at least $1-\delta$,
        \begin{equation*}
            \max_{A\in\mathcal A_q}
            \sup_{\substack{O_A=O_A^\dagger\\\|O_A\|_\infty\le1}}
            \left|
            \Tr\!\left[
            (O_A\otimes I_{[n]\setminus A})e^{t\cL}(\rho_0)
            \right]
            -
            \Tr\!\left[
            (O_A\otimes I_{B_A^{\operatorname{in}}\setminus A})
            e^{t\widehat\cL_{B_A^{\operatorname{in}}}}
            (\rho_{0,B_A^{\operatorname{in}}})
            \right]
            \right|
            \le\varepsilon.
        \end{equation*}
        The number of quantum samples is
        \begin{equation*}
            \widetilde{\mathcal O}_{k,q}\!\left(
            \frac{\bar\alpha^2t^2b_*^2\mathsf d_*^2}{\varepsilon^2}
            \left(q\log n+\log\frac{1}{\delta}\right)
            \right).
        \end{equation*}
        Since $b_*\le m_*$ and $\mathsf d_*=\mathcal O_k(m_*^{k-1})$, this is at most $\widetilde{\mathcal O}_{k,q}(\bar\alpha^2t^2m_*^{2k} \varepsilon^{-2}(q\log n+\log(1/\delta)))$.
    \end{corollary}
    For fixed $t$, $\bar\alpha$, $\varepsilon$, and $\delta$, polylogarithmic-size outer regions give polylogarithmic quantum sample complexity.

    \begin{proof}
        For each $A$, run the local learning procedure of Theorem~\ref{thm:local-observable-verification} with target accuracy $\varepsilon_{\rm loc}$ and failure probability $\delta/N_q$, using the same full-system data. Its generator-level event controls all Pauli observables on $A$ simultaneously. A union bound therefore gives, for every $A\in\mathcal A_q$ and every nonidentity $\mathbf P\in\mathcal P_A$,
        \begin{equation*}
            \left|
            \Tr\!\left[
            (\mathbf P\otimes I_{[n]\setminus A})e^{t\cL}(\rho_0)
            \right]
            -
            \Tr\!\left[
            (\mathbf P\otimes I_{B_A^{\operatorname{in}}\setminus A})
            e^{t\widehat\cL_{B_A^{\operatorname{in}}}}
            (\rho_{0,B_A^{\operatorname{in}}})
            \right]
            \right|
            \le\varepsilon_{\rm loc}.
        \end{equation*}
        The identity-observable difference vanishes. Expanding any Hermitian $O_A$ with $\|O_A\|_\infty\le1$ in the Pauli basis and applying Cauchy--Schwarz and Pauli orthogonality bounds its expectation-value difference by $2^{|A|}\varepsilon_{\rm loc}\le\varepsilon$. The sample bound follows from Theorem~\ref{thm:local-observable-verification}, $|B_A^{\operatorname{out}}|\le m_*$, and $\log(N_q/\delta)=\mathcal O_q(q\log n+\log(1/\delta))$.
    \end{proof}

\section{Independent contemporaneous and subsequent work}
    \label{sec:concurrent-work}

    The initial versions of this work and the independent papers by Sinha~\cite{Sinha2026}, Romanov et al.~\cite{Yelin.2026}, Arad et al.~\cite{Itai.2026}, and Lewis--Tang--Wright~\cite{LewisTangWright2026} appeared within a two-week period. Zhou--Gong, Song et al., and Chen--Yu appeared subsequently~\cite{ZhouGong2026,SongEtAl2026,ChenYu2026}.\footnote{We compare arXiv v2 of Arad et al., which appeared roughly a month after v1.} Our original submission already included support-free and guarded coefficient recovery, physical reconstruction, and a diamond-norm guarantee. This revision extends the inversion framework with the sitewise-tail guarantee of Theorem~\ref{thm:approximate-degree-learning}, using Lemma~\ref{lem:threshold-before-inversion}.

    We translate the main guarantees into common notation to the best of our understanding, referring to the cited statements for their full qualifications. We report experimental repetitions $N$ and total evolution time $T_{\rm tot}$ separately, with failure probability denoted by $\delta$. A repetition may contain several evolution segments in controlled protocols. The notation $\mathcal O_k$ hides constants depending on fixed locality $k$, and $\widetilde{\mathcal O}$ suppresses polylogarithmic factors. We use the Hamiltonian and dissipative coefficients $h,G$ of \eqref{eq:pauli-gksl-form}, combined as $x=(h,G)$ and $\chi$ defined by $\cL(X)=\sum_{\mathbf P,\mathbf Q}\chi_{\mathbf P,\mathbf Q}\mathbf P X\mathbf Q$. As in Section~\ref{sec:framework-and-definitions}, $\|\cdot\|_\infty$ denotes the operator norm for Hilbert-space operators and the largest absolute entry for coefficient arrays. Superoperator norms $\|\cdot\|_{p\to p}$ are induced by the Schatten $p$-norm.

    \subsection{Unknown local generators}
        The compared protocols use nonadaptive product-Pauli preparations and measurements without ancillas or intermediate control. Write $S_{\mathbf P}=\supp(\mathbf P)$ for a Hamiltonian coordinate and $S_{(\mathbf P,\mathbf Q)}=\supp(\mathbf P)\cup\supp(\mathbf Q)$ for a dissipative coordinate. Coefficient locality means $|S_v|\le k$ for every nonzero coefficient. Local positivity additionally requires a decomposition into $k$-local GKSL generators. For any coefficient array $z$, recall
        \begin{equation*}
            \|z\|_{\mathrm{loc},1}
            =\max_{u\in[n]}\sum_{v:\,u\in S_v}|z_v|,
            \qquad
            \deg(z)=\max_{u\in[n]}|\{v:u\in S_v,\ z_v\ne0\}|.
        \end{equation*}

        \paragraph{Results and assumptions.}
        Theorem~\ref{thm:approximate-degree-learning} assumes a local-positive decomposition $\cL=\sum_e\cL_e$ and receives $\bar\alpha>0$ with $\max_{u\in[n]}\sum_{e\ni u}\|\cL_e^\dagger\|_{2\to2}\le\bar\alpha$, together with a number $r\ge1$ such that, for some unknown $\tau>0$,
        \begin{equation}\label{eq:concurrent-our-promise}
            r_\tau\le r,
            \qquad
            \beta_\tau\le\frac{\varepsilon}{2\cdot4^k}.
        \end{equation}
        Neither the decomposition nor $\tau$ is supplied. Here $\mathcal I_k^\circ$ consists of Pauli pairs $v=(\mathbf P,\mathbf Q)\ne(I,I)$ with union support $S_v$ of size at most $k$. The clipped local sparsity and sitewise subthreshold tail are
        \begin{equation*}
            \begin{aligned}
                r_\tau&=\max\left\{1,\max_{u\in[n]}\sum_{v:\,u\in S_v}
                \min\left(1,\frac{|\chi_v|^2}{\tau^2}\right)\right\},\qquad
                \beta_\tau&=\max_{u\in[n]}\sum_{\substack{v:\,u\in S_v\\|\chi_v|\le\tau}}|\chi_v|.
            \end{aligned}
        \end{equation*}
        It guarantees $\|(2(\widehat h-h),\widehat G-G)\|_{\mathrm{loc},1}\le\varepsilon$ and estimates of degree $\mathcal O_k(r)$, using
        \begin{equation*}
            N=\widetilde{\mathcal O}_k(\bar\alpha^2r^2\varepsilon^{-2}\log(n/\delta)),
            \qquad
            T_{\rm tot}=\widetilde{\mathcal O}_k(\bar\alpha r^2\varepsilon^{-2}\log(n/\delta)).
        \end{equation*}
        The coefficient-local branch of Theorem~\ref{thm:approximate-degree-learning} instead accepts an upper bound on $\|\chi\|_{\mathrm{loc},1}$, without positive local GKSL blocks. A supplied $g>0$ with $\|(2h,G)\|_{\mathrm{loc},1}\le g$ gives $\|\chi\|_{\mathrm{loc},1}\le2g$ by \eqref{eq:physical-chi-map-bounds}, so $\bar\alpha$ can be replaced by $g$ up to constants.

        Theorem~\ref{thm:full-lindblad-learning} further gives a valid coefficient-local reconstruction with diamond error $\varepsilon_\diamond$ and $N=\widetilde{\mathcal O}_k(g^2r_\diamond^2n^2\varepsilon_\diamond^{-2}\log(n/\delta))$, assuming a supplied $r_\diamond\ge1$ satisfies $r_\tau\le r_\diamond$ and $\beta_\tau\le\varepsilon_\diamond/(16n4^k)$ for some unknown $\tau>0$. Its output has $\widehat G\succeq0$, whereas the local-positive branch additionally returns a decomposition into positive local GKSL blocks.

        Arad et al.~\cite[Thm.~3.11, v2]{Itai.2026} assume coefficient locality and the supplied local dynamical-strength bound
        \begin{equation*}
            \max_{u\in[n]}\sum_{S\ni u}
            \|\cL_S^\dagger\|_{\infty\to\infty}
            \le\Lambda_{\rm loc},
        \end{equation*}
        where $\cL_S$ groups the terms of \eqref{eq:pauli-gksl-form} with union support exactly $S$. These groups need not individually be GKSL generators. Let $\mathfrak d_{\rm dis}$ be a supplied upper bound on the maximum number of nonzero dissipative entries incident on any site. Their recovery factor is
        \begin{equation*}
            K_{\rm rec}=\min\{K_{\rm deg},K_{\rm den}\},
            \qquad
            K_{\rm deg}=2\sum_{j=0}^{k-1}(2\mathfrak d_{\rm dis})^j,
            \qquad
            K_{\rm den}=\mathcal O_k(n^{k-1}).
        \end{equation*}
        They achieve entrywise error $\varepsilon$ using
        \begin{equation*}
            N=\widetilde{\mathcal O}_k(\Lambda_{\rm loc}^2K_{\rm rec}^2\varepsilon^{-2}\log(n/\delta)),
            \qquad
            T_{\rm tot}=\widetilde{\mathcal O}_k(\Lambda_{\rm loc}K_{\rm rec}^2\varepsilon^{-2}\log(n/\delta)).
        \end{equation*}
        No gap or tail condition is required. Bounded dissipative degree gives logarithmic dependence on $n$, even with dense Hamiltonian interactions, while the generic sample factor is $n^{2k-2}$. Their lower bounds match the strength and precision dependence for bounded nonzero dissipative degree~\cite[Thms.~4.1--4.2]{Itai.2026}.

        Lewis--Tang--Wright~\cite[Thm.~5.1]{LewisTangWright2026} receive a strength bound $g>0$ and an approximate-degree bound $d\ge1$ satisfying $\|(2h,G)\|_{\mathrm{loc},1}\le g$ and $d\ge d_{\varepsilon/(100\cdot16^k)}(x)$. Their approximate degree~\cite[Def.~2.5]{LewisTangWright2026} is
        \begin{equation*}
            d_\eta(x):=
            \min_{\substack{z=(z^h,z^G):\\
            \|(2(h-z^h),G-z^G)\|_{\mathrm{loc},1}<\eta}}\deg(z).
        \end{equation*}
        Here $z$ ranges over the same $k$-local coordinates and need not be physical. Their local $B_1$ norm is $\|(2h,G)\|_{\mathrm{loc},1}$. They guarantee $\|(2(\widehat h-h),\widehat G-G)\|_{\mathrm{loc},1}\le\varepsilon$ without supplied locations or a coefficient gap, using
        \begin{equation*}
            N=\mathcal O_k(g^2d^2\varepsilon^{-2}\log(n/\delta)),
            \qquad
            T_{\rm tot}=\mathcal O_k(gd^2\varepsilon^{-2}\log(n/\delta)).
        \end{equation*}

        \paragraph{Comparison.}
        All three learn unknown locations from numerical promises. Arad et al.'s degree-dependent guarantee controls entrywise error under exact dissipative sparsity. The two local-norm guarantees instead use approximate sparsity, related as follows.

        \begin{rmk}[Comparison with local $B_1$ approximate degree]
            \label{rmk:b1-approximate-degree-comparison}
            The Lewis--Tang--Wright promise implies \eqref{eq:concurrent-our-promise} with $r=5d$. Choose $z$ with $\deg(z)\le d$ and $b:=\|(2(h-z^h),G-z^G)\|_{\mathrm{loc},1}<\varepsilon/(100\cdot16^k)$. Let $F$ be the linear map from the generator coefficients $x=(h,G)$ to the nonidentity $\chi$-coefficients on $\mathcal I_k^\circ$. By~\eqref{eq:physical-chi-map-bounds}, $\deg(Fz)\le3d$ and $\|\chi-Fz\|_{\mathrm{loc},1}=\|F(x-z)\|_{\mathrm{loc},1}\le2\|x-z\|_{\mathrm{loc},1}\le2b$. If $b>0$, separating indices where $(Fz)_v\ne0$ from those where $(Fz)_v=0$ and using $\min\{1,t^2\}\le t$ gives, at $\tau=b/d$,
            \begin{equation*}
                r_\tau\le3d+\frac{2b}{\tau}=5d,
                \qquad
                \beta_\tau\le3d\tau+2b=5b
                <\frac{\varepsilon}{2\cdot4^k}.
            \end{equation*}
            For $b=0$, take $\tau$ below the absolute values of nonzero $\chi$-coefficients, or arbitrary if $\chi=0$. With $\|\chi\|_{\mathrm{loc},1}\le2g$, this proves inclusion in the coefficient-local branch of Theorem~\ref{thm:approximate-degree-learning}, without supplied thresholds or locations.
        \end{rmk}

        The previous version of our work established no direct implication in either direction between our guarded-tail promise and their approximate-degree promise.

        Both results therefore give support-free local-norm recovery and guard-band support recovery by thresholding (Cor.~\ref{cor:approximate-degree-thresholding}), with sample complexity $\widetilde{\mathcal O}_k(g^2d^2\varepsilon^{-2}\log(n/\delta))$. Exact support recovery additionally requires a coefficient gap. Lewis--Tang--Wright use one positive time $\Theta_k(1/g)$, while our interpolation incurs additional polylogarithmic factors. The conversion also covers their fixed-$k$ spatial-decay examples~\cite[Cors.~5.8--5.9]{LewisTangWright2026}. Our coefficient scan costs $\mathcal O_k(n^k)$, with additional polynomial costs for data processing, interpolation, and physical projection. Their classical bound can be quasipolynomial in $d$~\cite[Thm.~5.1]{LewisTangWright2026}.

        Without sparsity promises, Corollary~\ref{cor:generic-local-norm-recovery} and Theorem~\ref{thm:full-lindblad-learning} give coefficient and diamond sample factors $n^{2k-2}$ and $n^{2k}$, improving $n^{2k}$ and $n^{4k}$ in Theorems 4.5 and 5.3 of our arXiv v1.
        
    \subsection{Unknown nonlocal generators}

        Sparse nonlocal learning uses global rather than sitewise sparsity. Write
        \begin{equation*}
            S_H:=\{\mathbf P:h_{\mathbf P}\ne0\},
            \qquad
            S_D:=\{\mathbf P:G_{\mathbf P,\mathbf P}>0\},
            \qquad M\ge|S_H|+|S_D|^2.
        \end{equation*}
        By positivity, $S_D$ indexes the nonzero rows of $G$, so $M$ bounds the enclosing coefficient count, not just nonzero entries. Ivashkov, Romanov, and Zhou--Gong assume $|h_{\mathbf P}|,|G_{\mathbf P,\mathbf Q}|\le1$.

        \paragraph{Results and assumptions.}
        Ivashkov et al.~\cite[Thm.~D.10]{IvashkovEtAl2026} use product-Pauli data and a supplied lower bound $\eta$ on the absolute values of nonzero coefficients to recover $S_D$ and a Hamiltonian candidate superset. Their coefficient stage~\cite[Thm.~E.10]{IvashkovEtAl2026} depends on $\nu_C=\|C^{-1}\|_{\infty\to\infty}$, the maximum absolute row sum of the inverse of the selected recovery matrix $C$. Combining both stages and the serialized alternative~\cite[Thm.~D.10, Eq.~(D52), and Rmk.~E.12]{IvashkovEtAl2026} gives entrywise error $\varepsilon$ with
        \begin{equation*}
            \begin{aligned}
                N&=\widetilde{\mathcal O}\bigl(M^4\eta^{-4}
                +q_{\rm I}d_{\rm I}^2\nu_C^2\varepsilon^{-2}\bigr),\\
                T_{\rm tot}&=\widetilde{\mathcal O}\bigl(M^3\eta^{-4}
                +q_{\rm I}d_{\rm I}\nu_C^2\varepsilon^{-2}\bigr),
                \qquad
                q_{\rm I}:=\min\{9^{k_H}+9^{k_D},\widehat M\}.
            \end{aligned}
        \end{equation*}
        Here $k_H$ is the largest Hamiltonian candidate support size, $k_D$ the largest union-support size of two dissipative candidates, and $\widehat M$ the number of Hamiltonian candidates plus all ordered pairs of dissipative candidates. The bound $d_{\rm I}\ge1$ concerns overlap degree of nonzero Pauli--GKSL summands, not sitewise degree.

        Romanov et al.~\cite[Result 2, Cor.~E.11 and Thm.~F.5]{Yelin.2026} learn all coefficients to entrywise error $\varepsilon$. Combining the experimental counts of their stages gives
        \begin{equation*}
            N=\widetilde{\mathcal O}(M^3+M^2\varepsilon^{-2}),
            \qquad
            T_{\rm tot}
            =\widetilde{\mathcal O}(M\varepsilon^{-2}
            +M^2\varepsilon^{-1}).
        \end{equation*}
        They use stabilizer-code error correction, noiseless ancillas/control, and Bell/Clifford operations. Their $N$ counts trajectories, excluding QEC cycles and unsuccessful code-state preparations~\cite[Thms.~D.5, E.3 and E.7]{Yelin.2026}. Under additional tail regularity, they achieve Heisenberg-limited Hamiltonian learning outside an accuracy-dependent dissipative subspace~\cite[Assump.~E.12 and Cor.~E.14]{Yelin.2026}.

        Sinha~\cite[Thm.~3.3.1]{Sinha2026} assumes $|\chi_{\mathbf P,\mathbf Q}|\le1$ and at most $s_\chi$ distinct Pauli labels in nonzero rows or columns of $\chi$, including the identity. Bell sampling, Clifford operations, and postselection give entrywise error $\varepsilon$ with polynomial classical processing. Algorithms 2--6 use
        \begin{equation*}
            N=\widetilde{\mathcal O}(s_\chi^8\varepsilon^{-6}),
            \qquad
            T_{\rm tot}=\widetilde{\mathcal O}(s_\chi^4\varepsilon^{-4}).
        \end{equation*}
        A state-preparation-and-measurement (SPAM) extension~\cite[Thm.~4.3.1]{Sinha2026} allows unknown preparation and measurement channels with summed diamond distance from the identity bounded by a constant below one. It recovers gauge-independent coefficients, invariant under indistinguishable SPAM/generator transformations, with exponential classical cost.

        Zhou--Gong~\cite[Thm.~4]{ZhouGong2026} use product-Pauli access for unknown nonlocal generators. Given $M$ and $\Gamma\ge2\|H\|_{\infty}+2\sum_{\mathbf P,\mathbf Q}|G_{\mathbf P,\mathbf Q}|$, with $\|H\|_{\infty}$ the operator norm, they achieve entrywise error $\varepsilon$ using
        \begin{equation*}
            N=\widetilde{\mathcal O}
            (\Gamma^2M^2\varepsilon^{-4}),
            \qquad
            T_{\rm tot}=\widetilde{\mathcal O}
            (\Gamma M^2\varepsilon^{-2}).
        \end{equation*}
        Their calibrated product Pauli-diagonal SPAM extension~\cite[Cor.~1 and Lem.~20]{ZhouGong2026} requires candidate coverage of heavy dissipative rows and nondegenerate calibration bounds, which determine its overhead.

        Chen--Yu~\cite[Cor.~4.7]{ChenYu2026} require neither locality nor sparsity, but receive $\Lambda_{\rm glob}\ge\|\cL^\dagger\|_{\infty\to\infty}$. Product-Pauli experiments discover candidates, while global stabilizer preparations and Clifford-basis measurements estimate their coefficients to entrywise error $\varepsilon$, with
        \begin{equation*}
            N=\widetilde{\mathcal O}
            (\Lambda_{\rm glob}^2\varepsilon^{-2}),
            \qquad
            T_{\rm tot}=\widetilde{\mathcal O}
            (\Lambda_{\rm glob}\varepsilon^{-2}),
        \end{equation*}
        and polynomial classical processing.

        \paragraph{Comparison.}
        These results allow nonlocal interactions but also apply to local generators, where their assumptions overlap with ours. They use global sparsity or strength bounds, whereas our guarantees exploit fixed locality and local strength. Global strength can grow with $n$, so explicit system-size factors alone do not compare resource costs. Ivashkov and Zhou--Gong share our product-Pauli access, whereas Romanov, Sinha, and Chen--Yu use additional coherent primitives. Their main guarantees above are entrywise, unlike our sitewise coefficient and diamond-error guarantees (Thms.~\ref{thm:approximate-degree-learning} and~\ref{thm:full-lindblad-learning}).

        Our local inversion has explicit stability bounds (Lem.~\ref{lem:threshold-before-inversion} and Cor.~\ref{thm:dissipative-learning}), whereas those of Ivashkov et al. depend on $\nu_C$. Proposition~\ref{prop:rank-completion-conditioning} shows that full rank alone does not bound this amplification. SPAM robustness concerns preparation and measurement errors, unlike the generator misspecification of Proposition~\ref{prop:agnostic-learning}.

    \subsection{Supplied candidate models}

        \paragraph{Results and assumptions.}
        Song et al.~\cite[Thm.~7]{SongEtAl2026} receive $M$ independent real Pauli--GKSL directions $\cL_j$ with $\cL=\sum_{j=1}^M\theta_j\cL_j$, $\theta_j\in[-1,1]$, and bounded dissipative union-support size. $M$ counts supplied directions, not a sparsity budget. They assume either bounded dictionary overlap degree, counting other candidate supports intersecting each candidate, or $s\kappa_0=\mathcal O(1)$, where
        \begin{equation*}
            s:=\max_j|\supp(\cL_j^\dagger)|,
            \qquad
            \kappa_0:=\max_{u\in[n]}
            \sum_{j:\,u\in\supp(\cL_j^\dagger)}
            \|\cL_j^\dagger\|_{\infty\to\infty}.
        \end{equation*}
        Thus $\kappa_0$ is unweighted dictionary strength, and Hamiltonian supports may exceed the dissipative bound. Their product-Pauli algorithms~\cite[Props.~6--7]{SongEtAl2026} achieve entrywise error $\varepsilon$ with $N=\widetilde{\mathcal O}(M\varepsilon^{-2})$ in the uniform bounded-overlap regime, using logarithmically many times or one fixed time. Bounded evolution times also give $T_{\rm tot}=\widetilde{\mathcal O}(M\varepsilon^{-2})$. Their optional physical projection gives entrywise error $2\varepsilon$~\cite[Cor.~11]{SongEtAl2026}, with SDP cost additional to response processing.

        \paragraph{Comparison.}
        Candidate locations are input, as in our supplied-candidate Corollary~\ref{cor:dissipative-learning-stable-sparse}, unlike the numerical promises of our support-free Theorem~\ref{thm:approximate-degree-learning}. Both provide physical projection with entrywise guarantees. Our Theorem~\ref{thm:full-lindblad-learning} additionally certifies diamond error.
\section*{Acknowledgments}
    CR would like to thank Peter Brown for helpful discussions. D.S.F. acknowledges financial support from the Novo Nordisk Foundation (Grant No. NNF20OC0059939 Quantum for Life) and by the ERC grant GIFNEQ 101163938. T.M. acknowledges support from the Deutsche Forschungsgemeinschaft (DFG, German Research Foundation), Project-ID 470903074, TRR 352, and Project-ID 575156903. This project was funded within the QuantERA II program, which has received funding from the EU's H2020 research and innovation program under GA No. 101017733. DSF and CR are supported by France 2030 under the French National Research Agency award number ``ANR-22-PNCQ-0002''.

\paragraph{Use of AI tools.}
    The underlying ideas and original research approach were developed by the authors. Extensions and a consistent notation were subsequently developed through interaction with ChatGPT and Codex (OpenAI, models 5~sol and 6~astra), which also assisted with language editing, literature comparisons, and checking mathematical arguments. The authors take responsibility for the final content.

\sloppy
\bibliographystyle{unsrturl}
\bibliography{refs_v6}

\appendix
\addtocontents{toc}{\protect\setcounter{tocdepth}{1}}
\addcontentsline{toc}{section}{Appendices}
\addtocontents{toc}{\protect\setcounter{tocdepth}{0}}

\section{A conditioning caveat for rank-only probe selection}\label{app:rank-selection-conditioning}
    This appendix gives a compact counterexample showing that full rank and fixed locality alone do not control the pseudoinverse factor
    \begin{equation*}
        \nu_C:=\|C^+\|_{\infty\to\infty}
    \end{equation*}
    appearing in the coefficient-learning stage of \cite{IvashkovEtAl2026}.
    \begin{proposition}[Rank-only completion can be ill-conditioned]
        \label{prop:rank-completion-conditioning}
        For every $m\ge1$, there is a $4$-local diagonal Pauli dissipator on $2m$ qubits, with support-intersection degree at most six and weighted interaction strength at most six, together with $m$ independent admissible patchwise Pauli probes with the following property.  Every square full-rank completion of these probes from the full patchwise row pool for the dissipative coefficient matrix on the candidate set has design matrix $C_{\rm full}$ satisfying
        \begin{equation*}
            \|C_{\rm full}^{-1}\|_{\infty\to\infty}
            \ge c\,\lambda_\star^m ,
        \end{equation*}
        where $c>0$ is universal and $\lambda_\star\approx1.46557$ is the largest real root of $\lambda^3=\lambda^2+1$.  The candidate dissipative patches have size at most eight.
    \end{proposition}
    \begin{proof}
        Use shared qubits $s_1,\ldots,s_m$ and private marker qubits $p_1,\ldots,p_m$.  Define
        \begin{equation*}
            \mathbf P_j:=Z_{p_j}Z_{s_j}Z_{s_{j-1}}Z_{s_{j-3}},
            \qquad j=1,\ldots,m,
        \end{equation*}
        omitting factors with indices outside $\{1,\ldots,m\}$.  Each $\mathbf P_j$ is $4$-local, each shared qubit is contained in at most three strings, and each private marker appears in only one.  Moreover, $\mathbf P_j$ can intersect $\mathbf P_\ell$ only if $|j-\ell|\in\{1,2,3\}$, so the support-intersection degree is at most six. Consider the diagonal dissipator
        \begin{equation*}
            \cL(\rho)=\sum_{j=1}^m
            \bigl(\mathbf P_j\rho\mathbf P_j-\rho\bigr).
        \end{equation*}
        For the local summand $\cL_j$ one has $\|\cL_j^\dagger\|_{2\to2}\le2$.  Since at most three summands touch any qubit, the weighted interaction strength in \eqref{eq:weightedintersectiondegree} obeys $\alpha\le6$. Let $\widehat S_D=\{\mathbf P_1,\ldots,\mathbf P_m\}$ be the candidate dissipative Pauli support.  Although the true diagonal generator is $4$-local, the full candidate matrix also contains coefficient positions $(\mathbf P_j,\mathbf P_\ell)$. Hence its patch locality in the notation of \cite{IvashkovEtAl2026} satisfies
        \begin{equation*}
            k_D
            =
            \max_{j,\ell}
            |\supp(\mathbf P_j)\cup\supp(\mathbf P_\ell)|
            \le8.
        \end{equation*}
        To describe the relevant columns of the design matrix, temporarily write
        \begin{equation*}
            \cL_a(\rho)
            :=
            \sum_{j=1}^m a_j
            \bigl(\mathbf P_j\rho\mathbf P_j-\rho\bigr).
        \end{equation*}
        For a probe with $\mathbf Q=\mathbf O$, its diagonal part is
        \begin{equation*}
            2^{-2m}\Tr[\mathbf O\cL_a^\dagger(\mathbf O)]
            =
            -2\sum_{\{\mathbf O,\mathbf P_j\}=0}a_j .
        \end{equation*}
        Choose $\mathbf O_i=\mathbf Q_i=X_{s_i}$.  This probe is admissible because it is supported inside the patch $\supp(\mathbf P_i)\cup\supp(\mathbf P_i)$.  It anticommutes exactly with $\mathbf P_i,\mathbf P_{i+1},\mathbf P_{i+3}$, ignoring nonexistent indices.  The corresponding diagonal block is therefore
        \begin{equation*}
            C_{\rm bad}=-2B_m,\qquad B_m=I+S+S^3,
        \end{equation*}
        where $S$ is the upper shift matrix.  In particular, $B_m$ is unit upper triangular and hence invertible. Let $(b_r)_{r\ge0}$ be the coefficients of the formal power series $(1+x+x^3)^{-1}$.  Since $S^m=0$, $B_m^{-1}=\sum_{r=0}^{m-1}b_rS^r$, and the defining power-series identity gives
        \begin{equation*}
            b_0=1,\quad b_1=-1,\quad b_2=1,\quad
            b_r=-b_{r-1}-b_{r-3}\quad(r\ge3).
        \end{equation*}
        Inductively, $b_r=(-1)^r c_r$, where
        \begin{equation*}
            c_0=c_1=c_2=1,\qquad
            c_r=c_{r-1}+c_{r-3}\quad(r\ge3).
        \end{equation*}
        The companion matrix of this positive recurrence is primitive and has Perron root $\lambda_\star$, so $c_r=\Theta(\lambda_\star^r)$.  The first row of $C_{\rm bad}^{-1}$ contains all coefficients $-\frac12b_0,\ldots,-\frac12b_{m-1}$, while every subsequent row contains an initial segment of this list preceded by zeros.  Consequently,
        \begin{equation*}
            \|C_{\rm bad}^{-1}\|_{\infty\to\infty}
            =
            \frac12\sum_{r=0}^{m-1}c_r
            =
            \Theta(\lambda_\star^m).
        \end{equation*}
        It remains to complete these diagonal rows to a system identifying the entire Hermitian dissipative coefficient matrix.  If $j\ne\ell$, every contribution of the coefficient $a_{j\ell}$ to a selected row above is, up to a sign, proportional to $2^{-2m}\Tr(\mathbf P_j\mathbf P_\ell)$.  The private markers make $\mathbf P_j\mathbf P_\ell$ a nonidentity Pauli, so this trace vanishes. Thus the selected rows have zero coefficients on all off-diagonal Hermitian variables. By \cite[Lem.~E.1 and Cor.~E.2]{IvashkovEtAl2026}, the full patchwise row pool has full column rank for the complete dissipative coefficient matrix on $\widehat S_D$.  Since the $m$ rows above are independent, the basis-extension theorem allows us to append rows from this pool to obtain a square full-rank matrix.  After ordering the unknowns as $(a_{11},\ldots,a_{mm},z_{\rm off})$, every such completion that includes the selected rows has block form
        \begin{equation*}
            C_{\rm full}=
            \begin{pmatrix}
                C_{\rm bad} & 0\\
                A & D
            \end{pmatrix}.
        \end{equation*}
        Full rank implies that $D$ is invertible, and block inversion gives
        \begin{equation*}
            C_{\rm full}^{-1}
            =
            \begin{pmatrix}
                C_{\rm bad}^{-1} & 0\\
                -D^{-1}AC_{\rm bad}^{-1} & D^{-1}
            \end{pmatrix}.
        \end{equation*}
        Therefore
        \begin{equation*}
            \|C_{\rm full}^{-1}\|_{\infty\to\infty}
            \ge
            \|C_{\rm bad}^{-1}\|_{\infty\to\infty}
            \ge c\,\lambda_\star^m
        \end{equation*}
        for a universal $c>0$.
    \end{proof}
    The proposition is an obstruction to deriving a uniform bound on $\nu_C$ merely from full rank, fixed locality, or bounded degree: a probe-selection criterion that enforces only rank permits exponentially ill-conditioned admissible completions.  It does not assert that every full-rank patchwise design, an optimally conditioned design, or a particular implementation of the selection step must be ill-conditioned.  Our reconstruction instead uses fixed local Pauli/Fierz inversion formulas, so its stability losses are explicit local-inversion and diagonal-extension factors rather than a data-dependent global pseudoinverse norm.

\section{Technical details for the local coefficient $1$-norm physical projection}
    \label{app:sdp-projection-details}
    Let
    \begin{equation*}
        \mathcal V_H:=\mathcal P_{n,\le k}\setminus\{I\},
        \qquad
        \mathcal V_G:=
        \left\{(\mathbf P,\mathbf Q):
        \mathbf P,\mathbf Q\ne I,\
        |\supp(\mathbf P)\cup\supp(\mathbf Q)|\le k
        \right\},
    \end{equation*}
    and regard their disjoint union as the coordinate set of $x=(h,G)$. For local matrices
    \begin{equation*}
        X=(X^e)_{e\in\mathcal R_{n,\le k}},
        \qquad
        X^e\in\mathbb C^{\cP_e^\circ\times\cP_e^\circ},
    \end{equation*}
    define
    \begin{equation}\label{eq:local-block-embedding}
        \cA(X):=\sum_{e\in\mathcal R_{n,\le k}}\iota_e(X^e),
    \end{equation}
    where
    \begin{equation*}
        \iota_e(X^e)_{\mathbf P,\mathbf Q}
        :=
        \begin{cases}
            (X^e)_{\mathbf P_e,\mathbf Q_e},
            &\supp(\mathbf P)\cup\supp(\mathbf Q)\subseteq e,\\
            0,&\text{otherwise}.
        \end{cases}
    \end{equation*}

    We first verify the sitewise trace constraints in \eqref{eq:local-norm-physical-projection}.  For a local contribution $\cL_e$, its PTM entries satisfy
    \begin{equation*}
        |L^e_{\mathbf A,\mathbf B}|
        =
        \frac1{2^{|e|}}
        \left|\Tr\!\left[\mathbf A\cL_e(\mathbf B)\right]\right|
        \le\|\cL_e^\dagger\|_{2\to2}.
    \end{equation*}
    The local Fierz formula then gives, for every $\mathbf P,\mathbf Q\in\cP_e^\circ$,
    \begin{equation*}
        |(X^e_{\mathrm{true}})_{\mathbf P,\mathbf Q}|
        \le\|\cL_e^\dagger\|_{2\to2}.
    \end{equation*}
    Since $X^e_{\mathrm{true}}\succeq0$ and $|\cP_e^\circ|=4^{|e|}-1\le r_k$, it follows that
    \begin{equation}\label{eq:local-block-trace-bound}
        \sum_{e\ni j}\Tr(X^e_{\mathrm{true}})
        \le
        r_k\sum_{e\ni j}\|\cL_e^\dagger\|_{2\to2}
        \le r_k\bar\alpha.
    \end{equation}
    Thus the local blocks of the true decomposition are feasible.

    For completeness, an explicit epigraph formulation of \eqref{eq:local-norm-physical-projection} is as follows.  For $v\in\mathcal V_H\sqcup\mathcal V_G$, define
    \begin{equation*}
        z_v(\widetilde h,X):=
        \begin{cases}
            2(\widetilde h_{\mathbf P}-\widehat h_{\mathbf P}),
            &v=\mathbf P\in\mathcal V_H,\\
            \cA(X)_{\mathbf P,\mathbf Q}
            -\widehat G_{\mathbf P,\mathbf Q},
            &v=(\mathbf P,\mathbf Q)\in\mathcal V_G.
        \end{cases}
    \end{equation*}
    Then the program is equivalently
    \begin{equation}\label{eq:local-norm-projection-epigraph}
        \begin{aligned}
            \min_{\widetilde h,X,u,\eta}\quad&\eta\\
            \textnormal{subject to}\quad
            &X^e\succeq0
            &&(e\in\mathcal R_{n,\le k}),\\
            &\widetilde h\in\mathbb R^{\mathcal V_H},\\
            &\sum_{e\ni j}\Tr(X^e)\le r_k\bar\alpha
            &&(j\in[n]),\\
            &\left\|(\Re z_v,\Im z_v)\right\|_2\le u_v
            &&(v\in\mathcal V_H\sqcup\mathcal V_G),\\
            &\sum_{v:\,j\in S_v}u_v\le\eta
            &&(j\in[n]).
        \end{aligned}
    \end{equation}
    The Hamiltonian residuals include their factor two. The Lorentz-cone constraints in the fourth line encode the absolute values of the complex residuals, so the objective equals the norm of the residual pair in \eqref{eq:local-norm-physical-projection}. Equation~\eqref{eq:local-block-trace-bound} and the true Hamiltonian coefficients show that the true $x$ is feasible. Hence, if $\|(2(\widehat h-h),\widehat G-G)\|_{\mathrm{loc},1}\le\varepsilon$, an additive-$\xi$ solution $x^{\mathrm{phys}}=(h^{\mathrm{phys}},G^{\mathrm{phys}})$ satisfies
    \begin{equation*}
        \begin{aligned}
            \|(2(h^{\mathrm{phys}}-\widehat h),G^{\mathrm{phys}}-\widehat G)\|_{\mathrm{loc},1}
            &\le\varepsilon+\xi,\\
            \|(2(h^{\mathrm{phys}}-h),G^{\mathrm{phys}}-G)\|_{\mathrm{loc},1}
            &\le2\varepsilon+\xi.
        \end{aligned}
    \end{equation*}

    We briefly justify the claimed classical complexity.  The program has $\mathcal O_k(n^k)$ real variables and affine constraints, $\mathcal O_k(n^k)$ positive-semidefinite blocks of order at most $2(4^k-1)$ after realification, and $\mathcal O_k(n^k)$ constant-dimensional Lorentz cones.  Its standard barrier parameter is therefore $\mathcal O_k(n^k)$.  To see strict feasibility, put
    \begin{equation*}
        D_*:=\max_{j\in[n]}\sum_{e\ni j}(4^{|e|}-1),
        \qquad
        \beta:=\frac{r_k\bar\alpha}{2D_*},
        \qquad
        X_0^e:=\beta I_e,
    \end{equation*}
    when $\bar\alpha>0$, take $\widetilde h_0=0$, choose each $u_v$ strictly larger than $|z_v(\widetilde h_0,X_0)|$, and then choose $\eta$ strictly larger than all sitewise sums.  The PSD and trace constraints have inverse-polynomial slack.  A computable polynomial upper bound on $\eta$, chosen large enough to contain this point and the true point using \eqref{eq:witness-strength-controls-coefficients}, makes the relevant sublevel set compact without changing the optimum.  If $\bar\alpha=0$, the zero generator is the only feasible target. Standard homogeneous self-dual interior-point methods~\cite{deKlerk2002SDP,NesterovNemirovski1994,VandenbergheBoyd1996SDP} apply.  A dense Newton step costs $\mathcal O_k(n^{4k})$ operations. Together with the $\mathcal O_k(n^k)$ barrier parameter, this gives the arithmetic cost
    \begin{equation*}
        \widetilde{\mathcal O}_k\!\left(
        n^{9k/2}\log\frac1\xi
        \right)
    \end{equation*}
    Here the normalized polynomial-input convention means unit-cost real arithmetic with $\bar\alpha$, $\bar\alpha^{-1}$, and $\|\widehat x\|_{\mathrm{loc},1}$ bounded by $n^{\mathcal O_k(1)}$ when $\bar\alpha>0$. The accuracy $0<\xi<1$ is measured in the same coefficient units.

    Finally, diagonalize each fitted block as
    \begin{equation*}
        X^e=\sum_s\ell_{e,s}\ell_{e,s}^\dagger,
        \qquad
        L_{e,s}:=\sum_{\mathbf P\in\cP_e^\circ}
        (\ell_{e,s})_{\mathbf P}\mathbf P.
    \end{equation*}
    Together with the fitted Hamiltonian coefficients $\widetilde h$, these local jump operators give a valid $k$-local GKSL generator.  Their construction costs $\mathcal O_k(n^k\log(1/\xi))$ additional operations and is dominated by solving the semidefinite program.

    \begin{proof}[Proof of Lemma~\ref{lem:coeff-to-diamond}]
        For every Pauli string $\mathbf P$,
        \begin{equation*}
            \left\|-i[\mathbf P,\bullet]\right\|_\diamond\le2.
        \end{equation*}
        Moreover, for
        \begin{equation*}
            \Phi_{\mathbf P,\mathbf Q}
            :=\mathbf P\bullet\mathbf Q
            -\frac12\{\mathbf Q\mathbf P,\bullet\},
        \end{equation*}
        unitary invariance and the triangle inequality give $\|\Phi_{\mathbf P,\mathbf Q}\|_\diamond\le2$.  Therefore
        \begin{equation*}
            \|\cL_x-\cL_{x'}\|_\diamond
            \le
            2\sum_{\mathbf P}|h_{\mathbf P}-h'_{\mathbf P}|
            +2\sum_{\mathbf P,\mathbf Q}
            |G_{\mathbf P,\mathbf Q}-G'_{\mathbf P,\mathbf Q}|.
        \end{equation*}
        Every physical coordinate has nonempty support, so
        \begin{equation*}
            \sum_v|x_v-x'_v|
            \le
            \sum_{j=1}^n\sum_{v:\,j\in S_v}|x_v-x'_v|
            \le n\|x-x'\|_{\mathrm{loc},1}.
        \end{equation*}
        This proves \eqref{eq:local-norm-to-diamond}.
    \end{proof}

    \begin{proof}[Proof of Theorem~\ref{thm:full-lindblad-learning}]
        Put
        \begin{equation*}
            \varepsilon_0:=\frac{\varepsilon_\diamond}{8n},
            \qquad
            \xi:=\frac{\varepsilon_\diamond}{4n}.
        \end{equation*}
        By Lemma~\ref{lem:coefficient-strength}, $\|x\|_{\mathrm{loc},1}\le c_k\bar\alpha$.  If $\varepsilon_\diamond\ge2nc_k\bar\alpha$, the zero generator already has diamond error at most $\varepsilon_\diamond$.  In the remaining regime, Corollary~\ref{cor:generic-local-norm-recovery} at accuracy $\varepsilon_0$ gives
        \begin{equation*}
            \|(2(\widehat h-h),\widehat G-G)\|_{\mathrm{loc},1}\le\varepsilon_0
        \end{equation*}
        with probability at least $1-\delta$.  Solving \eqref{eq:local-norm-projection-epigraph} to additive accuracy $\xi$ then yields
        \begin{equation*}
            \|(2(h^{\mathrm{phys}}-h),G^{\mathrm{phys}}-G)\|_{\mathrm{loc},1}
            \le2\varepsilon_0+\xi
            =\frac{\varepsilon_\diamond}{2n}.
        \end{equation*}
        Lemma~\ref{lem:coeff-to-diamond} proves the target diamond error. Substituting $\varepsilon_0=\Theta(\varepsilon_\diamond/n)$ into Corollary~\ref{cor:generic-local-norm-recovery} gives
        \begin{equation*}
            \widetilde{\mathcal O}_k\!\left(
            \bar\alpha^2n^{2k}\varepsilon_\diamond^{-2}
            \log\frac n\delta
            \right)
        \end{equation*}
        samples.

        Under the stated assumptions on $r_\diamond$ and $\beta_\tau$, the same argument applies with Theorem~\ref{thm:approximate-degree-learning} at accuracy $\varepsilon_0$ and supplied bound $r_\diamond$, proving \eqref{eq:approximate-degree-diamond-samples}.  For the broader coefficient-local model, use the coefficient-local branch of that theorem and replace the local blocks in \eqref{eq:local-norm-projection-epigraph} by one global constraint $G\succeq0$ and $\Tr(G)\le ng$, with forbidden nonlocal coefficient entries fixed to zero.  The same feasibility and triangle-inequality proof gives the stated $g^2r_\diamond^2n^2$ sample bound and a valid coefficient-local output.
    \end{proof}
\end{document}